\documentclass[reqno,a4paper,10pt]{amsart}
\usepackage[a4paper, centering]{geometry}
\usepackage[utf8]{inputenc} % allow utf-8 input
\usepackage[english]{babel}
\usepackage[T1]{fontenc}
\usepackage{lmodern}
\usepackage{caption}
\usepackage{subcaption}
\usepackage{enumerate}
\usepackage{amssymb, amsmath}
\usepackage{mathrsfs,dsfont}
\usepackage{amscd}
\usepackage{bbm}
\usepackage{amsthm}
\usepackage{cite}
\usepackage{esint }
\usepackage[mathcal]{euscript}
\usepackage[active]{srcltx}
\usepackage{verbatim}
\usepackage[colorlinks,linkcolor={blue},citecolor={blue},urlcolor={black}]{hyperref}
\usepackage[]{changebar}
\usepackage{xcolor}
\usepackage{mathtools}
\usepackage{url}            % simple URL typesetting
\usepackage{booktabs}       % professional-quality tables
\usepackage{amsfonts}       % blackboard math symbols
\usepackage{nicefrac}       % compact symbols for 1/2, etc.
\usepackage{microtype}      % microtypography
\usepackage{lipsum}
\usepackage{fancyhdr}       % header
\usepackage{graphicx}       % graphics
\graphicspath{{media/}}     % organize your images and other figures under media/ folder
\usepackage{bookmark} %(fa vedere le formule)
\usepackage[a4paper, centering]{geometry}
\usepackage{epstopdf}
\usepackage{amscd}
\usepackage{color}
\usepackage{graphicx}

\usepackage{yfonts}
\usepackage{fullpage}
\usepackage{comment}
\usepackage[braket,qm]{qcircuit} % quantum circuits
\usepackage{ stmaryrd } % per un commando che chiamo lket dopo

\usepackage[many]{tcolorbox}
\newtcolorbox{mybox}{enhanced,colback=red!5!white, colframe=red!75!black, width=\textwidth,box align=center,halign=center,valign=center, center}
\usepackage[new]{old-arrows} % large arrows
\usepackage{mathabx} % updownarrows

\usepackage{tabularx}
\usepackage{multirow}
\usepackage{makecell}

\newtheorem{thm}{Theorem}[section]

\newtheorem*{thm*}{Theorem}
\newtheorem{cor}{Corollary}[section]

\newtheorem{lem}{Lemma}[section]

\newtheorem*{prop*}{Proposition}
\newtheorem{ass}{A}

\theoremstyle{definition}
\newtheorem{defn}{Definition}[section]

\theoremstyle{remark}
\newtheorem{rem}{Remark}[section]

\numberwithin{equation}{section}
\def\N{{\mathbb N}}
\def\Z{{\mathbb Z}}
\def\R{{\mathbb R}}
\def\CC{{\mathbb C}}
\def\L{{\mathcal L}}
\def\Lg{{\mathscr{L}}}
\def\Lc{{\mathfrak L}}
\def\H{{\mathcal H}}
\def\O{{\mathcal O}}
\def\E{{\mathbb E}}
\def\dim{{\rm dim}}
\def\X{{\mathcal X}}

\def\D{{\mathcal D}}

\def\I{{\mathbbm{1}}}
\def\cP{{\mathscr{P}}}
\def\norma #1{\left\lVert #1 \right\rVert}

\def\P{{\mathbb{P}}}

\def\de{{\rm d}}

\def\F{{\mathscr{F}}}
\def\K{{\mathcal K}}

\def\lin{{\rm lin}}

\def\cE{{\mathcal{E}}}
\newcommand{\hol}[1]{\"{#1}}

\def\A{\mathscr{A}}
\def\U{\mathcal{U}}

\newcommand{\f}[1]{{\color{blue!85!black}#1}}
\newcommand{\id}{\mathbbm{1}}

\usepackage[dvipsnames]{xcolor}
\makeatletter
\def\cleardoublepage{\clearpage\if@twoside \ifodd\c@page\else
\hbox{}
\thispagestyle{empty}
\newpage
\if@twocolumn\hbox{}\newpage\fi\fi\fi}
\makeatother
\title[QML]{Lazy training of quantum physics informed neural networks}

\author[A.~Melchor Hernandez]{Anderson Melchor Hernandez\textsuperscript{*}}
\address[A.~Melchor Hernandez]{Department of Mathematics, University of Bologna, Piazza di Porta San Donato 5, 40126, Bologna (Italy)}
\email{anderson.melchor@unibo.it}
\thanks{\textsuperscript{*}Corresponding author: Name: Anderson; Surname: Melchor Hernandez; email: anderson.melchor@unibo.it}
\author[G.~De Palma]{Giacomo De Palma}
\address[G.~De Palma]{Department of Mathematics, University of Bologna, Piazza di Porta San Donato 5, 40126, Bologna (Italy)}
\email{giacomo.depalma@unibo.it}

\date{\today}

\begin{document}

\begin{abstract}
We study the gradient-flow training dynamics of quantum physics-informed neural networks (QPINNs) for the solution of second-order elliptic partial differential equations with Dirichlet boundary conditions. We consider parameterized quantum circuits as function approximators and analyze their overparameterized regime through the lens of the neural tangent kernel (NTK). Our contribution is a nonasymptotic lazy-training theory for QPINNs and their variational formulation: we prove that, for sufficiently large circuit width, the nonlinear gradient flow is quantitatively approximated by a linearized NTK model, with explicit bounds depending on the number of qubits, circuit depth, circuit light-cone geometry, and the dimension of the domain of the solution to the PDE.
\end{abstract}

\maketitle
\tableofcontents

\section{Introduction}\label{sec:1}
Solving partial differential equations (PDEs) is a central problem in applied mathematics and scientific computing \cite{quarteroni1994numerical,quarteroni2006numerical,ciarlet1990handbook}.
Classical numerical methods, such as finite difference, finite element, and Galerkin methods, rely on the explicit discretization of the spatial domain and on the projection of the solution onto finite-dimensional approximation spaces \cite{hildebrand1987introduction}.
While these approaches are well understood and highly effective, they may become computationally expensive or difficult to implement in high dimensions, complex geometries, or inverse problems \cite{allaire2007numerical}.

In this work, we focus specifically on QPINNs for elliptic PDEs and on their variational counterparts. More precisely, we investigate the lazy training regime for the elliptic problem 
\begin{align}\label{elliptic1}
\begin{cases}
&\A\,u=f\, \hskip 0,2cm \text{in $B$,}\\
& u=g\hskip 0,2cm \text{on $\partial B$,}
\end{cases}
\end{align}
where $B\subset \R^{d}$ is an open bounded set with Lipschitz boundary. Here, we take $f:B\rightarrow \R$ to be a measurable function with respect the Lebesgue measure on $B$, and $g:\partial B\rightarrow \R$  a measurable function with respect to the Hausdorff measure $\H^{d-1}$. We assume that our operator $\A$ has the form

\begin{align}\label{ellittico}
\A\,u= -\sum_{i,j=1}^{d}\frac{\partial}{\partial x_{j}}\left(a_{ij}(x)\frac{\partial u}{\partial x_{i}}\right)    
\end{align}
and it is elliptic with suitable regularity assumptions on the coefficient $a_{ij}$ that will be specified later on.

In recent years, physics-informed neural networks (PINNs) have emerged as a mesh-free alternative for approximating solutions of PDEs \cite{de2024numerical,farea2024understanding}.
In this framework, a neural network is trained to approximate the unknown solution by minimizing a loss function that penalizes violations of the governing differential equation and of the boundary or initial conditions.
PINNs have been successfully applied to a wide range of problems and have motivated a rapidly growing literature at the interface of numerical analysis and machine learning \cite{luo2025physics}.
In this paper, we study quantum neural networks applied to the elliptic PDE \eqref{ellittico}. We let $\overline{B}\coloneqq B\cup \partial B$ be the set of possible inputs, and $\Theta$ be the vector of the parameters. 
In what follows, we denote by $x\mapsto u(\Theta,x)$ the function generated by a quantum neural network. Before giving its formal expression, let us to describe informally how such a function is generated. A quantum neural network processes a classical input $x\in\R^d$ through a sequence of operations performed on a register of $m$ qubits. First, the input $x$ is encoded into the quantum state through a sequence of unitary gates whose action depends on the components of $x$; this 
step is the quantum analogue of feeding the input features into a classical network \cite{bishop2006}. Next, the resulting quantum state is transformed by a further sequence of unitary gates depending on a set of trainable parameters $\Theta$, playing the same role as the weights of a classical neural network; unlike a classical network, however, these two types of operations (the ones depending on $x$ and the ones depending on $\Theta$) are typically interleaved several times, forming what we call the layers of the circuit. Finally, a measurement of a fixed quantum observable $O$ is performed on the resulting state, and the function $u(\Theta,x)$ is defined as the expectation value of the outcome of this measurement. In this sense, $u(\Theta,x)$ plays a role analogous to the output of 
a classical neural network with input $x$ and weights $\Theta$, the fundamental difference being that the computation is carried out on a quantum state, through unitary evolutions. More precisely, $u(\Theta,x)$ is the expectation value of a quantum observable measured on the output of a quantum circuit made of parametric gates \cite{girardi2025,melchor2025quantitative,schuld2015,schuld2018}. Recall that a quantum gate is any unitary operator acting on one or more qubits. A parameterized (parametric) gate is a unitary depending on one or more tunable parameters; a common form is $e^{-i\frac{\theta H}{2}}$, where $H$ is a Pauli matrix, and $\theta\in\R$ \cite{liu2022representation}. The function $u(\Theta,x)$ is then given by
\begin{equation}\label{model}
    u(\Theta,x)=\bra{0^m}U^\dagger(\Theta,x)\, O\, U(\Theta,x)\ket{0^m},
\end{equation}
where $\ket{0^m}\coloneqq\ket{0}^{\otimes m}$ is the computational basis state, $O$ is the measured $m$-qubit observable, and $U(\Theta,x)$ is a parametric quantum circuit composed of $L$ layers, as anticipated above. Each layer combines two types of unitary operations: parametrized single-qubit gates, depending on the trainable parameters $\Theta=(\theta_1,\ldots,\theta_{Lm})^T$, and a further set of one or two-qubit gates, denoted by $V_\ell(x)$, acting on disjoint qubits. The latter gates are not all of the same nature: some of them depend on the 
input $x\in\overline B\subset\R^d$ and are used to encode $x$ into the quantum state (typically through angle encoding, i.e., gates of the form 
$e^{-ic_jx_jP_j}$, where $P_j$ is a Pauli matrix and $c_j$ is a constant); others act instead on qubits not directly associated with any input component; these are fixed, non-parametric gates -- typically two-qubit entangling gates whose action does not depend on $x$. Both types of gates act on different qubits and are included in $V_\ell(x)$, which therefore represents, for each layer $\ell$, the overall combination of input-encoding and fixed entangling gates prescribed by Definition \ref{def:numbL} below. Under mild regularity assumptions on the circuit architecture, the function $u(\Theta,x)$ is analytic both in the parameters $\Theta$ and in the input $x$. The detailed structure of 
the quantum circuit, including the layer decomposition and parameter indexing, is described in Section \ref{sub:quantum} below.
In the overparameterized regime, where the number of parameters is large compared to the number of training samples, the network operates in the so-called lazy training regime. In this regime, the parameters $\Theta$ remain close to their initialization throughout training, and the evolution of the network output is well approximated by its first-order Taylor expansion around the initial parameters $\Theta_{0}$. This linearization leads naturally to the neural tangent kernel (NTK) framework, where the kernel is defined as the expected value of the inner product between the gradients of $u(\Theta,x)$ with respect to $\Theta$ evaluated at two different inputs. The parameters are typically optimized by gradient descent, which involves iterative adjustments to minimize a cost function and improve the performance of the quantum circuit in the processing and analysis of data \cite{schuld2021effect}.
Let $\left\{\left(x^{(i)},\,f(x^{(i)}\right):i=1,\,\ldots,\,n_{1}\right\}$ be the set of the training examples made by the training inputs $x^{(i)}\in B$, and the corresponding training labels the values $f(x^{(i)})\in \R$. In the same manner, we let $\left\{\left(\hat{x}^{(j)},\,g(\hat{x}^{(j)}\right):j=1,\,\ldots,\,n_{2}\right\}$ training examples made by the training inputs $\hat{x}^{(i)}\in \partial B$, and the corresponding training labels the values $g(\hat{x}^{(i)})\in \R$. The goal of supervised learning is to adjust the parameters $\Theta$ so that $u(\Theta,x)$, and $\A u(\Theta,\hat{x})$ reproduce as closely as possible the training examples. This is usually achieved by minimizing a loss function such as the empirical quadratic loss
\begin{equation}\label{eq_loss}
\L(\Theta)\coloneqq \frac{1}{2}\sum_{i
=1}^{n_{1}}\left(\A u(\Theta,x^{(i)})-f(x^{(i)})\right)^{2}+ \frac{1}{2}\sum_{j=1}^{n_{2}}\left(u(\Theta,\hat{x}^{(j)})-g(\hat{x}^{(j)})\right)^{2}.
\end{equation}
 For simplicity, in this paper we will consider the continuous-time gradient flow rather than gradient descent. An important question is whether QNNs can actually provide quantum advantage, and their capabilities have been explored by several works \cite{lloyd2020quantum}. In \cite{liu2021} the authors showed that an exponential quantum speed-up can be obtained via the use of a quantum-enhanced feature space, where each data point is mapped in a non-linear way to a quantum state, and then classified by a linear classifier in a high-dimensional Hilbert space \cite{havlivcek2019}. Nevertheless, a significant disadvantage lies in the need to determine beforehand the appropriate parameters of the quantum circuit that maps the inputs to quantum states, and it is not yet clear whether these parameters can be effectively obtained using a variational technique \cite{cinelli2021var}. A rigorous mathematical characterization of the training dynamics of quantum neural networks becomes possible in the limit of infinite width. In \cite{girardi2025,melchor2025quantitative} the authors considered the joint limit of infinite width and depth and, under the hypothesis that the depth grows at most logarithmically with respect to the number of qubits, they proved that the probability distribution of the trained model function converges in distribution to a Gaussian process.
The key element of the proof is showing that the training happens in the lazy regime, and therefore the dependence of the model function $u(\Theta,x)$ on the parameters can be approximated by its linearized version near the initialization values. Consequently, in the limit the model becomes linear and the training has an analytic solution whose probability distribution is Gaussian with analytically computable mean and covariance.

As in the classical deep learning, the training dynamics of very wide neural networks is captured by the NTK, and therefore it is quite natural to extend this concept to QNNs. We point out that the challenge lies in the fact that computing the quantum NTK appears to be as hard as simulating the quantum network itself. Nevertheless, in \cite{hernandez2026efficient}, the authors presented an efficient classical algorithm to compute the NTK of a very broad family of QNNs. More precisely, the authors assumed that $U(\Theta,x)$ is a parametric quantum circuit made by nonparametric unitary operations belonging to the Clifford group with the components of $\Theta$ taken from the interval $[0,2\pi)$, and which can depend on the input $x$ in an arbitrary way but belonging to a finite set, interleaved with parametric gates given by the time evolution generated by a Hamiltonian $H$ belonging to the Pauli group. The key idea of the algorithm is that, in the computation of the NTK, the random initialization of the parameters of the network can be replaced by an average over the values $\left\{0,\,\frac{\pi}{2},\,\pi,\,\frac{3\pi}{2}\right\}$, for which the resulting quantum gates belong to the Clifford group. Since any quantum circuit made by gates belonging to the Clifford group can be simulated efficiently with a classical algorithm \cite{Dehaene_2003,Aaronson_2004}, the evaluation of the quantum NTK becomes tractable on a classical computer. However, in the present work $x$ belongs to the continuous set $B$, and since the Clifford group is finite there is no way to treat $x\mapsto u(\Theta,x)$ as a regular function parametrized by that group. To apply the simulability result of \cite{hernandez2026efficient}, one would need to discretize the domain $B$ so that the input-dependent gates $V_\ell(x)$ remain within the Clifford group for all training and test points. Such a discretization would necessarily destroy the smoothness of the map $x\mapsto u(\Theta,x)$, which is essential for the application of the differential operator $\A$ and for the variational formulation of the QPINN. Therefore, this simulability result does not apply in our setting. 

From a theoretical perspective, most existing works on PINNs focus on consistency, approximation properties, and a priori error estimates for minimizers of the PINN loss functional.Under suitable assumptions on the neural network architecture, the sampling of points, and the regularity of the PDE solution, it can be shown that minimizers of the PINN loss converge to the true solution of the PDE \cite{cheng2026consistency,wang2022and}.
However, these analyses typically rely on idealized optimization assumptions and provide limited insight into the training dynamics of the network \cite{hanin2018neural,hanin2019finitedepthwidthcorrections}.
The goal of this work is to bring the NTK and lazy-training perspective to the analysis of PINNs, focusing on the case of QNNs applied to the elliptic PDE \eqref{elliptic1} \cite{panichi2026quantum,klement2026explaining}. Quantum neural networks are a natural choice for PINN architectures because of their expressive power and the intrinsic structure induced by quantum circuits \cite{abedi2023,cerezo2021}. At the same time, their training dynamics exhibit distinctive features that significantly affect parameter behavior \cite{abedi2023,girardi2025,melchor2025quantitative}. Since a QNN $u(\Theta,x)$ is analytic, this property enables us to estimate its first and second derivatives when the operator $\A$ is applied to $u$. These estimates allow us to derive explicit bounds for our approximation results.

\textbf{Our contribution}
Our goal is to derive explicit, nonasymptotic bounds that make the lazy-training regime and the NTK description quantitative, in a way that can help to understand how the learned PDE solution evolves during training for both the differential equation and its variational counterpart. Our main results can be informally stated as follows.

\begin{thm}[Lazy training for QPINN --informal statement]\label{thm:lazy_informal_collocation}
Let us denote by $z\mapsto \U^{\lin}(\Theta,z)$ the first-order Taylor approximation of $z\mapsto \U(\Theta,z)$ with respect to the parameters $\Theta$ expanded around their initialization values, and where
\begin{align*}
\U(\Theta,z)
\coloneqq
\begin{pmatrix}
\A u(\Theta,x)\\
u(\Theta,\hat{x})
\end{pmatrix},
\end{align*}
with $z=(x,\hat{x})\in B\times \partial B$. Let $z\mapsto \U(\Theta_t^{\mathrm{lin}},z)$ be the model obtained by randomly initializing $\Theta$ and training $z\mapsto \U^{\mathrm{lin}}(\Theta,z)$ via gradient flow for time $t$. Then, for any $0<\delta<1$ there exist positive numbers $\left\{\gamma_{m,n,\delta}\right\}$ given by \eqref{grad3} depending on $m\delta$, $n=n_{1}+n_{2}$ and further constants, such that with probability at least $1-\delta$, one gets that
\begin{align}\label{intro:ratewasst3}
\sup_{\substack{z\in B\times \partial B\\t\geq 0}}\|\U(\Theta_t,z)-\U^{\mathrm{lin}}(\Theta_t^{\mathrm{lin}},z)\|_{2}&\leq \gamma_{m,n,\delta}.
\end{align}
\end{thm}
Let us now give a variational version of \autoref{thm:lazy_informal_collocation}. To this aim, let us fix a integer constant $M>0$, and we consider a finite dimensional space space $\H_{M}$ of dimension $M$ of test functions $v\in H_{0}^{1}(B)$ the Sobolev space of square integrable functions with square integrable on $B$ without trace at the boundary $\partial B$. Let us set

\begin{align}
\H_{M}\coloneqq\mathrm{span}\left\{v_{i}:i=1,\ldots,M\right\}.   
\end{align}
We define the terms
\begin{align}
 \cE(\Theta,v)\coloneqq \displaystyle\sum_{i,j=1}^{d}\int_{B}a_{ij}(x)\frac{\partial u(\Theta,\cdot)}{\partial x_{i}}\frac{\partial v}{\partial x_{j}}\de x,\hskip 0,2cm \F(v)\coloneqq \int_{B}fv\de x.
\end{align}
We take as a loss functional the one given by

\begin{align}
\L(\Theta)\coloneqq \frac{1}{2}\sum_{i=1}^{M}\left\vert \cE(\Theta,v_{i})-\F(v_{i})\right\vert^{2}+\frac{1}{2}\sum_{i=1}^{n_{2}}\left(u(\Theta,\hat{x}^{(i)})-g(\hat{x}^{(i)})\right)^{2}.   
\end{align}
\begin{thm}[Lazy training for variational QPINN --informal statement]\label{thm:lazy_informal_variational}
Let us denote by $w\mapsto \U^{\lin}(\Theta,w)$ the first-order Taylor approximation of $w\mapsto \U(\Theta,w)$ with respect to the parameters $\Theta$ expanded around their initialization values, and where
\begin{align}
\U(\Theta,w)\coloneqq
\begin{pmatrix}
\cE(\Theta,v)\\
u(\Theta,\hat{x})
\end{pmatrix},
\end{align}
with $w=(v,\hat{x})\in \H_{M}\times \partial B$. Let $w\mapsto \U(\Theta_t^{\mathrm{lin}},w)$ be the model obtained by randomly initializing $\Theta$ and training $w\mapsto \U^{\mathrm{lin}}(\Theta,w)$ via gradient flow for time $t$. Then, for any $0<\delta<1$ there exist positive numbers $\left\{\beta_{m,n,\delta}\right\}$ given by \eqref{grad3} depending on $m\delta$, $n=M+n_{2}$ and further constants, such that with probability at least $1-\delta$, one gets that
\begin{align}\label{intro:ratewasst3}
\sup_{\substack{w\in H_{M}\times \partial B\\t\geq 0}}\|\U(\Theta_t,w)-\U^{\mathrm{lin}}(\Theta_t^{\mathrm{lin}},w)\|_{2}&\leq \beta_{m,n,\delta}.
\end{align}
\end{thm}
The key differences between the case using the operator $\A$and variational approaches are as follows. First, the scaling condition for NTK convergence in the case using the differential operator \eqref{ipoconvergencentk} requires
\begin{align*}
\lim_{m\rightarrow+\infty}\frac{mL^{9}\vert\mathcal{M}\vert^{4}\vert\mathcal{N}\vert^{2}}{(b(m))^{4}}=0,
\end{align*}
while in the variational case \eqref{Var_ipoconvergencentk} we only need
\begin{align*}
\lim_{m\rightarrow+\infty}\frac{mL^{5}\vert\mathcal{M}\vert^{4}\vert\mathcal{N}\vert^{2}}{(b(m))^{4}}=0.
\end{align*}
The variational formulation thus requires a weaker growth condition on the depth $L$, improving from $L^9$ to $L^5$. This improvement stems from the fact that in the variational approach, we only need to control first-order spatial derivatives of $u(\Theta,x)$ (which grow as $L$), whereas the case using the operator $\A$ approach requires controlling second-order spatial derivatives appearing in $\A u(\Theta,x)$ (which grow as $L^2$). Correspondingly, the linearization error constant $\gamma_{m,n,\delta}$ scales as $L^{12}$ in the differential case, while $\beta_{m,n_V,\delta}$ scales as $L^{8}$ in the variational case. Second, in the case using the operator $\A$, we imposed the constraint \eqref{dimensionalconstraint}:
\begin{align*}
2A_{1}+4A_{0}\leq 1,
\end{align*}
where $A_0=\sup_{x\in\overline B}|a_{ij}(x)|$ and $A_1=\sup_{x\in\overline B}|\partial_{x_k} a_{ij}(x)|$. In contrast, the variational formulation does not require this constraint. The coefficients $a_{ij}$ need only satisfy
\begin{align*}
\sup_{x\in\overline{B}}|a_{ij}(x)|\leq A_0<\infty,\quad a_{ij}\in C^1(\overline{B}),
\end{align*}
with $A_0$ arbitrarily large. This is because the variational bounds involve only $A_0$ (and not $A_1$), and the key Lipschitz estimate depends linearly on $A_0$ without requiring any constraint relating $A_0$ and $A_1$. Thus, the PDE coefficients are essentially required to be merely $L^\infty$ functions with bounded first derivatives.
 
Lastly, the variational bounds involve the volume ${\rm vol}(B)$ of the spatial domain, which appears naturally through the integral formulation, while in the case using the operator $\A$ approach the domain geometry enters implicitly through the choice of training points. Finally, for the case with the operator $\A$, the method uses $n=n_1+n_2$ training samples (with $n_1$ interior points and $n_2$ boundary points), whereas the variational method uses $n_V=M+n_2$ degrees of freedom (with $M$ test functions in $H_{0}^{1}(B)$ and $n_2$ boundary points), providing additional flexibility in the choice of the test function space $\H_M$.
 
The precise statements, including the explicit expressions for $\gamma_{m,n,\delta}$ and $\beta_{m,n_V,\delta}$, are given in \autoref{thm:lazytraining} and \autoref{thm:Vlazytraining} below.

The paper is organized as follows. \autoref{sec:framework} introduces the mathematical framework for the PDE problem, the quantum circuit architecture, gradient flow dynamics, and the empirical neural tangent kernel (ENTK). In \autoref{sec:concentration} we establish concentration of the ENTK around its expected value, and we present the main lazy training results for the case using the operator $\A$ approach (see \autoref{thm:lazytraining}). In \autoref{sec:variational} we develop develop the variational formulation. After recalling the weak solution framework, we introduce the variational loss and prove concentration of the variational NTK (see  \autoref{thm:VNTK}), and we present \autoref{thm:Vlazytraining} which establishes the lazy training for the variational case. Section \autoref{sec:comparison} compares our results with previous results in the literature, emphasizing our treatment of variable coefficients and probabilistic estimates. Section \ref{sec:conclusions} concludes and outlines future directions. We add several Appendices \ref{sec:limit_thms} where we recall some useful limit theorems for stochastic processes, and further results of own interest.

\begin{table}[t]
    \caption{Notation concerning the QNN}
    \label{table1}
    \begin{tabularx}{\textwidth}{p{0.1\textwidth}X>{\raggedleft\arraybackslash}l} 
    \toprule
    Symbol & Description & Introduced in \\
    \midrule
      %{\underline{Indices}} \\
      $m$ & number of qubits in the parameterized quantum circuit & \autoref{sub:quantum}\\ 
      $L$ & number of layers in the parameterized quantum circuit& Def. \ref{def:numbL}\\
      %$[\ell \, m]$ & layer-qubit representation & Def. \ref{def:layerq}\\
      $\Theta$ & vector of the parameters of the quantum circuit&\autoref{parm}\\
      $\mathscr{P}$& denotes the parameter space, so that $\Theta\in \mathscr{P}$. Here $\mathscr{P}=[0,\pi]^{Lm}$& \autoref{sub:trdata}\\
      $|\Theta|$ & number of parameters $|\Theta|:=\dim\mathscr{P}=Lm$ & \autoref{sub:quantum} \\
      $U(\Theta,x)$ & parameterized quantum circuit (unitary operator) & \textit{ibid.}\\
      $u(\Theta,x)$ & function generated by the quantum neural network & \autoref{model1}\\
      $b(m)$ & normalization factor of the model & \textit{ibid.}\\
      $\mathcal{M}_i$& (extended) future light cone of the parameter $i$ & Def. \ref{extcone2}\\
      $\mathcal{N}_k$ &(extended) past light cone of the observable $k$& Def. \ref{lcone1}\\
      $|\mathcal{M}|$ & maximal cardinality of a future light cone in the circuit &  \autoref{maxicard}\\
      $|\mathcal{N}|$ & maximal cardinality of a past light cone in the circuit & \textit{ibid.}\\
      $\mathcal{P}_{i}$ &set of indices of the observables depending on the observable $i$&\autoref{pk1}\\
      $\widetilde{\mathcal{P}}_{i}$& set representing the union of the sets $\mathcal{P}_{j}$ for those $j$ in $\mathcal{P}_{i}$ &  \autoref{newsets}\\
      \bottomrule
    \end{tabularx}
\end{table}
\begin{table}[t]
    \caption{Notation concerning the training of the QPINN.}
    \label{table2}
    \begin{tabularx}{\textwidth}{p{0.1\textwidth}X>{\raggedleft\arraybackslash}l} 
    \toprule
    Symbol & Description & Introduced in \\
    \midrule
    $B$& open bounded set of $\R^{d}$ with Lipschitz boundary&\autoref{sub:trdata}\\
    $\partial B$& boundary of $B$&\textit{ibid}\\
    $\overline{B}$& topological closure of $B$ interpreted as the total input set& \textit{ibid}\\
    $\A$& differential operator&\autoref{operator}\\
    $a_{ij}$& coefficients variables of $\A$&\textit{ibid}\\
    $f$& datum of the PDE problem in $B$&\textit{ibid}\\
     $g$& datum of the PDE problem on $\partial B$&\textit{ibid}\\
      $x$ & a generic input belonging to $B$& \textit{ibid}\\
      $\hat{x}$ & a generic input belonging to $\partial B$& \textit{ibid}\\
      $\mathcal{Y}$ & the output space & \textit{ibid}\\
      $f(x)$ & a generic output value belonging to $\mathcal{Y}$& \textit{ibid}\\
      $g(\hat{x})$ & a generic output value belonging to $\mathcal{Y}$& \textit{ibid}\\
      $\mathcal{D}_{B}$ & training set on $B$, whose elements are denoted by $(x^{(i)},f(x^{(i)}))$ for $i=1,\dots,n_{1}$ & \textit{ibid.} \\
      $n_{1}$ & number of training samples on $B$ (i.e., cardinality of $\mathcal{D}_{B}$) & \textit{ibid.}\\
      $\mathcal{D}_{\partial B}$ & training set on $\partial B$, whose elements are denoted by $(\hat{x}^{(i)},g(\hat{x}^{(i)}))$ for $i=1,\dots,n_{2}$ & \textit{ibid.} \\
      $n_{2}$ & number of training samples on $B$ (i.e., cardinality of $\mathcal{D}_{\partial B}$) & \textit{ibid.}\\
      $X$ & vector containing the inputs of the training set $D_{B}$& \autoref{sub:NTK}\\
      $\hat{X}$ & vector containing the inputs of the training set $D_{\partial B}$& \textit{ibid.}\\
      $Z$ & vector containing the inputs of the training sets $D_{B}$, $ D_{\partial B}$ & \textit{ibid.}\\
      $Y$& vector containing the outputs of the training sets $D_{B}$, $D_{\partial B}$& \textit{ibid.}\\
      $z=(x,\hat{x})$& a generic input belonging to $X_{B}\times \partial B$& \textit{ibid.}\\
      $\U^{\mathrm{lin}}(\Theta,z)$ & linearized model & \autoref{sub:linearmod}\\
      $ \hat K_{\Theta}(z,z')$ & empirical neural tangent kernel & Def. \ref{def:ENTK}\\
      $ K(x,x')$ & analytic neural tangent kernel & Assumption \autoref{A3}\\
      $\lambda_{\min}^K$ & smallest eigenvalue of $K(X,X^T)$ & Assumption \autoref{A3} \\
      %$ \ov{K}(x,x')$ & limit kernel of the rescaled analytic neural tangent kernel & \autoref{limitk}\\
      $t$ & continuous or discrete training time & \autoref{sub:NTK}\\
      $\eta$ & learning rate, which enters the gradient flow equation and is a function of $m$ & \autoref{gradform1}\\
      $\mathcal{L}(\Theta)$ & cost function for the original model according to the training set & \autoref{costfunct2}\\
      $\Theta_t$ & parameter vector evolving via gradient flow according to $\mathcal{L}$ & \autoref{gradform1} \\
      $\mathcal{L}^{\mathrm{lin}}(\Theta)$ & cost function for the linearized model & \autoref{sub:linearmod} \\
      $\Theta_t^{\mathrm{lin}}$ & parameter vector evolving via gradient flow according to $\mathcal{L}^{\lin}$ & \textit{ibid.}\\
      $\mathbb{M}_{2\times 2}$ & space of $2\times 2$ matrices with real entries& \autoref{sub:linearmod}\\
      \bottomrule
    \end{tabularx}
\end{table}

\section{Preliminaries}\label{sec:framework}
Let us start by introducing the notation of the present work.
\subsection{Training data}\label{sub:trdata}
In what follows, we denote by $B\subset \R^{d}$ an open bounded with Lipschitz boundary subset of $\R^{d}$. We consider the partial differential equation, in short PDE, of the form

\begin{align}\label{linearPDE}
\begin{cases}
&\A\,u=f\, \hskip 0,2cm \text{in $B$,}\\
& u=g\hskip 0,2cm \text{on $\partial B$.}
\end{cases}
\end{align}
We consider the operator $\A$ be defined as 

\begin{align}\label{operator}
\A\,u(\theta,x)= -\sum_{i,j=1}^{d}\frac{\partial}{\partial x_{j}}\left(a_{ij}(x)\frac{\partial u(\theta,x)}{\partial x_{i}}\right)    
\end{align}
where $a_{ij}(\cdot)\in C^{1}(\overline{B})$, $1\leq i,j\leq d$ satisfying the ellipticity condition

\begin{align}\label{ellipticiy}
\sum_{i,j=1}^{d}a_{ij}(x)\xi_{i}\xi_{j}\geq \beta \norma{\xi}_{2}^{2}, \hskip 0,2cm \text{for all $x\in B$, and for all $\xi \in \R^{d}$ with $\beta>0$ fixed.}
\end{align}

Furthermore, we take $f:B\rightarrow \R$ to be a measurable function with respect the Lebesgue measure on $B$, and $g:\partial B\rightarrow \R$  a measurable function with respect to the Hausdorff measure $\H^{d-1}$. Here, we consider $\overline{B}$ as the set of all the possible inputs, and $\R$ the output set. Due to the elliptic equation \eqref{linearPDE}, we let
\begin{equation}
    \D_{B} \coloneqq\left\{(x^{(i)},f(x^{(i)})):i=1,\ldots,n_{1}\right\}\subset B\times\R
\end{equation}
be the training set constituted of training points on $B$. In the same way, we set

\begin{equation}
    \D_{\partial B} \coloneqq\left\{(\hat{x}^{(i)},g(\hat{x}^{(i)})):i=1,\ldots,n_{2}\right\}\subset \partial B\times\R,
\end{equation}
where we set $n_{1}=\vert \D_{B}\vert$ to be the cardinality of $\D_{B}$, and $n_{2}=\vert \D_{\partial B}\vert$ to be the cardinality of $\D_{\partial B}$. In what follows, we let $\cP$ be the parameter space, and let $\Theta\in\cP$ be the vector of the parameters. Let $u:\cP\times \overline{B}\rightarrow \mathbb{R}$ be a generic  parametric function, where $\overline{B}$ is the closure of $B$.

As a cost function, we consider the mean squared error on the training sets $\D_{B}$, and $\D_{\partial B}$ of cardinality $n_{1}$, and $n_{2}$, respectively:

\begin{align}\label{costfunct2}
 \L(\Theta)\coloneqq \frac{1}{2}\sum_{i
=1}^{n_{1}}\left(\A u(\Theta,x^{(i)})-f(x^{(i)})\right)^{2}+ \frac{1}{2}\sum_{i
=1}^{n_{2}}\left( u(\Theta,\hat{x}^{(i)})-g(\hat{x}^{(i)})\right)^{2}.
\end{align}
\subsection{Quantum neural networks}\label{sub:quantum}
Let $\CC^{2}$ be the Hilbert space of a single qubit. In what follows, we denote by $m\in \N$ the number of qubits of the quantum neural network. Hence, the Hilbert space of the system is $\H=\left(\CC^{2}\right)^{\otimes m}$, and its dimension denoted as $\dim\,\H$ is $2^{m}$. Here, a quantum gate is any unitary operator acting on one or more qubits. Following the notations of \cite{girardi2025}, we recall what a ``layer'' is.
\begin{defn}\label{def:numbL}
A layer is a unitary operation $U(\Theta,x)\in \Lc(\H)$ resulting from:
\begin{enumerate}
    \item[$1.$] the application on each qubit of a different parametrized single-qubit gate $W_{i}(\Theta)\in \Lc(\CC^{2})$; each parametrized gate depends on a single parameter $\theta_{i}$, which is different for each gate,
\end{enumerate}
followed by
\begin{enumerate}
    \item[$2.$] a set of one-qubit and two-qubit gates acting on disjoint qubits, that is, each qubit can be acted at most one gate; each gate may depend only on the input $x$; the resulting unitary operation will be called $V\in \Lc(\H)$. 
\end{enumerate}
In what follows, a quantum circuit is a combination of parameterized layers $U_{\ell}(\Theta,x)$, $\ell\in \mathbb{N}$. Next, we let $L\in \mathbb{N}$ be the number of layers in a quantum circuit, which may depend on  the number of qubits $m$.
\end{defn}
In the next, we consider $\theta_{1},\ldots, \theta_{Lm}$ be the parameters of a quantum circuit,  so that  $\Theta$ will be the vector

\begin{align}\label{parm}
\Theta\coloneqq
 \begin{pmatrix}
 \theta_{1}\\
 \theta_{2}\\
 \vdots\\
 \theta_{Lm}
 \end{pmatrix},
\end{align}
of dimension $\mathrm{dim}\Theta\coloneqq\vert \Theta\vert=Lm$. We now recall a convenient notation for the indices of the parameters used in \cite{girardi2025}.

\begin{defn}\label{def:layerq}
Each parameter index $i\in \{1,\ldots, Lm\}$ can be expressed in the form $i=m(\ell-1)+k$ for some $\ell\in\{1,\ldots,L\}$, and $k\in\{1,\ldots,m\}$.  Here, $k$ refers to the qubit involved in the single-qubit gate parametrized by $\theta_{i}$, while $\ell$ refers to the layer in which such gate acts. The following compact notation, which we call layer-qubit representation of the parameter index $i$, simplifies the above form:
\begin{align}
    i=[\ell m]\equiv m(\ell-1)+k.
\end{align}
\end{defn}
Therefore, a layer $U_{\ell}(\Theta,x)$ can be written as

\begin{align}
\begin{aligned}
U_{\ell}(\Theta,x)&\coloneqq V_{\ell}(x)\left(W_{[\ell 1]}\otimes \cdots \otimes W_{[\ell m]}\right)(\Theta)\\
&=V_{\ell}(x)W_{\ell}(\Theta),
\end{aligned}
\end{align}
where we have set $W_{\ell}(\Theta)\coloneqq \left(W_{[\ell 1]}\otimes \cdots \otimes W_{[\ell m]}\right)(\Theta)$. The result of the circuit on a initial state $\ket{\psi_{0}}$ is described by the unitary operation 

\begin{align}\label{formula:1}
U(\Theta,x)\coloneqq U_{L}(\Theta,x)\cdots U_{1}(\Theta,x); \qquad \ket{\psi_{{\rm out}}}\coloneqq U(\Theta,x)\ket{\psi_{0}}.
\end{align}
\subsection{Light cones}\label{sub:lightc}
In this part, we closely follow \cite{girardi2025,melchor2025quantitative} and we recall the notion of light cones. The architecture of the network generates a causal structure where the probability distribution of the outcome of the measurement of each output qubit can depend only on some of the parameters, and each parameter can influence only some output qubits. This causal structure is formalized by the notion of light cones:
\begin{defn}[Light cones]\label{lightcones}
For any $i\in\{1,\ldots, \vert \Theta\vert\}$, we define the future light cone $\Lg_{i}^{f}$ of the parameter $\theta_{i}$ as the subset 

\begin{align}\label{lcone1}
\Lg_{i}^{f}\coloneqq \left\{k\in\{1,\ldots,m\}:\text{$f_{k}(\Theta,x)$ depends on $\theta_{i}$}\right\}.
\end{align}
Analogously, we define the past light cone $\Lg_{k}^{p}$ of the qubit $k$ as the subset
\begin{align}\label{lcone2}
\Lg_{k}^{p}\coloneqq\{i\in\{1,\ldots,\vert \Theta\vert\}: \text{$f_{k}(\Theta,x)$ depends on $\theta_{i}$}\}.
\end{align}
\end{defn}
Both sets $\Lg_{i}^{f}$, and $\Lg_{k}^{f}$ are useful for tracking the dependence of observables on the parameters. In general, it is difficult to provide an explicit representation of them. For this reason, we now introduce another family of sets that can help us to explicitly track the dependence on the parameters. For any quantum circuit $U$, we define the following sets. For each layer $\ell$, and qubit $k$, we set

\begin{align}\label{auxset1}
    \mathcal{I}_{\ell,k}\coloneqq \{k'\in\{1,\ldots,m\}:\text{the qubit $k$ interacts with the qubit $k'$ in the layer $\ell$}\}\cup \{k\}.
\end{align}
We now set,

\begin{align}\label{auxset2}
\mathcal{J}_{k}^{\ell}\coloneqq 
\begin{cases}
& \mathcal{I}_{L,k} \hskip 0,3cm \text{if $\ell=L$,}\\
&\displaystyle\bigcup_{k'\in \mathcal{J}_{k}^{\ell+1}}\mathcal{I}_{\ell,k'}\hskip 0,3cm \text{if $\ell<L$}.
\end{cases}
\end{align}

In particular $\mathcal{J}_k^1$ is the set of qubits in the past light cone of the observable $k$, i.e., the qubits involved in the computation of its expectation value.

Furthermore, we set

\begin{align}
\mathcal{N}_{k}^{\ell}\coloneqq \displaystyle\bigcup_{k'\in \mathcal{J}_{k}^{\ell}}\{[\ell k']\}.
\end{align}

\begin{defn}[Extended light cones]\label{extlightc}
Let us fix a quantum circuit $U$. Given any qubit index $k\in\{1,\ldots,m\}$, we define the extended past light cone $\mathcal{N}_{k}$ as the subset of the parameter indices $\{1,\ldots, \vert \Theta\vert\}$ given by

\begin{align}\label{extcone1}
   \mathcal{N}_{k}\coloneqq \displaystyle \bigcup_{\ell=1}^{L}\mathcal{N}_{k}^{\ell} 
\end{align}
\end{defn}
Similarly, we define the extended future light cone of a parameter index $i\in\{1,\ldots, \vert \Theta\vert\}$, as

\begin{align}\label{extcone2}
 \mathcal{M}_{i}\coloneqq \{k\in\{1,\ldots,m\}: i\in \mathcal{N}_{k}\}. 
\end{align}
In the next, we set 
\begin{align}\label{maxicard}
\begin{aligned}
&|\mathcal{M}|\coloneqq \displaystyle\max_{i}\vert \mathcal{M}_{i}\vert;
&|\mathcal{N}|\coloneqq \displaystyle \max_{k}\vert \mathcal{N}_{k}\vert
\end{aligned}
\end{align}
the maximal cardinalities of the extended light cones.
\begin{rem}
We notice that $\Lg_{k}^{p}\subset \mathcal{N}_{k}$, and $\Lg_{i}^{f}\subset \mathcal{M}_{i}$. and thus from now on, we can only consider the extended light cones. 
\end{rem}
In what follows, let us set 

\begin{align}
V\coloneqq \{1,\ldots,m\}.
\end{align}
For each $k\in V$, we define

\begin{align}\label{pk1}
\mathcal{P}_{k}\coloneqq \{k'\in V: \text{$f_{k'}(\Theta,x)$ is not independent from $f_{k}(\Theta,x)$}\}.   
\end{align}
This set is crucial since takes track of the number of random variables $f_{k'}(\Theta,x)$ that have correlation with $f_{k}(\Theta,x)$. Let $\mathscr{G}=(V,E)$ be the graph  with vertices $V$, and edges $E$ defined as follows. We say
\begin{align}\label{graphrel}
\text{$(k,k')\in E$ if and only if $k'\in \mathcal{P}_{k}$.}   
\end{align}

Furthermore, we define the {\em maximal degree} $D$ of $\mathscr{G}$ as the maximum number of edges containing any fixed vertex as
\begin{align}\label{maxdegree}
&D\coloneqq \max_{k\in V}{\rm deg}\,k=\max_{k\in V}\vert\{k'\in V: (k,k')\in E\} \vert=\max_{k\in V}|\mathcal{P}_k|.   
\end{align}
Let us notice that according to the definition of $\mathcal{P}_{i}$, we have that $j\in \mathcal{P}_{i}$ if and only if $i\in \mathcal{P}_{j}$.
Let ${\rm dist}$ be the distance on $\mathscr{G}$ given by the length of the shortest path, such that for any $i\in V$ we have
\begin{equation}
\mathcal{P}_i = \left\{j\in V : \mathrm{dist}(i,j) \le 1\right\}\,.
\end{equation}

In what follows, we set
\begin{align}\label{newsets}
\widetilde{\mathcal{P}}_{i}\coloneqq \displaystyle\bigcup_{j\in \mathcal{P}_{i}}\mathcal{P}_{j} = \left\{j\in V : \mathrm{dist}(i,j) \le 2\right\},   
\end{align}
and 

\begin{align}\label{indmeas}
&\widetilde{D}\coloneqq \max_{1\leq i\leq m}\left\vert \left\{ (i,j): \widetilde{\mathcal{P}}_{i}\cap\widetilde{\mathcal{P}}_{j}\neq \emptyset\right\}\right\vert = \max_{1\leq i\leq m}\left\vert \left\{ (i,j): \mathrm{dist}(i,j)\le 4 \right\}\right\vert.
\end{align}

\begin{lem}[{\cite[Lemma 2.1, Lemma 2.2]{melchor2025quantitative}}]\label{stimacone1}
 For any $k\in\{1,\ldots,m\}$, let $\mathcal{P}_{k}$ be defined as in \eqref{pk1}. Then 
 \begin{align}\label{conseqlem}
   \vert\mathcal{P}_{k} \vert\leq\vert\mathcal{M} \vert \vert \mathcal{N}\vert.
 \end{align}
 In particular,
 \begin{align}
   D\leq\vert\mathcal{M} \vert \vert \mathcal{N}\vert\\
    \widetilde{D} \le \vert\mathcal{M} \vert^4\, \vert \mathcal{N}\vert^4\,.
    \end{align}
\end{lem}

\subsection{The neural tangent kernel}\label{sub:NTK}
Before presenting our main results, let us review some relevant facts about the quantum neural tangent kernel as presented in \cite{girardi2025}. We are interested in the analysis of the minimization of the cost function \eqref{costfunct2} via gradient flow:
\begin{align}\label{gradform1}
\frac{\de \Theta_{t}}{\de\,t}=-\eta \nabla_{\Theta}\L(\Theta_{t})  
\end{align}
where $\eta>0$ is the learning rate that can be reabsorbed by rescaling the training time, and the initial value of $\Theta$ is given by the random sampling of $\Theta_{0}$. Notice that 
\begin{align}\label{gradvar1}
\frac{\de }{\de t}\L(\Theta_{t})=\frac{\de \Theta_{t}}{\de t}\cdot \nabla_{\Theta}\L(\Theta_{t})=-\eta \norma{\nabla_{\Theta}\L(\Theta_{t})}_{2}^{2}\leq 0.   
\end{align}
We stress that in general, the loss function $\L(\Theta)$ is not convex, hence gradient flow is not guaranteed to converge to a global minimum.
Given the training set $\mathcal{D}_{B}=\{(x^{(i)},f(x^{(i)}))\}_{i=1,\dots,n_{1}}$, we will represent it in a vectorized form as follows
\begin{align}
X=\begin{pmatrix} x^{(1)}\\x^{(2)}\\\vdots\\x^{(n_{1})} \end{pmatrix},\qquad 
f(X)=\begin{pmatrix} f(x^{(1)})\\f(x^{(2)})\\\vdots\\f(x^{(n_{1})}) \end{pmatrix}.
\end{align}
Similarly for the training set $\D_{\partial B}$:

\begin{align}
\hat{X}=\begin{pmatrix} \hat{x}^{(1)}\\\hat{x}^{(2)}\\\vdots\\\hat{x}^{(n_{2})} \end{pmatrix},\qquad 
g(\hat{X})=\begin{pmatrix} g(\hat{x}^{(1)})\\g(\hat{x}^{(2)})\\\vdots\\g(\hat{x}^{(n_{2})}) \end{pmatrix}.
\end{align}

Given any function $g:\mathbb{R}\to\mathbb{R}$, we will often use the following notation:
\begin{align}
    g(X)\coloneqq \begin{pmatrix} g(x^{(1)})\\g(x^{(2)})\\\vdots\\g(x^{(n_{1})}) \end{pmatrix},\qquad g(X^T)\coloneqq \begin{pmatrix} g(x^{(1)})&g(x^{(2)})&\cdots& g(x^{(n_{1})}) \end{pmatrix}
\end{align}
Similarly, for any bivariate function $K:\mathbb{R}\times\mathbb{R}\to\mathbb{R}$ we will write $K(X,X^T)$ to indicate the $n_{1}\times n_{1}$ matrix with entries $\left(K(X,X^T)\right)_{ij}\coloneqq K(x^{(i)},x^{(j)})$ for $1\leq i,j\leq n_{1}$. In what follows, we denote by
\begin{equation}
    X^{T}= \left(x^{(1)},\,\ldots,\,x^{(n_{1})}\right)
\end{equation}
the vector of the training inputs.
We set
\begin{align}\label{fvectmod}
\begin{aligned}
\U(t)\coloneqq
\begin{pmatrix}
u(\Theta_{t},x^{(1)})\\
u(\Theta_{t},x^{(2)})\\
\vdots\\
u(\Theta_{t},x^{(n_{1})})
\end{pmatrix} 
\end{aligned}
= u(\Theta_{t},X)\,
\end{align}
where 

\begin{align}\label{modelfunction}
u(\Theta,x)=\frac{1}{b(m)}\sum_{k=1}^{m}u_{k}(\Theta,x),\hskip 0,2cm u_{k}(\Theta,x)\coloneqq \bra{0^{m}}U^{\dagger}(\Theta,x)\O_{k}U(\Theta,x)\ket{0^{m}}, \hskip 0,2cm x\in \overline{B}.    
\end{align}
From the gradient-flow equation \eqref{gradform1} and the chain rule, the evolution equations for the parameters and the model function can be written as
\begin{align}\label{gradeq1}
 \displaystyle\begin{cases}
  &\frac{\de \Theta_{t}}{\de t}=-\eta\nabla_{\Theta}\A u(\Theta_{t},X^{T})\nabla_{\A u(\Theta_{t},X)}\L(\Theta_{t})-\eta\nabla_{\Theta} u(\Theta_{t},\hat{X}^{T})\nabla_{u(\Theta_{t},\hat{X})}\L(\Theta_{t}),\\
 &\frac{\de}{\de t}u(\Theta_{t},x)=-\eta\left(\nabla_{\Theta}u(\Theta_{t},x)\right)^{T}\nabla_{\Theta}\A u(\Theta_{t},X^{T})\nabla_{\A u(\Theta_{t},X)}\L(\Theta_{t})\\
 &\hskip 4cm-\eta\left(\nabla_{\Theta}u(\Theta_{t},x)\right)^{T}\nabla_{\Theta}u(\Theta_{t},\hat{X}^{T})\nabla_{ u(\Theta_{t},\hat{X})}\L(\Theta_{t})\,,\\
  &\frac{\de}{\de t}\A u(\Theta_{t},x)=-\eta\left(\nabla_{\Theta}\A u(\Theta_{t},x)\right)^{T}\nabla_{\Theta}\A u(\Theta_{t},X^{T})\nabla_{\A u(\Theta_{t},X)}\L(\Theta_{t})\\
 &\hskip 4cm-\eta\left(\nabla_{\Theta}\A u(\Theta_{t},x)\right)^{T}\nabla_{\Theta}u(\Theta_{t},\hat{X}^{T})\nabla_{ u(\Theta_{t},\hat{X})}\L(\Theta_{t})\,,
 \end{cases}   
\end{align}
where $\nabla_{\Theta}u(\Theta_{t},X^{T})$ denotes the gradient of $u(\Theta_{t},X^{T})$ with respect to $\Theta$ while $\nabla_{\A u(\Theta_{t},X)}\L(\Theta_{t})$ indicates the gradient of the cost function $\L$ with respect to $\A u(\Theta_{t},X)$ (similarly for $u(\Theta_{t},X)$). Recall that $^T$ is the transposition operator.
Let us now linearize \eqref{gradeq1}. In what follows, we set $z=(x,\hat{x})$ where $x\in B$, and $\hat{x}\in \partial B$. We define
\begin{align}
\U(\Theta,z)
\coloneqq
\begin{pmatrix}
\A u(\Theta,x)\\
u(\Theta,\hat{x})
\end{pmatrix}
\end{align}
and we use the notation
\begin{align}
y=
\begin{pmatrix}
f(x)\\
g(\hat{x})
\end{pmatrix}.
\end{align}
Furthermore, we set
\begin{align}
Y=
\begin{pmatrix}
f(X)\\
g(\widehat{X})
\end{pmatrix}.
\end{align}

\begin{defn}\label{def:ENTK}
Let $z=(x,\hat{x})$, $z'=(x',\hat{x}')$. We define the empirical NTK as

\begin{align}
\begin{aligned}
\widehat{K}_{\Theta}(z,z')\coloneqq&\frac{1}{b_{K}(m)}\begin{pmatrix}
\left(\nabla_{\Theta}\A u(\Theta,x)\right)^{T}\nabla_{\Theta}\A u(\Theta,x') & \left(\nabla_{\Theta}\A u(\Theta,x)\right)^{T}\nabla_{\Theta} u(\Theta,\hat{x}')\\
\nabla_{\Theta}u(\Theta,\hat{x})\nabla_{\Theta}\A u(\Theta,x')& \left(\nabla_{\Theta}u(\Theta,\hat{x})\right)^{T}\nabla_{\Theta}u(\Theta,\hat{x}')
\end{pmatrix}, \\
&=\frac{1}{b_{K}(m)}\big(\nabla_\Theta \mathcal U(\Theta,z)\big)
\big(\nabla_\Theta \mathcal U(\Theta,z')\big)^{T}.
\end{aligned}
\end{align}
\end{defn}
The normalization constant $b_{K}(m)$ is chosen so that \autoref{A3} holds, which ensures the existence of a limiting kernel $\overline{K}$. Moreover, its growth can be shown to be controlled by $|\mathcal{N}|$; see \autoref{cor:ordinegrandezza} below. Notice that our set of equations can be written as

\begin{align}
\begin{aligned}
&\frac{\de \Theta_{t}}{\de t}=-\eta\nabla_{\Theta}\A u(\Theta_{t},X^{T})\nabla_{\A u(\Theta_{t},X)}\L(\Theta_{t})-\eta\nabla_{\Theta} u(\Theta_{t},\hat{X}^{T})\nabla_{u(\Theta_{t},\hat{X})}\L(\Theta_{t}),\\
&\frac{\de \U(\Theta_{t},z)}{\de t}=-\eta b_{K}(m)\widehat{K}_{\Theta}(z,Z^{T})\nabla_{\U(\Theta_{t},Z)}\L(\Theta_{t})
\end{aligned}
\end{align}
where $Z^{T}=(X^{T},\hat{X}^{T})$, and

\begin{align}
\nabla_{\U(\Theta_{t},Z)}\L(\Theta_{t})\coloneqq 
\begin{pmatrix}
\nabla_{\A u(\Theta_{t},X)}\L(\Theta_{t})\\
\nabla_{u(\Theta_{t},\widehat{X})}\L(\Theta_{t})
\end{pmatrix}.
\end{align}
On the other hand, the set of equations for the linearized model reads as

\begin{align}
\begin{aligned}
&\frac{\de \Theta_{t}^{\lin}}{\de t}=-\eta\nabla_{\Theta}\A u^{\lin}(\Theta_{t}^{\lin},X^{T})\nabla_{\A u^{\lin}(\Theta_{t}^{\lin},X)}\L^{\lin}(\Theta_{t}^{\lin})-\eta\nabla_{\Theta} u^{\lin}(\Theta_{t}^{\lin},\hat{X}^{T})\nabla_{u^{\lin}(\Theta_{t}^{\lin},\hat{X})}\L^{\lin}(\Theta_{t}^{\lin}),\\
&\frac{\de \U^{\lin}(\Theta_{t}^{\lin},z)}{\de t}=-\eta b_{K}(m)\widehat{K}_{\Theta}(z,Z^{T})\nabla_{\U^{\lin}(\Theta_{t},Z)}\L^{\lin}(\Theta_{t}^{\lin})
\end{aligned}
\end{align}
We fix the random initialization $\Theta_0$.
\subsection{Linearized model}\label{sub:linearmod}
The linearized (first-order Taylor) model is defined as
\begin{align}\label{Ulin_def}
\U^{\lin}(\Theta_{t}^{\lin},z)
=
\U(\Theta_0,z)
+\nabla_{\Theta}\U(\Theta_0,z)\big(\Theta_{t}^{\lin}-\Theta_0\big),
\end{align}
where $\Theta_{t}^{\lin}$ follows the linearized gradient flow. The linearized mean square loss functional is given by
\begin{align}
\L^{\lin}(\Theta_{t}^{\lin})
=\frac{1}{2}\|\,\U^{\lin}(\Theta_{t}^{\lin},Z)-Y\|_2^2.
\end{align}
The corresponding gradient flow for the parameters reads
\begin{align}\label{grad_lin_params}
\frac{\de \Theta_{t}^{\lin}}{\de t}
= -\eta\,\nabla_{\Theta}\L^{\lin}(\Theta_{t}^{\lin})
= -\eta\, \nabla_{\Theta}\U(\Theta_{0},Z^{T})\big(\U^{\lin}(\Theta_{t}^{\lin},Z)-Y\big).
\end{align}
By differentiating \eqref{Ulin_def} with respect to time and using
\eqref{grad_lin_params}, we obtain
\begin{align}
\frac{\de}{\de t}\U^{\lin}(\Theta_{t}^{\lin},z)
&= \nabla_{\Theta}\U(\Theta_0,z)\,\frac{\de \Theta_{t}^{\lin}}{\de t} \\
&= -\eta\, \nabla_{\Theta}\U(\Theta_0,z)[\nabla_{\Theta}\A u(\Theta_{0},X^{T}), \nabla_{\Theta} u(\Theta_{0},\widehat{X}^{T})]\big(\U^{\lin}(\Theta_{t}^{\lin},Z)-Y\big),\\
&= -\eta b_{K}(m)\, \widehat{K}_{\Theta_0}(z,Z^{T})\big(\U^{\lin}(\Theta_{t}^{\lin},Z)-Y\big).
\end{align}
Setting
\begin{align}
\cE(t)\coloneqq \U^{\lin}(\Theta_{t}^{\lin},Z)-Y,   
\end{align}
we obtain the linear ODE system
\begin{align}
\frac{\de}{\de t}\cE(t) = -\eta\,b_{K}(m) \widehat{K}_{\Theta_0} \cE(t),
\qquad
\cE(0)=\,\U(\Theta_0,Z)-Y.
\end{align}
The explicit solution of the linearized dynamics is
\begin{align}\label{Ulin_solution}
\U^{\lin}(\Theta_{t}^{\lin},Z)
= Y + e^{-\eta b_{K}(m)\widehat{K}_{\Theta_0} t}\big(\U(\Theta_0,Z)-Y\big).
\end{align}
Hence, 

\begin{align}
\begin{aligned}
\frac{\de}{\de t}\U^{\lin}(\Theta_{t}^{\lin},z)&= -\eta b_{K}(m)\, \widehat{K}_{\Theta_0}(z,Z^{T})e^{-\eta b_{K}(m)\widehat{K}_{\Theta_0} t}\big(\U(\Theta_0,Z)-Y\big).
\end{aligned}
\end{align}
From \eqref{grad_lin_params}, integrating in time yields
\begin{align}
\begin{aligned}
\Theta_{t}^{\lin}-\Theta_0&= -\eta\int_0^t \nabla_{\Theta}\U(\Theta_{0},Z^{T})\big(\U^{\lin}(\Theta_{t}^{\lin},Z)-Y\big)\\
&=-\eta\int_0^t \nabla_{\Theta}\U(\Theta_{0},Z^{T})e^{-\eta b_{K}(m)\widehat{K}_{\Theta_0} t}\big(\U(\Theta_0,Z)-Y\big)
\end{aligned}
\end{align}
If $\widehat{K}_{\Theta_0}$ is invertible, we obtain
\begin{align}\label{Theta_lin_solution}
\Theta_{t}^{\lin}
= \Theta_0- \frac{1}{b_{K}(m)}\nabla_{\Theta}\U(\Theta_{0},Z^{T}) \widehat{K}_{\Theta_0}^{-1}
\big(\I-e^{-\eta b_{K}(m)\widehat{K}_{\Theta_0} t}\big)\,(\U(\Theta_0,Z)-Y).
\end{align}
To make the subsequent analysis precise, we state here the hypotheses on the quantum circuit architecture, the initialization procedure, the differential operator, and the training data that are used throughout the paper. These assumptions are chosen to provide a quantitative NTK description of the training dynamics in the lazy regime, and ensure concentration of the empirical NTK around its mean, and guarantee existence, uniqueness  and regularity of the weak solution of the elliptic PDE \eqref{linearPDE}. In what follows, we make the following assumptions.
\begin{ass}\label{A0}
We assume that for some constant $\eta_{0}>0$, the learning rate $\eta$ is given by

\begin{align}
  \eta=\frac{1}{b_{K}(m)}\eta_{0}.  
\end{align}
\end{ass}
Notice that under this assumption we have that 

\begin{align}\label{Ulin_solution_2}
\U^{\lin}(\Theta_{t}^{\lin},Z)
= Y + e^{-\eta_{0}\widehat{K}_{\Theta_0} t}\big(\U(\Theta_0,Z)-Y\big).
\end{align}
As mentioned before, the role of $b_{K}(m)$ relies on the existence of a limiting kernel for $\widehat{K}_{\Theta_0}$. In general, one has that $b_{K}(m)={\rm poly}(m)$ \cite{melchor2025quantitative}. 
\begin{ass}\label{A1}
We consider an observable $\O$ given by the sum of single qubit observables $\O_{k}$:
    \begin{equation}\label{ipot1}
     \O=\sum_{k=1}^{m}\O_{k} = \O_{1}\otimes \I_{2}\otimes\cdots\otimes \I_{m} +\I_{1}\otimes \O_{2}\otimes\cdots\otimes \I_{m} +\I_{1}\otimes \I_{2}\otimes\cdots\otimes \O_{m},
    \end{equation}
    where each $\O_k$ is traceless and has the spectrum contained in the interval $[-1,1]$. We further assume that the parametric one-qubit gates of the circuit $W_{i}(\theta_{i})$ can be written as time evolutions generated by hermitian hamiltonians $\mathcal{G}_{i}$ with spectrum in $\{-1,1\}$, \emph{i.e.}, 
    \begin{align}\label{shape1}
        W_{i}(\theta_{i})=e^{-i\mathcal{G}_{i}\theta_{i}}\,,\qquad \mathcal{G}_i = \mathcal{G}_i^\dag = \mathcal{G}_i^{-1}\,.
    \end{align}
We notice that $\theta_i\mapsto W_{i}(\theta_{i})$ is periodic with period $\pi$ up to an irrelevant multiplicative constant, so we fix the parameter space to be $\mathscr{P}=[0,\pi]^{Lm}$, and thus $\vert \Theta \vert={\mathrm dim}\mathscr{P}=Lm$.
The function generated by the network is then
    \begin{align}\label{model1}
    \begin{aligned}
        u(\Theta,x)&\coloneqq \frac{1}{b(m)}\bra{0^{m}}U^{\dag}(\Theta,x)\O\,U(\Theta,x)\ket{0^{m}}\\
        &=\frac{1}{b(m)}\sum_{k=1}^{m}u_{k}(\Theta,x)
        \end{aligned}
    \end{align}
    where
    \begin{align}\label{model2}
    u_{k}(\Theta,x)\coloneqq\bra{0^{m}}U^{\dag}(\Theta,x)\O_{k}\,U(\Theta,x)\ket{0^{m}}, \hskip 0,2cm x\in \overline{B}.
    \end{align}
   \end{ass}
Here, we notice that $b(m)$ is a normalizing constant determined by the covariance function of the model at initialization. In particular, a quantum circuit suffering of the problem of barren plateaus could have a normalization $b(m)$ exponentially decaying as a function of $m$.
\begin{ass}\label{A2}
The training sets $\D_B=\{(x^{(i)},f(x^{(i)}))\}_{i=1}^{n_1}$ and
$\D_{\partial B}=\{(\hat x^{(j)},g(\hat x^{(j)}))\}_{j=1}^{n_2}$ are finite.
Denote $Z=(X,\hat X)$ the training inputs and $Y=(f(X),g(\hat X))$.
\end{ass}

\begin{ass}\label{A3}
There exists a normalization factor $b_K(m)>0$ (possibly depending on the
number of qubits $m$) such that the empirical NTK at initialization $\widehat{K}_{\Theta_0}(z,z')$ admits an expectation (analytic NTK) $K(z,z')=\mathbb E_{\Theta_0}[\widehat{K}_{\Theta_0}(z,z')]$. Furthermore, the architecture of the quantum circuit, and the normalization $b(m)$ are such that ${\rm diag}\left(\E[\U(\Theta,z)(\U(\Theta,z))^{T}]\right)$ is a positive matrix, and 
\begin{align*}
 &\max_{z\in \X_B \times \partial B}{\rm diag}\left(\E[\U(\Theta,z)(\U(\Theta,z))^{T}]\right)=1.
\end{align*}
Furthermore, we assume that the parameters $\Theta_0$ are drawn independent on $\cP$, and 
\begin{align*}
 &\mathbb{E}[u_{k}(\Theta,x)]=0, \hskip 0,2cm \text{for all $x\in \overline{B}$.}\\
\end{align*}
\end{ass}
\begin{ass}\label{A4}
We assume that the finite matrix $K\coloneqq K(Z,Z^T)\in\mathbb R^{(n_1+n_2)\times(n_1+n_2)}$ has strictly positive minimum eigenvalue $\lambda_{\min}^{K}$. We also denote by $\lambda_{\max}^{K}$ its maximum eigenvalue.
\end{ass}
Our assumptions do not fix a unique normalization, but fix such a normalization for which our assumptions hold true.
\begin{ass}[Coefficients of $\A$]\label{A5}
The coefficients $a_{ij}\in C^1(\overline B)$ satisfy
\[
\sup_{x\in\overline B}|a_{ij}(x)|\le A_0,\qquad
\sup_{x\in\overline B}|\partial_{x_k} a_{ij}(x)|\le A_1
\]
for all $i,j,k\in\{1,\dots,d\}$, with finite constants $A_0,A_1$.
\end{ass}

\begin{ass}[Spatial regularity of U]\label{A6}
The terms $V\in C^2(\overline B)$ satisfy
\begin{align*}
&
\sup_{x\in\overline B}\|\partial_{x_i} V_{\ell}(x)\|_{\rm op}\le 1,\hskip 0,2cm \sup_{x\in\overline B}\|\partial_{x_{j}}\partial_{x_i} V_{\ell}(x)\|_{\rm op}\le 1,
\end{align*}
for all $\ell\in\{0,\dots,Lm\}$, and for all $i,j\in \{1,\ldots,d\}$.
\end{ass}
These assumptions are standard in the literature on elliptic PDEs, and we adopt them because we are concerned with the existence and uniqueness of weak solutions to our PDE problem. Naturally, we can consider weaker hypotheses. However, since we need to find precise scalings to prove the convergence of $\U(\Theta,z)$ towards a Gaussian process, and since the operator $\A$ does not depend on $m$, we assume these uniform estimates. Let us give some comments about the scaling factor $b(m)$. In this work, we provide quantitative bounds weher the dependence on the number of qubits will appear as
\begin{align}\label{condizione}
\left(\frac{L^\alpha m^\beta |\mathcal{M}|^\gamma|\mathcal{N}|^\delta}{b(m)}(\log b(m))^\sigma\right)^\nu
\end{align}
for some $\alpha,\beta,\gamma,\delta,\sigma,\nu>0$ according to the corresponding statement. As discussed in \cite{girardi2025,melchor2025quantitative}, in some cases it is possible to estimate $b(m),|\mathcal{M}|$ and $|\mathcal{N}|$ so that the asymptotical behavior of the bounds for wide circuits can be studied. In this work, we instead have that 
\begin{equation}
b(m)\leq \sqrt{5} d^2 L^2\, \sqrt{m\,|\mathcal M|\,|\mathcal N|}.
\end{equation}
We remark that this bound differs from those in \cite{girardi2025,melchor2025quantitative} by the multiplicative factor $\sqrt{5}d^{2}L^{2}$. In the previous works, $b(m)$ is denoted $N(m)$ and is bounded as $N(m)\leq \sqrt{m\,|\mathcal{M}|\,|\mathcal{N}|}$. The additional factor $\sqrt{5} d^2 L^2$ arises from applying the differential operator $\A$ to the quantum neural network $u(\Theta,x)$. Let us notice that, in a generic setting the number of qubits in the past light cone of any observable $O_k$ can grow as
\begin{equation}
    |\mathcal{J}_k^1|=O(2^L).
\end{equation}
Under the hypothesis of geometrical locality (i.e., each qubit can interact only with the nearest neighbor qubits), a $d'$-dimensional lattice of qubits has
\begin{equation}
    |\mathcal{J}_k^1|=O(L^{d'}).
\end{equation} 
Suppose that $L=\epsilon\log_2 m$. Without assumptions on the geometrical locality, $|\mathcal{J}_k^1|=O(m^{\epsilon})$, which is an upper bound, so any growth $|\mathcal{J}_k^1|=\Theta(m^{\epsilon'})$ with $\epsilon'\leq\epsilon$ can be achieved by an appropriate choice of the interactions. Let us notice that, the dependence on the number of qubits will provide a prefactor \eqref{condizione} asymptotically vanishing provided that $\epsilon,\epsilon'$ are small enough. In fact, one has that \eqref{condizione} behaves as
\begin{equation}\label{eq:limit}
    \lim_{m\to\infty}\left(\frac{(\log_2m)^\alpha m^{\beta+(\gamma+\delta)\epsilon'}}{m^{1/2-C\epsilon}}\left(\log m\right)^\sigma\right)^\nu=0
\end{equation} 
for some $C>0$, and some $\beta<1/2$. However, the local Hilbert spaces have a quasi-exponential dimension $2^{m^{\epsilon'}}$. Furthermore, in the geometrically local setting, we can choose $|\mathcal{J}_k^1|=\Theta((\epsilon\log_2m)^{d'})$; if $d'\geq 2$, then the local Hilbert spaces have a super-polynomial dimension, since $2^{(\epsilon \log_2 m)^{d'}}=m^{ \epsilon^{d'}\log^{d'-1}_2 m}$. Besides, \eqref{condizione} is asymptotically vanishing for $\epsilon$ small enough.\\
We should mention that the exponential decrease of $b(m)$ on the number of layers is needed because the training of quantum neural networks can suffer from bad local minima or gradients whose size decreases exponentially with the number of qubits, a phenomenon called \textit{barren plateaus} \cite{napp2022quantifying}. 

\section{NTK concentration for QPINN}\label{sec:concentration}
In this section, we establish a concentration bound for the NTK of a QPINN. We begin by stating preliminary bounds that allow us to compute the precise Lipschitz constants for the QPINN model, its gradient, and its NTK. Our approach closely follows the methodology developed in \cite{girardi2025,melchor2025quantitative}; 
however, we derive the explicit constants required for the concentration bound stated in \autoref{thm:NTK}. To this end, we establish precise bounds for the differential operator applied to the quantum neural network $u(\Theta,x)$.
\begin{lem}\label{lem:pointwise_derivatives}
Assume that \autoref{A0}--\autoref{A2}, and \autoref{A4}--\autoref{A5} hold true. Then
\begin{align}
\label{pt:0}
&|\A u_{k}(\Theta,x)|\le 2d^{2}(2A_{1}+4A_{0})L^2,\\
\label{pt:1}
&\left\vert \partial_{\theta_{s}}\A u_{k}(\Theta,x)\right\vert \le 4d^{2}(2A_{1}+4A_{0})L^2,\\\label{pt:2}
&\left\vert \partial_{\theta_{s'}}\partial_{\theta_{s}}\A u_{k}(\Theta,x)\right\vert \le 8d^{2}(2A_{1}+4A_{0})L^2,
\end{align}
hold true.
\end{lem}

\begin{proof}
Recall that 
\begin{align*}
u_{k}(\Theta,x)=\bra{0^{m}}U^{\dagger}(\Theta,x)\O_{k}U(\Theta,x)\ket{0^{m}}.    
\end{align*}
Then

\begin{align*}
\frac{\partial u_{k}(\Theta,x)}{\partial x_{i}}= \bra{0^{m}}\frac{\partial U^{\dagger}(\Theta,x)}{\partial x_{i}}\O_{k}U(\Theta,x)\ket{0^{m}} + \bra{0^{m}}U^{\dagger}(\Theta,x)\O_{k}\frac{\partial U(\Theta,x)}{\partial x_{i}}\ket{0^{m}},
\end{align*}
and 
\begin{align}\label{derivativesmixed}
\begin{aligned}
\frac{\partial}{\partial x_{j}}\left(a_{ij}(x)\frac{\partial u_{k}(\Theta,x)}{\partial x_{i}}\right)&=\frac{\partial a_{ij}(x)}{\partial x_{j}}\bra{0^{m}}\frac{\partial U^{\dagger}(\Theta,x)}{\partial x_{i}}\O_{k}U(\Theta,x)\ket{0^{m}}\\
&+   a_{ij}(x)\bra{0^{m}}\frac{\partial^{2} U^{\dagger}(\Theta,x)}{\partial x_{j}\partial x_{i}}\O_{k}U(\Theta,x)\ket{0^{m}}+\\
&+a_{ij}(x)\bra{0^{m}}\frac{\partial U^{\dagger}(\Theta,x)}{\partial x_{i}}\O_{k}\frac{\partial U(\Theta,x)}{\partial x_{j}}\ket{0^{m}}+\\
&+\frac{\partial}{\partial x_{j}}\left(a_{ij}(x)\bra{0^{m}}U^{\dagger}(\Theta,x)\O_{k}\frac{\partial U(\Theta,x)}{\partial x_{i}}\ket{0^{m}}\right).
\end{aligned}
\end{align}
Notice that 

\begin{align*}
\frac{\partial U(\Theta,x)}{\partial x_{i}}=\sum_{\ell=1}^{L}V_{L}(x)W_{L}(\Theta)\cdots \frac{\partial V_{\ell}(x)}{\partial x_{i}}W_{\ell}(\Theta)\cdots V_{1}(x)W_{1}(\Theta),
\end{align*}
and by our hypotheses \autoref{A1}, and \autoref{A6}, we find that 

\begin{align*}
\left\vert \frac{\partial U(\Theta,x)}{\partial x_{i}}\right\vert\leq L.
\end{align*}
In the same manner, one has 

\begin{align*}
\frac{\partial^{2} U(\Theta,x)}{\partial x_{j}\partial x_{i}}=\sum_{\ell, \ell'=1}^{L}V_{L}(x)W_{L}(\Theta)\cdots \frac{\partial V_{\ell}(x)}{\partial x_{i}}W_{\ell}(\Theta)\cdots\frac{\partial V_{\ell'}(x)}{\partial x_{j}}W_{\ell'}(\Theta)\cdots V_{1}(x)W_{1}(\Theta),
\end{align*}
and thus
\begin{align*}
\left\vert \frac{\partial^{2} U(\Theta,x)}{\partial x_{j}\partial x_{i}}\right\vert\leq L^2.
\end{align*}
In combination with \autoref{A5}, we get that

\begin{align*}
\left\vert \frac{\partial}{\partial x_{j}}\left(a_{ij}(x)\frac{\partial u_{k}(\Theta,x)}{\partial x_{i}}\right) \right\vert&\leq 2\norma{\frac{\partial a_{ij}}{\partial x_{j}}}\norma{\frac{\partial U}{\partial x_{i}}}\norma{U}\norma{\O_{k}}+2\norma{a_{ij}}\norma{\frac{\partial^{2} U}{\partial x_{j}\partial x_{i}}}\norma{U}\norma{\O_{k}}+\\
&+2\norma{a_{ij}}\norma{\frac{\partial U}{\partial x_{i}}}\norma{\frac{\partial U}{\partial x_{j}}}\norma{\O_{k}}\\
&\leq 2A_{1}L+ 2A_{0}L^{2}+ 2A_{0}L^{2}\\
&\leq 2A_{1}L^2+4A_{0}L^2.
\end{align*}
Therefore, 

\begin{align*}
|\A u_{k}(\Theta,x)|\le 2d^{2}(2A_{1}+4A_{0})L^2.
\end{align*}
On the other hand, notice that 

\begin{align}\label{mixedder}
\begin{aligned}
\frac{\partial}{\partial \theta_{s}}\frac{\partial}{\partial x_{j}}\left(a_{ij}(x)\frac{\partial u_{k}(\Theta,x)}{\partial x_{i}}\right)&=
\frac{\partial}{\partial x_{j}}\left(a_{ij}(x)\frac{\partial^{2} u_{k}(\Theta,x)}{\partial x_{i}\partial \theta_{s}}\right).
\end{aligned}
\end{align}
for all $s=1,\ldots,Lm$. Further, by the shift-parameter rule (see e.g. \cite{banchi2021measuring}) one has

\begin{align}\label{shiftrule}
 \frac{\partial u_{k}(\Theta,x)}{\partial \theta_{s}}=u_{k}(\Theta+\Delta^{(s)},x)-u_{k}(\Theta-\Delta^{(s)},x), \hskip0,2cm
\Delta_{j}^{(s)}\coloneqq
\begin{cases}
&\frac{\pi}{4}\, \hskip 0,2cm \text{if $j=s$,}\\
&0 \hskip 0,2cm \text{otherwise.}
\end{cases}
\end{align}
Therefore,

\begin{align}\label{doubleoperator1}
\partial_{\theta_{s}}\A u_{k}(\Theta,x)=\A u_{k}(\Theta+\Delta^{(s)})-\A u_{k}(\Theta-\Delta^{(s)}).  
\end{align}
By using \autoref{doubleoperator1}, and \autoref{pt:0}, we get
\begin{align*}
\vert \partial_{\theta_{s}}\A u_{k}(\Theta,x)\vert\leq 4d^{2}(2A_{1}+4A_{0})L^2.  
\end{align*}
Analogously, we have that 

\begin{align*}
\partial_{\theta_{s'}}\partial_{\theta_{s}}\A u_{k}(\Theta,x)=\A \partial_{\theta_{s'}}u_{k}(\Theta+\Delta^{(s)})-\A \partial_{\theta_{s'}}u_{k}(\Theta-\Delta^{(s)}).  
\end{align*}
Then, by applying again the shift-parameter rule,
\begin{align}\label{doubleoperator2}
\begin{aligned}
\partial_{\theta_{s'}}\partial_{\theta_{s}}\A u_{k}(\Theta,x)&=  \A u_{k}(\Theta+ \Delta^{(s)}+\Delta^{(s')})-\A u_{k}(\Theta+ \Delta^{(s)}-\Delta^{(s')})\\
&-\left(\A u_{k}(\Theta-\Delta^{(s)}+\Delta^{(s')})-\A u_{k}(\Theta-\Delta^{(s)}-\Delta^{(s')})\right),
\end{aligned}
\end{align}
and from here,

\begin{align*}
\vert \partial_{\theta_{s'}}\partial_{\theta_{s}}\A u_{k}(\Theta,x)\vert\leq 8d^{2}(2A_{1}+4A_{0})L^2.
\end{align*}
\end{proof}
\begin{lem}\label{thm:lipschitz_model}
Assume that \autoref{A0}--\autoref{A2}, and \autoref{A4}--\autoref{A5} hold true. Then
\begin{align}\label{lip:A}
\norma{\nabla_{\Theta}\A u(\Theta,x)-\nabla_{\Theta}\A u(\Theta',x)}_{\infty}&\le \frac{8d^{2}(2A_{1}+4A_{0})L^2\vert \mathcal{M}\vert^{2}\vert\mathcal{N}\vert}{b(m)}\norma{\Theta-\Theta'}_{\infty},
\end{align}
holds uniformly in $x\in B$.
\end{lem}

\begin{proof}
Let us set $h_{x}(\Theta)\coloneqq\nabla_{\Theta}\A u(\Theta,x)$, and consider its differential $\de h_{x}(\Theta): (\R^{\vert \Theta\vert},\ell^{\infty})\rightarrow (\R^{\vert\Theta\vert},\ell^{\infty})$ with operator norm given by

\begin{align*}
\norma{\de h_{x}(\Theta)}_{\ell^{\infty}\rightarrow\ell^{\infty}}= \sup_{\norma{v}_{\infty}\leq 1}\norma{\de h_{x}(\Theta)v}_{\infty}. 
\end{align*}
Let us observe that 

\begin{align}\label{seconderivatives}
\begin{aligned}
\sup_{\norma{v}_{\infty}\leq 1}\norma{\de h_{x}(\Theta)v}_{\infty}&=\sup_{\norma{v}_{\infty}\leq 1}\norma{\sum_{s=1}^{\vert \Theta\vert}\partial_{\theta_{s}}\nabla_{\Theta} \A u(\Theta,x) v_{s}}_{\infty}\\
&\leq\sup_{\vert v_{s}\vert\leq 1}\max_{1\leq s'\leq \vert \Theta\vert}\sum_{s=1}^{\vert\Theta\vert}\left\vert \partial _{\theta_{s}}\partial_{\theta_{s'}}\A u(\Theta,x) v_{s}\right\vert\\
&\leq\max_{1\leq s'\leq \vert \Theta\vert}\sum_{s=1}^{\vert\Theta\vert}\frac{1}{b(m)}\sum_{k\in \mathcal{M}_{s}\cap \mathcal{M}_{s'}}\left\vert \partial _{\theta_{s}}\partial_{\theta_{s'}}\A u_{k}(\Theta,x)\right\vert\\
&\leq \frac{8d^{2}(2A_{1}+4A_{0})L^2}{b(m)}\max_{1\leq s'\leq \vert \Theta\vert}\sum_{s=1}^{\vert \Theta\vert} \vert \mathcal{M}_{s}\cap \mathcal{M}_{s'}\vert.
\end{aligned}
\end{align}
Applying \cite[Lemma 4.19]{girardi2025}, we have that

\begin{align}\label{stimafilippo}
\max_{1\leq s'\leq \vert \Theta\vert}\sum_{s=1}^{\vert \Theta\vert} \vert \mathcal{M}_{s}\cap \mathcal{M}_{s'}\vert\leq \vert \mathcal{M}\vert^{2}\vert\mathcal{N}\vert.
\end{align}
Therefore,
\begin{align}
\norma{\de h_{x}(\Theta)}_{\ell^{\infty}\rightarrow\ell^{\infty}}\le\frac{8d^{2}(2A_{1}+4A_{0})L^2\vert \mathcal{M}\vert^{2}\vert\mathcal{N}\vert}{b(m)}.
\end{align}
Since the domain $\cP$ of $\Theta$ is convex, then we have that 
\begin{align}
\norma{\nabla_{\Theta}\A u(\Theta,x)-\nabla_{\Theta}\A u(\Theta',x)}_{\infty}&\le \frac{8d^{2}(2A_{1}+4A_{0})L^2\vert \mathcal{M}\vert^{2}\vert\mathcal{N}\vert}{b(m)}\norma{\Theta-\Theta'}_{\infty}.
\end{align}
\end{proof}

\begin{lem}\label{lem:PDEres}
Assume that \autoref{A0}--\autoref{A2}, and \autoref{A4}--\autoref{A5} hold true. The following holds true:
\begin{align}
\left\vert \A u(\Theta,x)-\A u(\Theta',x)\right\vert\le \frac{4d^{2}(2A_{1}+4A_{0})mL^3\vert \mathcal{M}\vert}{b(m)}\norma{\Theta-\Theta'}_{\infty}.    
\end{align}
\end{lem}
\begin{proof}
Notice that 

\begin{align}
\left\vert \A u(\Theta,x)-\A u(\Theta',x)\right\vert \le \max_{\Theta\in \cP}\norma{\de \A u(\Theta,x)}_{\mathcal{L}}\norma{\Theta-\Theta'}_{\infty},  
\end{align}
with 
\begin{align*}
\max_{\Theta\in \cP}\norma{\de \A u(\Theta,x)}_{\mathcal{L}}&=\sup_{\norma{v}_{\infty}\leq 1}\left\vert\sum_{s=1}^{\vert \Theta\vert}\partial_{\theta_{s}}\A u(\Theta,x)v_{s}\right\vert\\
&\leq \sum_{s=1}^{\vert \Theta\vert}\sum_{k\in \mathcal{M}_{s}}\frac{1}{b(m)}\left\vert \partial_{\theta_{s}}\A u_{k}(\Theta,x)\right\vert\\
&\leq \frac{4d^{2}(2A_{1}+4A_{0})L^2}{b(m)}\sum_{s=1}^{\vert \Theta\vert}\vert \mathcal{M}_{s}\vert\\
&\leq \frac{4d^{2}(2A_{1}+4A_{0})mL^3\vert \mathcal{M}\vert}{b(m)}.
\end{align*}
\end{proof}
In what follows, based on the previous bounds, we state the Lipschitzness of the NTK for the QIPNN.
\begin{lem}[Lipschitzness of the NTK]\label{lem:lipschitzntk}
Assume that \autoref{A0}--\autoref{A2}, and \autoref{A4}--\autoref{A5} hold true. Let $z=(x,\hat{x})$, and $z'=(x',\hat{x}')$. The following holds true:
\begin{align}
\norma{ \hat{K}_{\Theta'}(z,z')-\hat{K}_{\Theta}(z,z')}_{F}\leq \frac{1}{b_{K}(m)}\left(\frac{128d^{4}(2A_{1}+4A_{0})^{2}mL^{5}\vert\mathcal{M}\vert^{3}\vert\mathcal{N}\vert}{(b(m))^{2}}+\frac{16Lm\vert\mathcal{M}\vert^{3}\vert\mathcal{N}\vert}{(b(m))^{2}}\right)\norma{\Theta-\Theta'}_{\infty}.
\end{align}
\end{lem}
\begin{proof}
  We have that 

  \begin{align}\label{dlimitare}
   \begin{aligned}
      \norma{\hat{K}_{\Theta'}(z,z')-\hat{K}_{\Theta}(z,z')}_{F}^{2}&=\frac{1}{(b_{K}(m))^{2}}\left\vert \nabla_{\Theta}\A u(\Theta,x)\cdot \nabla_{\Theta}\A u(\Theta,x')-\nabla_{\Theta}\A u(\Theta',x)\cdot \nabla_{\Theta}\A u(\Theta',x')\right\vert^{2}\\
      &+\frac{1}{(b_{K}(m))^{2}}\left\vert \nabla_{\Theta}\A u(\Theta,x)\cdot \nabla_{\Theta}u(\Theta,\hat{x}')-\nabla_{\Theta}\A u(\Theta',x)\cdot \nabla_{\Theta}u(\Theta',\hat{x}')\right\vert^{2}\\
      &+\frac{1}{(b_{K}(m))^{2}}\left\vert \nabla_{\Theta}u(\Theta,\hat{x})\cdot \nabla_{\Theta}\A u(\Theta,x')-\nabla_{\Theta} u(\Theta',\hat{x})\cdot \nabla_{\Theta}\A u(\Theta',x')\right\vert^{2}\\
      &+\frac{1}{(b_{K}(m))^{2}}\left\vert \nabla_{\Theta}u(\Theta,\hat{x})\cdot \nabla_{\Theta} u(\Theta,\hat{x}')-\nabla_{\Theta} u(\Theta',\hat{x})\cdot \nabla_{\Theta}u(\Theta',\hat{x}')\right\vert^{2}.
   \end{aligned}   
  \end{align}
Let us now give a bound for each term of \eqref{dlimitare}. Notice that

\begin{align}
\begin{aligned}
 &\left\vert \nabla_{\Theta}\A u(\Theta,x)\cdot \nabla_{\Theta}\A u(\Theta,x')-\nabla_{\Theta}\A u(\Theta',x)\cdot \nabla_{\Theta}\A u(\Theta',x')\right\vert \\
 &=\left\vert \left(\nabla_{\Theta}\A u(\Theta',x)-\nabla_{\Theta}\A u(\Theta,x)\right)\cdot \nabla_{\Theta}\A u(\Theta',x')+ \right.\\
    &\left. +\nabla_{\Theta}\A u(\Theta,x)\cdot\left(\nabla_{\Theta}\A u(\Theta',x')-\nabla_{\Theta}\A u(\Theta,x')\right)\right\vert 
\end{aligned}    
\end{align}
and thus

  \begin{align}
   \begin{aligned}
    &\left\vert \left(\nabla_{\Theta}\A u(\Theta',x)-\nabla_{\Theta}\A u(\Theta,x)\right)\cdot \nabla_{\Theta}\A u(\Theta',x')+\nabla_{\Theta}\A u(\Theta,x)\cdot\left(\nabla_{\Theta}\A u(\Theta',x')-\nabla_{\Theta}\A u(\Theta,x')\right)\right\vert\\ 
    &\quad\leq \norma{\nabla_{\Theta}\A u(\Theta',x)-\nabla_{\Theta}\A u(\Theta,x)}_{\infty}\norma{ \nabla_{\Theta}\A u(\Theta',x')}_{1}\\
    &\quad\quad+\norma{\nabla_{\Theta}\A u(\Theta',x')-\nabla_{\Theta}\A u(\theta,x')}_{\infty}\norma{\nabla_{\Theta}\A u(\Theta,x)}_{1}\\
    &\leq \frac{8d^{2}(2A_{1}+4A_{0})L^2\vert \mathcal{M}\vert^{2}\vert\mathcal{N}\vert}{b(m)}\norma{\Theta-\Theta'}_{\infty}\left(\norma{\nabla_{\Theta}\A u(\Theta,x)}_{1}+\norma{ \nabla_{\Theta}\A u(\Theta',x')}_{1}\right).
   \end{aligned}   
  \end{align}
  Notice that 

  \begin{align}
  \begin{aligned}
    \norma{\nabla_{\Theta}\A u(\Theta,x)}_{1}&=\sum_{s=1}^{Lm}\vert\partial_{\theta_{s}}\A u(\Theta,x) \vert\\
    &\leq \frac{1}{b(m)}\sum_{s=1}^{Lm}\sum_{k\in \mathcal{M}_{s}}\vert\partial_{\theta_{s}}\A u_{k}(\Theta,x) \vert\\
    &\leq \frac{4d^{2}(2A_{1}+4A_{0})mL^{3}\vert\mathcal{M}\vert}{b(m)}.
    \end{aligned}
  \end{align}
  Thus,

  \begin{align}
  \begin{aligned}
    &\left\vert \left(\nabla_{\Theta}\A u(\Theta',x)-\nabla_{\Theta}\A u(\Theta,x)\right)\cdot \nabla_{\Theta}\A u(\Theta',x')+\nabla_{\Theta}\A u(\Theta,x)\cdot\left(\nabla_{\Theta}\A u(\Theta',x')-\nabla_{\Theta}\A u(\theta,x')\right)\right\vert\\ 
    &\quad\quad\leq\frac{64d^{4}(2A_{1}+4A_{0})^{2}mL^{5}\vert\mathcal{M}\vert^{3}\vert\mathcal{N}\vert}{(b(m))^{2}}.
  \end{aligned}    
  \end{align}
  On the other hand, by following the proof of \cite[Lemma 4.22]{girardi2025}, one has
  \begin{align}
  \begin{aligned}
&\left\vert \left(\nabla_{\Theta}u(\Theta',\hat{x})-\nabla_{\Theta}u(\Theta,\hat{x})\right)\cdot \nabla_{\Theta} u(\Theta',\hat{x}') +\nabla_{\Theta}u(\Theta,\hat{x})\cdot\left(\nabla_{\Theta}u(\Theta',\hat{x}')-\nabla_{\Theta} u(\theta,\hat{x}')\right)\right\vert\\
&\quad\quad\leq \frac{16Lm\vert\mathcal{M}\vert^{3}\vert\mathcal{N}\vert}{(b(m))^{2}}.  
  \end{aligned}
  \end{align}
 Now, notice that

\begin{align}
\begin{aligned}
 &\left\vert \nabla_{\Theta}\A u(\Theta,x)\cdot \nabla_{\Theta}u(\Theta,\hat{x}')-\nabla_{\Theta}\A u(\Theta',x)\cdot \nabla_{\Theta} u(\Theta',\hat{x}')\right\vert \\
 &=\left\vert \left(\nabla_{\Theta}\A u(\Theta',x)-\nabla_{\Theta}\A u(\Theta,x)\right)\cdot \nabla_{\Theta} u(\Theta',\hat{x}')+ \right.\\
    &\left. +\nabla_{\Theta}\A u(\Theta,x)\cdot\left(\nabla_{\Theta}u(\Theta',\hat{x}')-\nabla_{\Theta}u(\theta,\hat{x}')\right)\right\vert \\
    &\leq \norma{\nabla_{\Theta}\A u(\Theta',x)-\nabla_{\Theta}\A u(\Theta,x)}_{\infty}\norma{\nabla_{\Theta} u(\Theta',\hat{x}')}_{1}+\\
    &+\norma{\nabla_{\Theta}u(\Theta',\hat{x}')-\nabla_{\Theta}u(\theta,\hat{x}')}_{\infty}\norma{\nabla_{\Theta}\A u(\Theta,x)}_{1}\\
    &\leq \frac{8d^{2}(2A_{1}+4A_{0})L^2\vert \mathcal{M}\vert^{2}\vert\mathcal{N}\vert}{b(m)}\norma{\Theta-\Theta'}_{\infty}\norma{\nabla_{\Theta} u(\Theta',\hat{x}')}_{1}+\\
    &+\norma{\nabla_{\Theta}u(\Theta',\hat{x}')-\nabla_{\Theta}u(\theta,\hat{x}')}_{\infty}\frac{4d^{2}(2A_{1}+4A_{0})mL^{3}\vert\mathcal{M}\vert}{b(m)}.
\end{aligned}    
\end{align}
 
Recall that 

\begin{align}
\begin{aligned}
&\norma{\nabla_{\Theta} u(\Theta',\hat{x}')}_{1}\leq \frac{2Lm\vert\mathcal{M}\vert}{b(m)}\\
&\norma{\nabla_{\Theta}u(\Theta',\hat{x}')-\nabla_{\Theta}u(\Theta,\hat{x}')}_{\infty}\leq \frac{4\vert\mathcal{M}\vert^2\vert \mathcal{N}\vert}{b(m)}\norma{\Theta-\Theta'}_{\infty}.
\end{aligned}
\end{align}
Then we have that

\begin{align}
\begin{aligned}
&\left\vert \nabla_{\Theta}\A u(\Theta,x)\cdot \nabla_{\Theta}u(\Theta,\hat{x}')-\nabla_{\Theta}\A u(\Theta',x)\cdot \nabla_{\Theta} u(\Theta',\hat{x}')\right\vert \\
&\leq \frac{16d^{2}(2A_{1}+4A_{0})mL^3\vert \mathcal{M}\vert^{3}\vert\mathcal{N}\vert}{(b(m))^{2}}\norma{\Theta-\Theta'}_{\infty}+\\
    &+\frac{16d^{2}(2A_{1}+4A_{0})mL^{3}\vert\mathcal{M}\vert^{3}\vert \mathcal{N}\vert}{(b(m))^{2}}\norma{\Theta-\Theta'}_{\infty}\\
    &=\frac{32d^{2}(2A_{1}+4A_{0})mL^{3}\vert\mathcal{M}\vert^{3}\vert \mathcal{N}\vert}{(b(m))^{2}}\norma{\Theta-\Theta'}_{\infty}
\end{aligned}   
\end{align}
We conclude that

  \begin{align}
  \begin{aligned}
  \norma{\hat{K}_{\Theta'}(z,z')-\hat{K}_{\Theta}(z,z')}_{F}^{2}
  &\leq \frac{1}{(b_{K}(m))^{2}}\left(\left(\frac{64d^{4}(2A_{1}+4A_{0})^{2}mL^{5}\vert\mathcal{M}\vert^{3}\vert\mathcal{N}\vert}{(b(m))^{2}}\right)^{2}+\left(\frac{16Lm\vert\mathcal{M}\vert^{3}\vert\mathcal{N}\vert}{(b(m))^{2}}\right)^{2}+\right.\\
  &\left.+2\left(\frac{32d^{2}(2A_{1}+4A_{0})mL^{3}\vert\mathcal{M}\vert^{3}\vert \mathcal{N}\vert}{(b(m))^{2}}\right)^{2}\right)\norma{\Theta-\Theta'}_{\infty}^{2}.
  \end{aligned}  
  \end{align}
  Therefore, since $\sqrt{a+b}\leq \sqrt{a}+\sqrt{b}$ for all $a,b\geq 0$, we get
\begin{align}
  \begin{aligned}
  \norma{\hat{K}_{\Theta'}(z,z')-\hat{K}_{\Theta}(z,z')}_{F}
  &\leq \frac{1}{b_{K}(m)}\left(\frac{64d^{4}(2A_{1}+4A_{0})^{2}mL^{5}\vert\mathcal{M}\vert^{3}\vert\mathcal{N}\vert}{(b(m))^{2}}+\frac{16Lm\vert\mathcal{M}\vert^{3}\vert\mathcal{N}\vert}{(b(m))^{2}}+\right.\\
  &\left.+\frac{(\sqrt{2})32d^{2}(2A_{1}+4A_{0})mL^{3}\vert\mathcal{M}\vert^{3}\vert \mathcal{N}\vert}{(b(m))^{2}}\right)\norma{\Theta-\Theta'}_{\infty}.
  \end{aligned}  
  \end{align}
  and then

  \begin{align}
  \begin{aligned}
  \norma{\hat{K}_{\Theta'}(z,z')-\hat{K}_{\Theta}(z,z')}_{F}
  &\leq \frac{1}{b_{K}(m)}\left(\frac{128d^{4}(2A_{1}+4A_{0})^{2}mL^{5}\vert\mathcal{M}\vert^{3}\vert\mathcal{N}\vert}{(b(m))^{2}}+\frac{16Lm\vert\mathcal{M}\vert^{3}\vert\mathcal{N}\vert}{(b(m))^{2}}\right)\norma{\Theta-\Theta'}_{\infty}.
  \end{aligned}  
  \end{align}
\end{proof}
In what follows, to reduce the dependence on the constants stated in \autoref{A6}, and since the coefficients of the differential operator does not depend on $m$, we assume that $2A_{1}+4A_{0}\leq 1$.
\begin{thm}[NTK concentration]\label{thm:NTK}
Let $B\subset \R^{d}$ be an open bounded subset of $\R^{d}$ with Lipschitz boundary $\partial B$, and $d\geq 1$ the dimensionality of $\R^{d}$. Assume that hypotheses \autoref{A1}--\autoref{A6} hold true, and that the constants $A_{0},A_{1}$ satisfy

\begin{align}\label{dimensionalconstraint}
2A_{1}+4A_{0}\leq 1.
\end{align}
Then for any $z,z'\in B\times \partial B$, it holds that
\begin{align}
\P\left[ \norma{\hat{K}_{\Theta}(z,z')-K(z,z')}_{F}\right]\geq \varepsilon]\leq \exp\left[\frac{-\varepsilon^{2}(b_{K}(m))^{2}(b(m))^{4}}{4(64)^{2}mL^{9}d^{8}\vert\mathcal{M}\vert^{4}\vert\mathcal{N}\vert^{2}}\right].
\end{align}
Furthermore, assume that
\begin{align}\label{ipoconvergencentk}
\lim_{m\rightarrow+\infty}\frac{mL^{9}\vert\mathcal{M}\vert^{4}\vert\mathcal{N}\vert^{2}}{(b(m))^{4}}=0.    
\end{align}
Then the NTK converges in probability to the analytic NTK as $m\rightarrow +\infty$.
\end{thm}
\begin{proof}
Let us notice that 

\begin{align}
\P\left[\varepsilon\leq\norma{\hat{K}_{\Theta}(z,z')-K(z,z')}_{F}\right]=\P\left[\varepsilon^{2}\leq\norma{\hat{K}_{\Theta}(z,z')-K(z,z')}_{F}^{2}\right].
\end{align}
Since
\begin{align}
  \norma{\hat{K}_{\Theta}(z,z')-K(z,z')}_{F}^{2}=\sum_{i,j=1}^{2}\left((\hat{K}_{\Theta}(z,z'))_{i,j}-(K(z,z'))_{i,j}\right)^{2}  
\end{align}
we have that 

\begin{align}\label{bounddafa}
\P\left[\varepsilon^{2}\leq\norma{\hat{K}_{\Theta}(z,z')-K(z,z')}_{F}^{2}\right]\leq    \sum_{i,j=1}^{2}\P\left[\frac{\varepsilon^{2}}{4}\leq \left((\hat{K}_{\Theta}(z,z'))_{i,j}-(K(z,z'))_{i,j}\right)^{2}\right].
\end{align}

In what follows, we aim to estimate the right-hand side of \eqref{bounddafa}. To obtain an estimate, we then use the McDiarmid's concentration inequality. First, observe that

\begin{align}
\begin{aligned}
(\hat{K}_{\Theta}(z,z'))_{11}&= \frac{1}{b_{K}(m)}\nabla_{\Theta}\A u(\Theta,x)\cdot\nabla_{\Theta}\A u(\Theta,x)\\
&=\frac{1}{b_{K}(m) (b(m))^{2}}\sum_{k,k'=1}^{m}\sum_{s=1}^{\vert\Theta\vert}\partial_{\theta_{s}}\A u_{k}(\Theta,x)\partial_{\theta_{s}}\A u_{k'}(\Theta,x)\\
&=\frac{1}{b_{K}(m) (b(m))^{2}}\sum_{k,k'=1}^{m}\sum_{s\in \mathcal{N}_{k}\cap\mathcal{N}_{k'}}\partial_{\theta_{s}}\A u_{k}(\Theta,x)\partial_{\theta_{s}}\A u_{k'}(\Theta,x),
\end{aligned}
\end{align}
\begin{align}
\begin{aligned}
(\hat{K}_{\Theta}(z,z'))_{2,2}&=\frac{1}{b_{K}(m)}\nabla_{\Theta}u(\Theta,\hat{x})\cdot \nabla_{\Theta}u(\Theta,\hat{x})\\
&=\frac{1}{b_{K}(m) (b(m))^{2}}\sum_{k,k'=1}^{m}\sum_{s\in \mathcal{N}_{k}\cap\mathcal{N}_{k'}}\partial_{\theta_{s}}u_{k}(\Theta,\hat{x})\partial_{\theta_{s}}u_{k'}(\Theta,\hat{x}),
\end{aligned}  
\end{align}
and 

\begin{align}
\begin{aligned}
(\hat{K}_{\Theta}(z,z'))_{1,2}&=\frac{1}{b_{K}(m)}\nabla_{\Theta}\A u(\Theta,x)\cdot \nabla_{\Theta}u(\Theta,\hat{x})\\
&=\frac{1}{b_{K}(m) (b(m))^{2}}\sum_{k,k'=1}^{m}\sum_{s\in \mathcal{N}_{k}\cap\mathcal{N}_{k'}}\partial_{\theta_{s}}\A u_{k}(\Theta,x)\partial_{\theta_{s}}u_{k'}(\Theta,\hat{x}').
\end{aligned}  
\end{align}
Let us define 

\begin{align}\label{defpsi}
\Psi_{i}\coloneqq\left\{(k,k',s): s\in \mathcal{N}_{k}\cap \mathcal{N}_{k'},\, i\in \mathcal{N}_{k}\cup \mathcal{N}_{k'}\right\}.   
\end{align}
We notice that

\begin{align}
\begin{aligned}
&(\hat{K}_{\Theta}(z,z'))_{1,1}-(\hat{K}_{\Theta'}(z,z'))_{1,1}=\\
&=\frac{1}{b_{K}(m)(b(m))^{2}}\sum_{(k,k',s)\in \Psi_{i}}\left[\partial_{\theta_{s}}\A u_{k}(\Theta,x)\partial_{\theta_{s}}\A u_{k'}(\Theta,x)-\partial_{\theta_{s}}\A u_{k}(\Theta',x)\partial_{\theta_{s}}\A u_{k'}(\Theta',x)\right]. 
\end{aligned}    
\end{align}
Then, by considering \eqref{pt:1}, one gets
\begin{align}
    \begin{aligned}
    \left\vert (\hat{K}_{\Theta}(z,z'))_{1,1}-(\hat{K}_{\Theta'}(z,z'))_{1,1}\right\vert\leq \frac{2(4d^{2}(2A_{1}+4A_{0})L^2)^{2}\vert \Psi_{i}\vert}{b_{K}(m)(b(m))^{2}}.  
    \end{aligned}
\end{align}
On the other hand, since $\vert \Psi_{i}\vert\leq 2\vert\mathcal{M}_{i}\vert\vert\mathcal{M}\vert\vert\mathcal{N}\vert$, one has

\begin{align}
\begin{aligned}
\left\vert (\hat{K}_{\Theta}(z,z'))_{1,1}-(\hat{K}_{\Theta'}(z,z'))_{1,1}\right\vert&\leq \frac{\vert\mathcal{M}_{i}\vert\vert\mathcal{M}\vert\vert\mathcal{N}\vert}{b_{K}(m)(b(m))^{2}}64d^{4}(2A_{1}+4A_{0})^{2}L^{4}\\
&\leq \frac{64d^{4}L^4\vert\mathcal{M}\vert^{2}\vert\mathcal{N}\vert}{b_{K}(m)(b(m))^{2}},
\end{aligned}   
\end{align}
where in the last inequality, we have used \eqref{dimensionalconstraint}. Hence, by using the McDiarmid's concentration inequality as stated in \autoref{thm:McDiarmid}, we have
\begin{align}
c_{i}\coloneqq\frac{64d^{4}L^4\vert\mathcal{M}\vert^{2}\vert\mathcal{N}\vert}{b_{K}(m)(b(m))^{2}}, 
\end{align}
and for any $\delta>0$

\begin{align}
\P\left[\vert(\hat{K}_{\Theta}(z,z'))_{11}-\E[(\hat{K}_{\Theta}(z,z'))_{1,1}]\vert\geq \delta\right]\leq \exp\left[\frac{-2\delta^{2}}{\sum_{i=1}^{\vert \Theta\vert}c_{i}^{2}}\right].  
\end{align}
Notice that
\begin{align}
\sum_{i=1}^{\vert \Theta\vert}c_{i}^{2}= \frac{(64)^{2}mL^{9}d^{8}\vert\mathcal{M}\vert^{4}\vert\mathcal{N}\vert^{2}}{(b_{K}(m))^{2}(b(m))^{4}}.
\end{align}
Hence, we conclude that for $\delta=\frac{\varepsilon}{2}$
\begin{align}
\P\left[ \vert (\hat{K}_{\Theta}(z,z'))_{1,1}-\E[(\hat{K}_{\Theta}(z,z'))_{1,1}]\vert\geq \frac{\varepsilon}{2}\right]\leq \exp\left[\frac{-2\varepsilon^{2}(b_{K}(m))^{2}(b(m))^{4}}{4(64)^{2}mL^{9}d^{8}\vert\mathcal{M}\vert^{4}\vert\mathcal{N}\vert^{2}}\right].
\end{align}
On the other hand, let us recall that
\begin{align}
\P\left[ \vert (\hat{K}_{\Theta}(z,z'))_{2,2}-\E[(\hat{K}_{\Theta}(z,z'))_{2,2}]\vert\geq \frac{\varepsilon}{2}\right]\leq \exp\left[\frac{-2\varepsilon^{2}(b_{K}(m))^{2}(b(m))^{4}}{4(16)^{2}mL\vert\mathcal{M}\vert^{4}\vert\mathcal{N}\vert^{2}}\right].
\end{align}
Similarly, 

\begin{align}
\begin{aligned}
&(\hat{K}_{\Theta}(z,z'))_{1,2}-(\hat{K}_{\Theta'}(z,z'))_{1,2}=\\
&=\frac{1}{b_{K}(m)(b(m))^{2}}\sum_{(k,k',s)\in \Psi_{i}}\left[\partial_{\theta_{s}}\A u_{k}(\Theta,x)\partial_{\theta_{s}}u_{k'}(\Theta,\hat{x}')-\partial_{\theta_{s}}\A u_{k}(\Theta',x)\partial_{\theta_{s}}u_{k'}(\Theta',\hat{x}')\right]. 
\end{aligned}    
\end{align}
Then
\begin{align}\label{midterm}
    \begin{aligned}
    \left\vert (\hat{K}_{\Theta}(z,z'))_{1,2}-(\hat{K}_{\Theta'}(z,z'))_{1,2}\right\vert&\leq \frac{4(4d^{2}(2A_{1}+4A_{0})L^2)\vert \Psi_{i} \vert}{b_{K}(m)(b(m))^{2}}\\
    &\leq \frac{16d^{2}L^{2}\vert\mathcal{M}\vert^{2}\vert\mathcal{N}\vert}{b_{K}(m)(b(m))^{2}}
    \end{aligned}
\end{align}
where in the first inequality we have used \eqref{dimensionalconstraint}, and $\vert\partial_{\theta_{s}}u_{k}(\Theta,\cdot) \vert\leq 2$,  and in the last inequality we have used \eqref{pt:0}.  We conclude that 

\begin{align}
\P\left[ \vert (\hat{K}_{\Theta}(z,z'))_{1,2}-\E[(\hat{K}_{\Theta}(z,z'))_{1,2}]\vert\geq \frac{\varepsilon}{2}\right]\leq \exp\left[\frac{-2\varepsilon^{2}(b_{K}(m))^{2}(b(m))^{4}}{4(16)^{2}d^{4}mL^{5}\vert\mathcal{M}\vert^{4}\vert\mathcal{N}\vert^{2}}\right].
\end{align}
Using the same argument of \eqref{midterm}, we have

\begin{align}
    \left\vert (\hat{K}_{\Theta}(z,z'))_{2,1}-(\hat{K}_{\Theta'}(z,z'))_{2,1}\right\vert\leq\frac{16d^{2}L^{2}\vert\mathcal{M}\vert^{2}\vert\mathcal{N}\vert}{b_{K}(m)(b(m))^{2}},
\end{align}
and thus
\begin{align}
\P\left[ \vert (\hat{K}_{\Theta}(z,z'))_{2,1}-\E[(\hat{K}_{\Theta}(z,z'))_{2,1}]\vert\geq \frac{\varepsilon}{2}\right]\leq \exp\left[\frac{-2\varepsilon^{2}(b_{K}(m))^{2}(b(m))^{4}}{4(16)^{2}d^{4}mL^{5}\vert\mathcal{M}\vert^{4}\vert\mathcal{N}\vert^{2}}\right].
\end{align}
Therefore,
\begin{align}\label{bounddafa_last}
\begin{aligned}
\P\left[\varepsilon^{2}\leq\norma{\hat{K}_{\Theta}(z,z')-K(z,z')}_{F}^{2}\right]&\leq    \sum_{i,j=1}^{2}\P\left[\frac{\varepsilon^{2}}{4}\leq \left((\hat{K}_{\Theta}(z,z'))_{i,j}-(K(z,z'))_{i,j}\right)^{2}\right]\\
&=\sum_{i,j=1}^{2}\P\left[\frac{\varepsilon}{2}\leq \left\vert(\hat{K}_{\Theta}(z,z'))_{i,j}-(K(z,z'))_{i,j}\right\vert\right]\\
&=\exp\left[\frac{-2\varepsilon^{2}(b_{K}(m))^{2}(b(m))^{4}}{4(64)^{2}mL^{9}d^{8}\vert\mathcal{M}\vert^{4}\vert\mathcal{N}\vert^{2}}\right]+\exp\left[\frac{-2\varepsilon^{2}(b_{K}(m))^{2}(b(m))^{4}}{4(16)^{2}mL\vert\mathcal{M}\vert^{4}\vert\mathcal{N}\vert^{2}}\right]+\\
&+2\exp\left[\frac{-2\varepsilon^{2}(b_{K}(m))^{2}(b(m))^{4}}{4(16)^{2}d^{4}mL^{5}\vert\mathcal{M}\vert^{4}\vert\mathcal{N}\vert^{2}}\right]\\
&\leq 4\exp\left[\frac{-2\varepsilon^{2}(b_{K}(m))^{2}(b(m))^{4}}{4(64)^{2}mL^{9}d^{8}\vert\mathcal{M}\vert^{4}\vert\mathcal{N}\vert^{2}}\right].
\end{aligned}
\end{align}
Furthermore, for m large enough we can have 

\begin{align}
 4\exp\left[\frac{-2\varepsilon^{2}(b_{K}(m))^{2}(b(m))^{4}}{4(64)^{2}mL^{9}d^{8}\vert\mathcal{M}\vert^{4}\vert\mathcal{N}\vert^{2}}\right]\leq \exp\left[\frac{-\varepsilon^{2}(b_{K}(m))^{2}(b(m))^{4}}{4(64)^{2}mL^{9}d^{8}\vert\mathcal{M}\vert^{4}\vert\mathcal{N}\vert^{2}}\right], 
\end{align}
and we are done.
\end{proof}
A further consequence of the bounds on $\A u$ is that the QIPNN $\U(\Theta,z)$ can be approximated with its linearized version.
\begin{lem}[Discrepancy with the linearized model]\label{lem:discrepancies}
The following bound holds true: Let $\Theta_{0}\in \cP$, then for each $z=(x,\hat{x})$ we have
\begin{align}\label{proximitytolinearized}
\begin{aligned}
\norma{\U(\Theta,z)-\U^{\lin}(\Theta,z)}_{2}&\leq \frac{4d^{2}(2A_{1}+4A_{0})L^3m}{b(m)}\vert\mathcal{M}\vert^{2}\vert \mathcal{N}\vert\norma{\Theta_{0}-\Theta}_{\infty}^{2}+\\
&\quad\quad+\frac{Lm}{b(m)}\vert\mathcal{M}\vert^{2}\vert \mathcal{N}\vert\norma{\Theta_{0}-\Theta}_{\infty}^{2}.
\end{aligned}
\end{align}
\end{lem}
\begin{proof}
We use the Taylor expansion with integral remainder to $u(\Theta,x)$ with respect to the parameters, and we traslate it to $\U(\Theta,x)$:

\begin{align}\label{expansiontaylor1}
u(\Theta,\hat{x})=u(\Theta_{0},\hat{x})+\nabla_{\Theta}u(\Theta_{0},\hat{x})^{T}(\Theta-\Theta_{0})+ \sum_{i,j=1}^{Lm}R_{ij}(\Theta,\hat{x})(\theta_{i}-\theta_{i,0})(\theta_{j}-\theta_{j,0}),    
\end{align}
with

\begin{align}
R_{ij}(\Theta,\hat{x})\coloneqq\frac{1}{2}\int_{0}^{1}(1-t)\partial_{\theta_{i}}\partial_{\theta_{j}}u(\Theta_{0}+t(\Theta-\Theta_{0}),x)\de t.   
\end{align}
By \cite[Lemma 4.17]{girardi2025}

\begin{align}
\vert R_{ij}(\Theta,\hat{x})\vert\leq \frac{\vert \mathcal{M}_{i}\cap \mathcal{M}_{j}\vert}{b(m)}.
\end{align}
By \autoref{expansiontaylor1}, we have

\begin{align}\label{expansiontaylor2}
\A u(\Theta,x)=\A u(\Theta_{0},x)+\nabla_{\Theta}\A u(\Theta_{0},x)^{T}(\Theta-\Theta_{0})+ \sum_{i,j=1}^{Lm}\hat{R}_{ij}(\Theta,x)(\theta_{i}-\theta_{i,0})(\theta_{j}-\theta_{j,0}),    
\end{align}
where
\begin{align}\label{remaider1}
\hat{R}_{ij}(\Theta,x)\coloneqq\frac{1}{2}\int_{0}^{1}(1-t)\partial_{\theta_{i}}\partial_{\theta_{j}}\A u(\Theta_{0}+t(\Theta-\Theta_{0}),x)\de t.
\end{align}
By taking into account the argument of \autoref{seconderivatives}, we have
\begin{align}\label{reminder2}
\begin{aligned}
\vert \hat{R}_{ij}(\Theta,x)\vert&\leq \frac{8d^{2}(2A_{1}+4A_{0})L^2}{b(m)}\vert \mathcal{M}_{i}\cap \mathcal{M}_{j}\vert \frac{\int_{0}^{1}(1-t)\de t}{2}\\
&=\frac{4d^{2}(2A_{1}+4A_{0})L^2}{b(m)}\vert \mathcal{M}_{i}\cap \mathcal{M}_{j}\vert.
\end{aligned}
\end{align}
Then

\begin{align}
\begin{aligned}
\norma{\U(\Theta,z)-\U^{\lin}(\Theta,z)}_{2}^{2}&=\left\vert \sum_{i,j=1}^{Lm}\hat{R}_{ij}(\Theta,x)(\theta_{i}-\theta_{i,0})(\theta_{j}-\theta_{j,0})\right\vert^{2}+ \left\vert \sum_{i,j=1}^{Lm}R_{ij}(\Theta,\hat{x})(\theta_{i}-\theta_{i,0})(\theta_{j}-\theta_{j,0})\right\vert^{2}\\
&\leq \left(\sum_{i,j=1}^{Lm}\vert \hat{R}_{ij}(\Theta,x)\vert\right)^{2}\norma{\Theta_{0}-\Theta}_{\infty}^{4} +\left(\sum_{i,j=1}^{Lm}\vert R_{ij}(\Theta,\hat{x})\vert\right)^{2}\norma{\Theta_{0}-\Theta}_{\infty}^{4}.
\end{aligned}   
\end{align}
Let us observe that 

\begin{align}
\begin{aligned}
&\sum_{i,j=1}^{Lm}\vert R_{ij}(\Theta,\hat{x})\vert\leq \frac{Lm}{b(m)}\max_{1\leq j\leq Lm}\sum_{i=1}^{Lm}\vert \mathcal{M}_{i}\cap M_{j}\vert,\\
&\sum_{i,j=1}^{Lm}\vert \hat{R}_{ij}(\Theta,x)\vert\leq \frac{4d^{2}(2A_{1}+4A_{0})L^3m}{b(m)}\max_{1\leq j\leq Lm}\sum_{i=1}^{Lm}\vert \mathcal{M}_{i}\cap M_{j}\vert.
\end{aligned}
\end{align}
Since 
\begin{align}
 \max_{1\leq j\leq Lm}\sum_{i=1}^{Lm}\vert \mathcal{M}_{i}\cap M_{j}\vert\leq \vert\mathcal{M}\vert^{2}\vert \mathcal{N}\vert,  
\end{align}
we conclude that

\begin{align}
\begin{aligned}
\norma{\U(\Theta,z)-\U^{\lin}(\Theta,z)}_{2}^{2}&\leq \left(\frac{4d^{2}(2A_{1}+4A_{0})L^3m}{b(m)}\vert\mathcal{M}\vert^{2}\vert \mathcal{N}\vert\right)^{2}\norma{\Theta_{0}-\Theta}_{\infty}^{4}+\\
&\quad\quad+\left(\frac{Lm}{b(m)}\vert\mathcal{M}\vert^{2}\vert \mathcal{N}\vert\right)^{2}\norma{\Theta_{0}-\Theta}_{\infty}^{4},
\end{aligned}
\end{align}
and since $\sqrt{a+b}\leq \sqrt{a}+\sqrt{b}$, for all $a,b\geq 0$, we get

\begin{align}
\begin{aligned}
\norma{\U(\Theta,z)-\U^{\lin}(\Theta,z)}_{2}&\leq \frac{4d^{2}(2A_{1}+4A_{0})L^3m}{b(m)}\vert\mathcal{M}\vert^{2}\vert \mathcal{N}\vert\norma{\Theta_{0}-\Theta}_{\infty}^{2}+\\
&\quad\quad+\frac{Lm}{b(m)}\vert\mathcal{M}\vert^{2}\vert \mathcal{N}\vert\norma{\Theta_{0}-\Theta}_{\infty}^{2}.
\end{aligned}
\end{align}
\end{proof}

\begin{thm}\label{thm:lazytraining}
Let us assume that \autoref{A0}--\autoref{A6}, and \eqref{dimensionalconstraint} hold true. Let us set $n\coloneqq n_{1}+n_{2}$, and define
\begin{align}\label{eq:R2}
    R(\delta)&\coloneqq \norma{Y}_{2}+\sqrt{\frac{2(n_{1}+n_{2})}{\delta}}
\end{align}
for a fixed a constant $0<\delta<1$ such that

\begin{align}\label{hp:lambda_pos}
\lambda_{\min}^{K}\geq\frac{4}{3}\left(g(\delta)+\sqrt{4BC}\right)   
\end{align}
where

\begin{align*}
&g(\delta)\coloneqq \frac{128\sqrt{m}L^{\frac{9}{2}}d^{4}|\mathcal{M}|^{2}|\mathcal{N}|(n_1+n_2)}{b_K(m)(b(m))^2}\sqrt{\log\left(\frac{2(n_1+n_2)^2}{\delta}\right)},\\
&B\coloneqq (n_1+n_2)\frac{256d^4mL^5|\mathcal M|^3|\mathcal N|}{b_K(m)(b(m))^2},\hskip 0,2cmC\coloneqq (\sqrt{n_1}+\sqrt{n_2})\frac{4d^2L^2|\mathcal M|}{b_K(m)b(m)}R(\delta).
\end{align*}
Then, there exists a positive number $\widetilde{\lambda}_{\min}(\delta)$ satisfying
\begin{align}\label{relation1}
    \widetilde{\lambda}_{\min}(\delta)\geq \frac{1}{4}\lambda_{\min}^{K},
\end{align} 
whose explicit expression is provided below, such that, when applying gradient flow with learning rate $\eta_{0}$, the following inequalities hold with probability at least $1-\delta$ over random initialization:
\begin{align}
\label{grad1}\mathcal{L}(\Theta_{t})&\leq \frac{R^2(\delta)}{2} e^{-2\eta_{0}\widetilde{\lambda}_{\min}(\delta)t} &\forall\,  t\geq 0,\\
\label{grad2}\|\Theta_{t}-\Theta_{0}\|_\infty&\leq \frac{1}{\widetilde{\lambda}_{\min}(\delta)}(\sqrt{n_{1}}+\sqrt{n_{2}})\frac{4d^{2}L^2\vert\mathcal{M}\vert}{b_{K}(m)b(m)} R(\delta)\left(1-e^{-\eta_{0}\widetilde{\lambda}_{\min}(\delta)t}\right) &\forall\,  t\geq 0,\\
\label{grad3}
\sup_{\substack{z\in B\times \partial B \\t\geq 0}}\norma{\U(\Theta_{t},z)-\U^{\lin}(\Theta_{t}^{\lin},z)}_{2}&\leq \frac{(5)2^{9}(n_{1}+n_{2})^{2}(R(\delta))^{2}d^{10}}{\widetilde{\lambda}_{\min}(\delta)}\Bigg[1+\frac{4}{\widetilde{\lambda}_{\min}(\delta)}+\\
&+\frac{256}{\eta_{0}(\widetilde{\lambda}_{\min}(\delta))^{2}b_{K}(m)}\Bigg]\frac{L^{12}m^{2}\vert\mathcal{M}\vert^{5}\vert\mathcal{N}\vert^{2}}{(b_{K}(m))^{2}\,(b(m))^{5}}(1+\log b(m))
\end{align}
\end{thm}
\begin{proof}
We first recall Chebyshev's inequality for random vectors. Let $V\in\mathbb{R}^d$ be a random vector with mean $\mathbb{E}[V]=v$ and covariance
$\mathrm{Cov}(V)=\Sigma$. Then, for any $t>0$,
\begin{align}\label{eq:chebyshev}
\mathbb{P}\left(\|V-v\|_2 \ge t\right)
\le \frac{\mathbb{E}\|V-v\|_2^2}{t^2}=\frac{\mathrm{Tr}(\Sigma)}{t^2}.
\end{align}
Let us now set $Z=(X,\hat X)^{T}\in \R^{n_{1}+n_{2}}$ 
By Assumption $\U(\Theta_{0},Z)$ is a centered random vector and let
\begin{align}
\mathcal{K}_0(Z,Z^T)\coloneqq \mathbb{E}\left[\U(\Theta_{0},Z)(\U(\Theta_{0},Z))^T\right].    
\end{align}
Applying \eqref{eq:chebyshev} with $v=0$, we obtain that, for any
$0<\delta<1$,
\begin{align}
\mathbb{P}\left(\|\U(\Theta_{0},Z)\|_2\ge \sqrt{\frac{2\mathrm{Tr}(\mathcal{K}_0(Z,Z^T))}{\delta}}
\right)\le\frac{\delta}{2}.
\end{align}

Since
\[
\mathrm{Tr}(\mathcal{K}_0(Z,Z^T))
=
\left\|
\big(\mathrm{diag}(\mathcal{K}_0(Z,Z^T))\big)^{1/2}
\right\|_2^2,
\]
it follows that, with probability at least $1-\frac{\delta}{2}$,
\begin{align}
\|\U(\Theta_{0},Z)\|_2 \le \sqrt{\frac{2}{\delta}}\left\|\big(\mathrm{diag}(\mathcal{K}_0(Z,Z^T))\big)^{1/2}\right\|_2.
\end{align}
Let us set 

\begin{align}
R(\delta)\coloneqq \|Y\|_2+\sqrt{\frac{2}{\delta}}
\left\|\big(\mathrm{diag}(\mathcal{K}_0(Z,Z^T))\big)^{1/2}
\right\|_2
\end{align}
Then one gets
\begin{align}\label{eq:corollaryR}
\|\U(\Theta_{0},Z)-Y\|_2 \le R(\delta),
\end{align}
with probability at least $1-\frac{\delta}{2}$. On the other hand, let us notice that

\begin{align}
 \left\|\big(\mathrm{diag}(\mathcal{K}_0(Z,Z^T))\big)^{1/2}
\right\|_2 =\left(\sum_{i=1}^{n_{1}+n_{2}}{\rm var}(F_{i}(0))\right)^{\frac{1}{2}}\leq \sqrt{n_{1}+n_{2}}   
\end{align}
where we have used that $\max_{z\in \X_B \times \partial B}{\rm diag}\left(\E\left[\U(\Theta,z)\U(\Theta,z)^{T}\right]\right)=1$.
Then we obtain that

\begin{align}
R(\delta)\leq \norma{Y}_{2}+\sqrt{\frac{2(n_{1}+n_{2})}{\delta}}.    
\end{align}
Let us consider $B_{r}(\Theta_{0})\coloneqq \left\{\Theta:\norma{\Theta-\Theta_{0}}_{\infty}< r \right\}$. Let $\|\,\cdot\,\|_{{\rm F}}$ be the Frobenius norm and let us apply \autoref{thm:NTK} to $M_{ij}\coloneqq K(z_{i},z_{j})-\hat K_{\Theta_{0}}(z_{i},z_{j})$, where $1\leq i,j\leq n_{1}+n_{2}$, i.e. $M=K-\hat K_{\Theta_{0}}$:
\begin{align}
\nonumber
\mathbb{P}[\|M\|_{{\rm F}}\geq \varepsilon]&=
\mathbb{P}\left[\sum_{i,j=1}^n(M_{ij})^2\geq \epsilon^2\right]\leq \mathbb{P}\left[\max_{ij}|F_{ij}|\geq \varepsilon/n\right]\\
&\leq\sum_{i,j=1}^{n_{1}+n_{2}}\mathbb{P}\left[|M_{ij}|\geq \frac{\varepsilon}{n_{1}+n_{2}}\right]\leq (n_{1}+n_{2})^2 \exp\left(-\frac{1}{4(64)^{2}}\frac{(b(m))^{4}}{mL^{9}d^{8}\vert \mathcal{M}\vert^{4} \vert\mathcal{N} \vert^{2}}\frac{\varepsilon^{2}}{(n_{1}+n_{2})^2}\right).
\end{align}
By letting 

\begin{align}
 g(\delta)\coloneqq \frac{2(64)\sqrt{m}L^{\frac{9}{2}}d^{4}\vert\mathcal{M}\vert^{2}\vert\mathcal{N}\vert(n_{1}+n_{2})}{b_{K}(m)(b(m))^{2}}\sqrt{\log\frac{2(n_{1}+n_{2})^{2}}{\delta}}
\end{align}
one has
\begin{align}
    \mathbb{P}\left[\|M\|_{{\rm F}}\geq g(\delta)\right]\leq \frac{\delta}{2}.
\end{align}
Let us observe that when $\|M\|_{{\rm F}}< g(\delta)$, the maximum eigenvalue of $|M|$ is $\lambda_M< g(\delta)$, so
\[ F\preceq |F| \preceq \lambda_F\id\prec g(\delta) \id.\]
From here, we deduce that 
\begin{align*}
K-\hat K_{\Theta_{0}} \prec g(\delta) \id \quad \Rightarrow\quad \hat K_{\Theta_{0}} \succ K-g(\delta) \id\succeq (\lambda_{\min}^{K}-g(\delta))\id.    
\end{align*}

So,
\begin{align}
\hat K_{\Theta_{0}}(Z,Z^T)\succ\left(\lambda_{\min}^{K}-g(\delta)\right)\id 
\end{align}
with probability at least $1-\frac{\delta}{2}$. Notice that by \autoref{lem:lipschitzntk}, we have for $t>0$ that
\begin{align}
\|\hat K_{\Theta_{t}}(z,z')-\hat K_{\Theta_{0}}(z,z')\|_{F}&\leq 
\frac{1}{b_{K}(m)}\left(\frac{128d^{4}(2A_{1}+4A_{0})^{2}mL^{5}\vert\mathcal{M}\vert^{3}\vert\mathcal{N}\vert}{(b(m))^{2}}+\frac{16Lm\vert\mathcal{M}\vert^{3}\vert\mathcal{N}\vert}{(b(m))^{2}}\right)\norma{\Theta_{t}-\Theta_{0}}_{\infty}\\
&\leq \frac{1}{b_{K}(m)}\frac{256d^{4}mL^{5}\vert\mathcal{M}\vert^{3}\vert\mathcal{N}\vert}{(b(m))^{2}}\norma{\Theta_{t}-\Theta_{0}}_{\infty}\\
&\leq  \frac{1}{b_{K}(m)}\frac{256d^{4}mL^{5}\vert\mathcal{M}\vert^{3}\vert\mathcal{N}\vert}{(b(m))^{2}}\rho(m),
\end{align}
where for some $\rho(m)>0$, we have defined
\begin{align}
t_1=\inf\left\{t:\|\Theta_{t}-\Theta_{0}\|_\infty\geq \rho(m)\right\}.
\label{t1}
\end{align}

Therefore
\begin{align}\label{ntkbound}
\|\hat K_{\Theta_{t}}(Z,Z^{T})-\hat K_{\Theta_{0}}(Z,Z^{T})\|_{{\rm F}}&\leq (n_{1}+n_{2})\frac{1}{b_{K}(m)}\frac{256d^{4}mL^{5}\vert\mathcal{M}\vert^{3}\vert\mathcal{N}\vert}{(b(m))^{2}}\rho(m)\eqqcolon h(\delta),
\end{align}
whence
\begin{align}\label{usataforq}
\hat K_{\Theta_{t}}(Z,Z^T)\succ\left(\lambda_{\min}^{K}-g(\delta)-h(\delta)\right)\id\eqqcolon\widetilde{\lambda}_{\min}(\delta)\id \qquad \forall \,t\leq t_1,
\end{align}
with probability at least $1-\delta$ (by the union bound applied to the events described above). In what follows, we then consider $\rho(m)$ such that
\begin{align}
\rho(m)=\frac{1}{\widetilde{\lambda}_{\min}(\delta)}(\sqrt{n_{1}}+\sqrt{n_{2}})\frac{4d^{2}L^2\vert\mathcal{M}\vert}{b_{K}(m)b(m)} R(\delta).
\end{align}
Since
\begin{align}
\widetilde{\lambda}_{\min}(\delta) = \lambda_{\min}^K - g(\delta) - h(\delta),
\end{align}
we have by substituting that
\begin{align}
\rho(m) \left[\lambda_{\min}^K - g(\delta)-(n_1+n_2)\frac{256d^4mL^5|\mathcal M|^3|\mathcal N|}{b_K(m)(b(m))^2}\rho(m)\right] =(\sqrt{n_1}+\sqrt{n_2})\frac{4d^2L^2|\mathcal M|}{b_K(m)b(m)}R(\delta),
\end{align}
which can be written as
\begin{align}
B(\rho(m))^2 - A\rho(m) + C = 0,
\label{eq:quadratic_rho}
\end{align}
with $A=\lambda_{\min}^{K}-g(\delta)$. By our hypothesis on $\lambda_{\min}^{K}$ one gets that \eqref{eq:quadratic_rho} admits a positive solution because $A^2\geq 4BC$. We now need to prove that under hypothesis \eqref{hp:lambda_pos}, we have
\begin{align}
\widetilde{\lambda}_{\min}(\delta) = \lambda_{\min}^{K} - g(\delta) - h(\delta) \geq \frac{1}{4}\lambda_{\min}^{K}.
\end{align}
Recall that
\begin{align}
h(\delta) = (n_1+n_2)\frac{256d^4mL^5|\mathcal M|^3|\mathcal N|}{b_K(m)(b(m))^2}\rho(m) = B\rho(m),
\end{align}
where $\rho(m)$ satisfies the quadratic equation \eqref{eq:quadratic_rho}. The solutions of this quadratic equation are given by
\begin{align}
\rho(m) = \frac{A \pm \sqrt{A^2 - 4BC}}{2B}.
\end{align}
We consider the smaller solution which is the relevant one for our problem:
\begin{align}
\rho(m) = \frac{A - \sqrt{A^2 - 4BC}}{2B}.
\end{align}
Therefore,
\begin{align}
h(\delta) = B\rho(m) = B \cdot \frac{A - \sqrt{A^2 - 4BC}}{2B} = \frac{A - \sqrt{A^2 - 4BC}}{2}.
\end{align}
Now, by computing $\widetilde{\lambda}_{\min}(\delta)$ one has
\begin{align}
\begin{aligned}
\widetilde{\lambda}_{\min}(\delta) &= \lambda_{\min}^{K} - g(\delta) - h(\delta)\\
&= A - h(\delta)\\
&= A - \frac{A - \sqrt{A^2 - 4BC}}{2}\\
&= \frac{A + \sqrt{A^2 - 4BC}}{2}.
\end{aligned}
\end{align}
Then we need to show that 
\begin{align}
\widetilde{\lambda}_{\min}(\delta) = \frac{A + \sqrt{A^2 - 4BC}}{2} \geq \frac{1}{4}\lambda_{\min}^{K},
\end{align}
which is equivalent to
\begin{align}\label{eq:equiv_ineq}
A + \sqrt{A^2 - 4BC} \geq \frac{1}{2}\lambda_{\min}^{K}.
\end{align}
Since $\widetilde{\lambda}_{\min}(\delta) = \lambda_{\min}^{K} - g(\delta) - h(\delta)$, the inequality \eqref{eq:equiv_ineq} is equivalent to
\begin{align}
g(\delta) + h(\delta) \leq \frac{3}{4}\lambda_{\min}^{K}.
\end{align}
Substituting the expression for $h(\delta)$, we have
\begin{align}
g(\delta) + h(\delta) &= g(\delta) + \frac{A - \sqrt{A^2 - 4BC}}{2}\nonumber\\
&= g(\delta) + \frac{(\lambda_{\min}^{K} - g(\delta)) - \sqrt{(\lambda_{\min}^{K} - g(\delta))^2 - 4BC}}{2}\nonumber\\
&= \frac{g(\delta) + \lambda_{\min}^{K}}{2} - \frac{\sqrt{A^2 - 4BC}}{2}.
\end{align}
We need to verify that
\begin{align}
\frac{g(\delta) + \lambda_{\min}^{K}}{2} - \frac{\sqrt{A^2 - 4BC}}{2} \leq \frac{3}{4}\lambda_{\min}^{K}.
\end{align}

Rearranging, this is equivalent to
\begin{align}
g(\delta) - \frac{1}{2}\lambda_{\min}^{K} \leq \sqrt{A^2 - 4BC} = \sqrt{(\lambda_{\min}^{K} - g(\delta))^2 - 4BC}.
\end{align}

From hypothesis \eqref{hp:lambda_pos}, we have
\begin{align}
\lambda_{\min}^{K} \geq \frac{4}{3}(g(\delta) + 2\sqrt{BC}),
\end{align}
which implies
\begin{align}
\frac{3}{4}\lambda_{\min}^{K} \geq g(\delta) + 2\sqrt{BC},
\end{align}
and thus
\begin{align}\label{eq:g_bound}
g(\delta) \leq \frac{3}{4}\lambda_{\min}^{K} - 2\sqrt{BC}.
\end{align}

Therefore,
\begin{align}
g(\delta) - \frac{1}{2}\lambda_{\min}^{K} \leq \frac{3}{4}\lambda_{\min}^{K} - 2\sqrt{BC} - \frac{1}{2}\lambda_{\min}^{K} = \frac{1}{4}\lambda_{\min}^{K} - 2\sqrt{BC}.
\end{align}

From \eqref{hp:lambda_pos}, we also have
\begin{align}
\lambda_{\min}^{K} \geq \frac{4}{3} \cdot 2\sqrt{BC} = \frac{8}{3}\sqrt{BC},
\end{align}
which gives
\begin{align}
\frac{1}{4}\lambda_{\min}^{K} \geq \frac{2}{3}\sqrt{BC}.
\end{align}
Thus,
\begin{align}
g(\delta) - \frac{1}{2}\lambda_{\min}^{K} \leq \frac{1}{4}\lambda_{\min}^{K} - 2\sqrt{BC} < 0.
\end{align}
Since the left-hand side is negative and $\sqrt{A^2 - 4BC} \geq 0$, it follows that
\begin{align}
g(\delta) - \frac{1}{2}\lambda_{\min}^{K} \leq \sqrt{A^2 - 4BC}.
\end{align}
Hence, we conclude that
\begin{align}
\widetilde{\lambda}_{\min}(\delta) = \lambda_{\min}^{K} - g(\delta) - h(\delta) \geq \frac{1}{4}\lambda_{\min}^{K}.
\end{align}
Now, recall that we have 
\begin{align}
    \frac{\de \U(\Theta_{t},z)}{\de t}=-\eta_{0}\widehat{K}_{\Theta_{t}}(z,Z^{T})\nabla_{\U(\Theta_{t},Z)}\L(\Theta_{t})
\end{align}
and for $t\leq t_1$ we have by \eqref{usataforq} with probability at least $1-\delta$ that
\begin{align}
\nonumber \frac{\de}{\de t}\|\U(\Theta_{t},Z)-Y\|_2^2&=-2\eta_{0}(\U(\Theta_{t},Z)-Y)^T\hat K_{\Theta_{t}}(\U(\Theta_{t},Z)-Y)\leq -2\eta_{0}\widetilde{\lambda}_{\min}(\delta)\|\U(\Theta_{t},Z)-Y\|_2^2,
\end{align}
so that
\begin{align}\label{disug2}
\begin{aligned}
\mathcal{L}(\Theta_{t})=\frac{1}{2}\|\U(\Theta_{t},Z)-Y\|_2^2&\leq \frac{1}{2}e^{-2\eta_{0}\widetilde{\lambda}_{\min}(\delta)t}\|\U(\Theta_{0},Z)-Y\|_2^2\\
&=\frac{1}{2}e^{-2\eta_{0}\widetilde{\lambda}_{\min}(\delta)t}\|\U(\Theta_{0},Z)-Y\|_2^2\\
&\leq e^{-2\eta_{0}\widetilde{\lambda}_{\min}(\delta)t}\frac{R^2(\delta)}{2}.
\end{aligned}
\end{align}
Let also take into account that

\begin{align}
\frac{\de \Theta_{t}}{\de t}=-\eta\nabla_{\Theta}\A u(\Theta_{t},X^{T})\nabla_{\A u(\Theta_{t},X)}\L(\Theta_{t})-\eta\nabla_{\Theta} u(\Theta_{t},\hat{X}^{T})\nabla_{u(\Theta_{t},\hat{X})}\L(\Theta_{t}),\\
\end{align}
so we have that 

\begin{align}
\frac{\de\left\vert \theta_{i}(t)-\theta_{i}(0)\right\vert}{\de t}&\leq \left\vert\frac{\de \theta_{i}(t)}{\de t}\right\vert\\
&\leq \eta \norma{\partial_{\theta_{i}}\A u(\Theta_{t},X)}_{2}\norma{\A u(\Theta_{t},X)-f(X)}_{2}+\\
&\phantom{formula}+\eta \norma{\partial_{\theta_{i}}u(\Theta_{t},\hat{X})}_{2}\norma{u(\Theta_{t},\hat{X})-g(\hat{X})}_{2}.
\end{align}
By recalling that

\begin{align}\label{boundsdop}
\begin{aligned}
&\left\vert \partial_{\theta_{s}}\A u(\Theta,x)\right\vert \le \frac{4d^{2}(2A_{1}+4A_{0})L^2\vert\mathcal{M}\vert}{b(m)} ,\\
&\left\vert \partial_{\theta_{s}} u_{k}(\Theta,\hat{x}) \right\vert \leq\frac{2\vert\mathcal{M}\vert}{b(m)},
\end{aligned}  
\end{align}
and that $\eta=\frac{\eta_{0}}{b_{K}(m)}$, we have

\begin{align}
\frac{\de\left\vert \theta_{i}(t)-\theta_{i}(0)\right\vert}{\de t}&\leq \eta_{0}\sqrt{n_{1}}\frac{4d^{2}(2A_{1}+4A_{0})L^2\vert\mathcal{M}\vert}{b_{K}(m)b(m)}\norma{\A u(\Theta_{t},X)-f(X)}_{2} +\\
&\phantom{formula}+\eta_{0}\sqrt{n_{2}}\frac{2\vert\mathcal{M}\vert}{b_{K}(m)b(m)}\norma{u(\Theta_{t},\hat{X})-g(\hat{X})}_{2}\\
&\leq\left(\eta_{0}\sqrt{n_{1}}\frac{4d^{2}(2A_{1}+4A_{0})L^2\vert\mathcal{M}\vert}{b_{K}(m)b(m)}+\eta_{0}\sqrt{n_{2}}\frac{2\vert\mathcal{M}\vert}{b_{K}(m)b(m)}\right)\norma{\U(\Theta_{t},Z)-Y}_{2}\\
&\leq \eta_{0}(\sqrt{n_{1}}+\sqrt{n_{2}})\frac{4d^{2}L^2\vert\mathcal{M}\vert}{b_{K}(m)b(m)}\norma{\U(\Theta_{t},Z)-Y}_{2}.
\end{align}
Hence, we have

\begin{align}
\frac{\de\left\vert \theta_{i}(t)-\theta_{i}(0)\right\vert}{\de t}\leq \eta_{0}(\sqrt{n_{1}}+\sqrt{n_{2}})\frac{4d^{2}L^2\vert\mathcal{M}\vert}{b_{K}(m)b(m)} R(\delta)e^{-\eta_{0}\widetilde{\lambda}_{\min}(\delta)t}  
\end{align}
and thus 

\begin{align}\label{disug3}
\vert \theta_{i}(t)-\theta_{i}(0)\vert\leq \frac{1}{\widetilde{\lambda}_{\min}(\delta)}(\sqrt{n_{1}}+\sqrt{n_{2}})\frac{4d^{2}L^2\vert\mathcal{M}\vert}{b_{K}(m)b(m)} R(\delta)\left(1-e^{-\eta_{0}\widetilde{\lambda}_{\min}(\delta)t}\right)   
\end{align}
for all $t\leq t_{1}$ with probability at least $1-\delta$, and thus

\begin{align}
 \norma{\Theta_{t}-\Theta_{0}}_{\infty}\leq \rho(m)\left(1-e^{-\eta_{0}\widetilde{\lambda}_{\min}(\delta)t}\right)   
\end{align}
for all $t\leq t_{1}$ with probability at least $1-\delta$. Notice that if $t_1<\infty$, then  
\begin{align}
\|\Theta_{t_1}-\Theta\|_\infty&\leq\rho(m)\left(1-e^{-\eta_{0}\widetilde{\lambda}_{\min}(\delta)t_1}\right)<\rho(m)\qquad \forall\,t\leq t_1,
\end{align}
with probability at least $1-\delta$, but this contradicts the definition \eqref{t1} of $t_1$, so we must have $t_1=\infty$. Therefore, have that \eqref{disug2} and \eqref{disug3} hold for any $t>0$, and thus we have proved \eqref{grad1} and \eqref{grad2}. Let us now proceed with the proof of \eqref{grad3}. Notice that by \autoref{lem:discrepancies}, we can estimate
\begin{align}
\begin{aligned}
\norma{\U(\Theta,z)-\U^{\lin}(\Theta,z)}_{2}&\leq \frac{4d^{2}(2A_{1}+4A_{0})L^3m}{b(m)}\vert\mathcal{M}\vert^{2}\vert \mathcal{N}\vert\norma{\Theta-\Theta_{0}}_{\infty}^{2}+\\
&\quad\quad+\frac{Lm}{b(m)}\vert\mathcal{M}\vert^{2}\vert \mathcal{N}\vert\norma{\Theta-\Theta_{0}}_{\infty}^{2}\\
&\leq \frac{8d^{2}L^3m}{b(m)}\vert\mathcal{M}\vert^{2}\vert \mathcal{N}\vert\norma{\Theta-\Theta_{0}}_{\infty}^{2}
\end{aligned}
\end{align}
so that,
\vskip -1cm

\begin{align}
\begin{aligned}
\sup_{t\geq 0}\norma{\U(\Theta_{t},z)-\U^{\lin}(\Theta_{t},z)}_{2}&\leq \frac{8d^{2}L^3m}{b(m)}\vert\mathcal{M}\vert^{2}\vert \mathcal{N}\vert\sup_{t\geq 0}\norma{\Theta_{t}-\Theta_{0}}_{\infty}^{2}\\
&\leq \frac{8d^{2}L^3m}{b(m)}\vert\mathcal{M}\vert^{2}\vert \mathcal{N}\vert(\rho(m))^{2}
\end{aligned}
\end{align}
with probability at least $1-\delta$. Let us define
$\Delta(t)\coloneqq\|\U(\Theta_{t},Z)-\U^{\mathrm{lin}}(\Theta_{t}^{\lin},Z)\|_2.$
By recalling that,
\begin{align}
 \displaystyle\begin{cases}
 &\frac{\de}{\de t}u(\Theta_{t},x)=-\eta\left(\nabla_{\Theta}u(\Theta_{t},x)\right)^{T}\nabla_{\Theta}\A u(\Theta_{t},X^{T})(\A u(\Theta_{t},X)-f(X))\\
 &\hskip 4cm-\eta\left(\nabla_{\Theta}u(\Theta_{t},x)\right)^{T}\nabla_{\Theta}u(\Theta_{t},\hat{X}^{T})(u(\Theta_{t},\hat{X})-g(\hat{X}))\,,\\
  &\frac{\de}{\de t}\A u(\Theta_{t},x)=-\eta\left(\nabla_{\Theta}\A u(\Theta_{t},x)\right)^{T}\nabla_{\Theta}\A u(\Theta_{t},X^{T})(\A u(\Theta_{t},X)-f(X))\\
 &\hskip 4cm-\eta\left(\nabla_{\Theta}\A u(\Theta_{t},x)\right)^{T}\nabla_{\Theta}u(\Theta_{t},\hat{X}^{T})(u(\Theta_{t},\hat{X})-g(\hat{X}))\,,
 \end{cases}   
\end{align}
we get that 
\begin{align}
\frac{1}{2}\frac{\de}{\de t}\Delta^2(t)
    &= \sum_{i=1}^{n_{1}}\left(\A u(\Theta_{t},x^{(i)})- \A u^{\mathrm{lin}}(\Theta^{\mathrm{lin}}_t,x^{(i)})\right)\left(\frac{\de}{\de t}\A u(\Theta_{t},x^{(i)})-\frac{\de}{\de t}\A u^{\mathrm{lin}}(\Theta^{\mathrm{lin}}_t,x^{(i)})\right)+\\
    &\sum_{i=1}^{n_{2}}\left(u(\Theta_{t},\hat{x}^{(i)})- u^{\mathrm{lin}}(\Theta^{\mathrm{lin}}_t,\hat{x}^{(i)})\right)\left(\frac{\de}{\de t} u(\Theta_t,x^{(i)})-\frac{\de}{\de t}u^{\mathrm{lin}}(\Theta^{\mathrm{lin}}_t,\hat{x}^{(i)})\right).    
\end{align}
By expanding the right-hand side we obtain,
\begin{align}    
\frac{1}{2}\frac{\de}{\de t}\Delta^2(t) &= -\eta\sum_{i=1}^{n_{1}}\left(\A u(\Theta_{t},x^{(i)})- \A u^{\mathrm{lin}}(\Theta^{\mathrm{lin}}_t,x^{(i)})\right)\times\\
    \nonumber
    & \qquad \times \left(\left(\nabla_{\Theta}\A u(\Theta_{t},\hat{x}^{(i)})\right)^{T}\nabla_{\Theta}\A u(\Theta_{t},X^{T})(\A u(\Theta_{t},X)-f(X))-\right.\\
    \\
    &\phantom{formula}\left. \left(\nabla_{\Theta}\A u(\Theta_{0},\hat{x}^{(i)})\right)^{T}\nabla_{\Theta}\A u(\Theta_{0},X^{T})(\A u^{\lin}(\Theta_{t}^{\lin},X)-f(X))\right)\\
    & -\eta\sum_{i=1}^{n_{1}}\left(\A u(\Theta_{t},x^{(i)})- \A u^{\mathrm{lin}}(\Theta^{\mathrm{lin}}_t,x^{(i)})\right)\times\\
    \nonumber
    & \qquad \times \left(\left(\nabla_{\Theta}\A u(\Theta_{t},x^{(i)})\right)^{T}\nabla_{\Theta}u(\Theta_{t},\hat{X}^{T})(u(\Theta_{t},\hat{X})-g(\hat{X}))-\right.\\
    \\
    &\phantom{formula}\left. \left(\nabla_{\Theta}\A u(\Theta_{0},x^{(i)})\right)^{T}\nabla_{\Theta} u(\Theta_{0},\hat{X}^{T})(\A u^{\lin}(\Theta_{t}^{\lin},\hat{X})-g(\hat{X}))\right)\\
    &-\eta\sum_{i=1}^{n_{2}}\left(u(\Theta_{t},\hat{x}^{(i)})- u^{\mathrm{lin}}(\Theta^{\mathrm{lin}}_t,\hat{x}^{(i)})\right)\times\\
    & \qquad \times \left(\left(\nabla_{\Theta} u(\Theta_{t},\hat{x}^{(i)})\right)^{T}\nabla_{\Theta}\A u(\Theta_{t},X^{T})(\A u(\Theta_{t},X)-f(X))-\right.\\
    \\
    &\phantom{formula}\left. \left(\nabla_{\Theta} u(\Theta_{0},\hat{x}^{(i)})\right)^{T}\nabla_{\Theta}\A u(\Theta_{0},X^{T})(\A u^{\lin}(\Theta_{t}^{\lin},X)-f(X))\right)\\
    & -\eta\sum_{i=1}^{n_{2}}\left(u(\Theta_{t},\hat{x}^{(i)})- u^{\mathrm{lin}}(\Theta^{\mathrm{lin}}_t,\hat{x}^{(i)})\right)\times\\
    \nonumber
    & \qquad \times \left(\left(\nabla_{\Theta} u(\Theta_{t},\hat{x}^{(i)})\right)^{T}\nabla_{\Theta}u(\Theta_{t},\hat{X}^{T})(u(\Theta_{t},\hat{X})-g(\hat{X}))-\right.\\
    \\
    &\phantom{formula}\left. \left(\nabla_{\Theta}u(\Theta_{0},\hat{x}^{(i)})\right)^{T}\nabla_{\Theta} u(\Theta_{0},\hat{X}^{T})(\A u^{\lin}(\Theta_{t}^{\lin},\hat{X})-g(\hat{X}))\right)\\
    &=-\eta\left(\U(\Theta_{t},Z)-\U^{\lin}(\Theta_{t}^{\lin},Z)\right)^{T}\hat{K}_{\Theta_{t}}\left(\U(\Theta_{t},Z)-Y\right)\\
    &+\eta\left(\U(\Theta_{t},Z)-\U^{\lin}(\Theta_{t}^{\lin},Z)\right)^{T}\hat{K}_{\Theta_{0}}\left(\U^{\lin}(\Theta_{t}^{\lin},Z)-Y\right).
\end{align}
Then one gets

\begin{align}
 \frac{1}{2}\frac{\de}{\de t}\Delta^2(t)&=-\eta\left(\U(\Theta_{t},Z)-\U^{\lin}(\Theta_{t}^{\lin},Z)\right)^{T}\hat{K}_{\Theta_{t}}\left(\U(\Theta_{t},Z)-Y\right)\\
 &-\eta\left(\U(\Theta_{t},Z)-\U^{\lin}(\Theta_{t}^{\lin},Z)\right)^{T}\hat{K}_{\Theta_{0}}\left(\U(\Theta_{t},Z)-\U^{\lin}(\Theta_{t}^{\lin},Z)\right)\\
    &+\eta\left(\U(\Theta_{t},Z)-\U^{\lin}(\Theta_{t}^{\lin},Z)\right)^{T}\hat{K}_{\Theta_{0}}\left(\U(\Theta_{t},Z)-Y\right).
\end{align}
Since $\hat K_{\Theta_{0}}$ is positive semidefinite, we have 

\begin{align}
\eta\left(\U(\Theta_{t},Z)-\U^{\lin}(\Theta_{t}^{\lin},Z)\right)^{T}\hat{K}_{\Theta_{0}}\left(\U(\Theta_{t},Z)-\U^{\lin}(\Theta_{t}^{\lin},Z)\right)\leq 0,     
\end{align}
so that

\begin{align}
\frac{1}{2}\frac{\de}{\de t}\Delta^2(t)&\leq -\eta\left(\U(\Theta_{t},Z)-\U^{\lin}(\Theta_{t}^{\lin},Z)\right)^{T}\hat{K}_{\Theta_{t}}\left(\U(\Theta_{t},Z)-Y\right)\\
  &+\eta\left(\U(\Theta_{t},Z)-\U^{\lin}(\Theta_{t}^{\lin},Z)\right)^{T}\hat{K}_{\Theta_{0}}\left(\U(\Theta_{t},Z)-Y\right)\\
  &=-\eta\left(\U(\Theta_{t},Z)-\U^{\lin}(\Theta_{t}^{\lin},Z)\right)^{T}\left(\hat{K}_{\Theta_{t}}-\hat{K}_{\Theta_{0}}\right)\left(\U(\Theta_{t},Z)-Y\right).
\end{align}
Then, we obtain

\begin{align}
\left\vert \Delta(t) \frac{\de}{\de t}\Delta(t)\right\vert\leq \eta\norma{\U(\Theta_{t},Z)-\U^{\lin}(\Theta_{t}^{\lin},Z)}_{2}\norma{\hat{K}_{\Theta_{t}}-\hat{K}_{\Theta_{0}}}_{{\rm op}}\norma{\U(\Theta_{t},Z)-Y}_{2},  
\end{align}
from which 

\begin{align}
\left\vert\frac{\de}{\de t}\Delta(t)\right\vert \leq\eta\norma{\hat{K}_{\Theta_{t}}-\hat{K}_{\Theta_{0}}}_{{\rm op}}\norma{\U(\Theta_{t},Z)-Y}_{2}.  
\end{align}
Since by \eqref{disug2}
\begin{align}
    \norma{\U(\Theta_{t},Z)-Y}_{2}\leq R(\delta)e^{-\eta_{0}\widetilde{\lambda}_{\min}(\delta)t},
\end{align}
and by \eqref{ntkbound}
\begin{align}\label{opdiventaf}
\|\hat K_{\Theta_t}-\hat K_{\Theta_0}\|_{\mathrm{op}}&\leq \|\hat K_{\Theta_t}-\hat K_{\Theta_0}\|_{{\rm HS}}\leq (n_{1}+n_{2})\frac{1}{b_{K}(m)}\frac{256d^{4}mL^{5}\vert\mathcal{M}\vert^{3}\vert\mathcal{N}\vert}{(b(m))^{2}}\rho(m),
\end{align}
then we obtain that 

\begin{align}
\left\vert\frac{\de}{\de t}\Delta(t)\right\vert \leq\eta(n_{1}+n_{2})\frac{1}{b_{K}(m)}\frac{256d^{4}mL^{5}\vert\mathcal{M}\vert^{3}\vert\mathcal{N}\vert}{(b(m))^{2}}\rho(m)R(\delta)e^{-\eta_{0}\widetilde{\lambda}_{\min}(\delta)t}   
\end{align}
implying that 

\begin{align}\label{eq:deltaF}
\begin{aligned}
\Delta(t)\leq\frac{\eta}{\eta_{0}\widetilde{\lambda}_{\min}(\delta)}(n_{1}+n_{2})\frac{1}{b_{K}(m)}\frac{256d^{4}mL^{5}\vert\mathcal{M}\vert^{3}\vert\mathcal{N}\vert}{(b(m))^{2}}\rho(m)R(\delta)\left(1-e^{-\eta_{0}\widetilde{\lambda}_{\min}(\delta)t}\right)\\
 =\frac{1}{\widetilde{\lambda}_{\min}(\delta)}(n_{1}+n_{2})\frac{1}{(b_{K}(m))^{2}}\frac{256d^{4}mL^{5}\vert\mathcal{M}\vert^{3}\vert\mathcal{N}\vert}{(b(m))^{2}}\rho(m)R(\delta)\left(1-e^{-\eta_{0}\widetilde{\lambda}_{\min}(\delta)t}\right).    
\end{aligned}
\end{align}
We notice that
\begin{align}
\big\|\dot\Theta_{t}-\dot\Theta_t^{\mathrm{lin}}\big\|_\infty&=\eta\left\|\nabla_\Theta \U(\Theta_{t},Z^T)(\U(\Theta_{t},Z)-Y)-\nabla_\Theta \U(\Theta_0,Z^T)(\U^{\lin}(\Theta_{t}^{\lin},Z)-Y)\right\|_\infty\\
&\leq \eta\left\|\left(\nabla_\Theta \U(\Theta_{t},Z^T)-\nabla_\Theta \U(\Theta_0,Z^T)\right)(\U(\Theta_{t},Z)-Y)\right\|_\infty+\\
&+\eta\left\|\nabla_\Theta \U(\Theta_0,Z^T)(\U^{\lin}(\Theta_{t}^{\lin})-\U(\Theta_{t},Z))\right\|_\infty\\
&\leq \eta\sup_i\|\partial_{\theta_i} \U(\Theta_{t},Z)-\partial_{\theta_i} \U(\Theta_0,Z)\|_2\|\U(\Theta_{t},Z)-Y\|_2+\\
&+\eta\sup_i\|\partial_{\theta_i} \U(\Theta_0,Z)\|_2\|\U^{\lin}(\Theta_{t}^{\lin},Z)-\U(\Theta_{t},Z)\|_2.
\end{align}
In what follows, we bound each term of the previous expression. Recall that by \eqref{lip:A} we have

\begin{align}
\norma{\nabla_{\Theta}\A u(\Theta,x)-\nabla_{\Theta}\A u(\Theta',x)}_{\infty}&\le \frac{8d^{2}(2A_{1}+4A_{0})L^2\vert \mathcal{M}\vert^{2}\vert\mathcal{N}\vert}{b(m)}\norma{\Theta-\Theta'}_{\infty}\\
&\leq \frac{8d^{2}L^2\vert \mathcal{M}\vert^{2}\vert\mathcal{N}\vert}{b(m)}\norma{\Theta-\Theta'}_{\infty},
\end{align}
and combining this bound with the Lipschitzness result of \autoref{lemma4.20}  one has
\begin{align}
&\sup_i\|\partial_{\theta_i}U(\Theta_{t},Z)-\partial_{\theta_i} \U(\Theta_0,Z)\|_2
\\
&=\sup_{i}\sqrt{\sum_{j=1}^{n_{1}}\left(\partial_{\theta_i}\A u(\Theta_{t},x^{(j)})-\partial_{\theta_i}\A u(\Theta_0,x^{(j)})\right)^{2}+\sum_{j=1}^{n_{2}}\left(\partial_{\theta_i} u(\Theta_{t},\hat{x}^{(j)})-\partial_{\theta_i} u(\Theta_0,\hat{x}^{(j)})\right)^{2}}\\
&\leq \frac{\sqrt{n_{1}+n_{2}}16d^{2}L^2\vert \mathcal{M}\vert^{2}\vert\mathcal{N}\vert}{b(m)}\norma{\Theta_{t}-\Theta_{0}}_{\infty}.
\end{align}
Together with the lazy training bound \eqref{disug3}, and using the convergence to the examples \eqref{disug2}, we control the first term:
\begin{align}
&\eta\sup_i\|\partial_{\theta_i} \U(\Theta_{t},Z)-\partial_{\theta_i} \U(\Theta_0,Z)\|_2\|\U(\Theta_{t},Z)-Y\|_2\leq \\
&\phantom{formula}\leq \eta R(\delta)\frac{\sqrt{n_{1}+n_{2}}16d^{2}L^2\vert \mathcal{M}\vert^{2}\vert\mathcal{N}\vert}{b(m)}\rho(m)e^{-\eta_{0}\widetilde{\lambda}_{\min}(\delta)t}\eqqcolon T_{1}(t).
\end{align}
Let us now estimate the second term. To this aim, we need two different estimates to be used for ``small'' and ``large'' $t$, as follows. The first estimate is based on the Lipschitzness of the gradient \autoref{lemma4.20}, \eqref{boundsdop}:
\begin{align}\label{consequence}
\sup_i\|\partial_{\theta_i} \U(\Theta_0,Z)\|_2\leq \sqrt{n_{1}+n_{2}}\frac{8d^{2}L^2\vert\mathcal{M}\vert}{b(m)},
\end{align}
and combined with \eqref{eq:deltaF}, one has

\begin{align}
&\eta\sup_i\|\partial_{\theta_i} \U(\Theta_0,Z)\|_2\|\U^{\lin}(\Theta_{t}^{\lin},Z)-\U(\Theta_{t},Z)\|_2\leq  \\
&\phantom{formula}\leq\frac{1}{\widetilde{\lambda}_{\min}(\delta)}\sqrt{n_{1}+n_{2}}(n_{1}+n_{2})\frac{8}{(b_{K}(m))^{2}}\frac{256d^{6}mL^{7}\vert\mathcal{M}\vert^{4}\vert\mathcal{N}\vert}{(b(m))^{3}}\rho(m)R(\delta)\eqqcolon T_{2}(t).
\end{align}
Lastly, let us notice that by considering the analytic solution of the linearized model, and \eqref{eq:corollaryR}
\begin{align}
\nonumber
&\eta\sup_i\|\partial_{\theta_i} \U(\Theta_0,Z)\|_2\|\U^{\lin}(\Theta_{t}^{\lin},Z)-\U(\Theta_{t},Z)\|_2\leq \\
&\phantom{formula}\leq \eta\sqrt{n_{1}+n_{2}}\frac{8d^{2}L^2\vert\mathcal{M}\vert}{b(m)}\left(\|\U^{\lin}(\Theta_{t}^{\lin},Z)-Y\|_2+\|\U(\Theta_{t},Z)-Y\|_2\right)\\
\nonumber
&\phantom{formula}\leq \eta\sqrt{n_{1}+n_{2}}\frac{8d^{2}L^2\vert\mathcal{M}\vert}{b(m)}\left(e^{-\eta_{0}\lambda_{\min}^{K} t}\|\U(\Theta_{0},Z)-Y\|_2+R(\delta) e^{-\eta_{0}\widetilde{\lambda}_{\min}(\delta)t}\right)\\
\nonumber
&\phantom{formula}\leq 2\eta\sqrt{n_{1}+n_{2}}\frac{8d^{2}L^2\vert\mathcal{M}\vert}{b(m)}R_{\delta}(\delta)\left(e^{-\eta_{0}\widetilde{\lambda}_{\min}(\delta)t}\right)\eqcolon T_{3}(t).
\end{align}
Then

\begin{align}
    \big\|\dot\Theta_{t}-\dot\Theta_{t}^{\lin}\big\|_\infty\leq T_{1}(t)+T_{2}(t) \quad \text{and}\quad \big\|\dot\Theta_{t}-\dot\Theta_t^{\mathrm{lin}}\big\|_\infty\leq T_{1}(t)+T_{3}(t).
\end{align}
Let us define
\begin{align}
    t^{\ast} = \frac{1}{\eta_{0} \widetilde{\lambda}_{\min}(\delta)}\log b(m).
\end{align}
We have
\begin{align}
\big\|\Theta_{t}-\Theta_t^{\mathrm{lin}}\big\|_\infty&\leq \int_0^\infty T_{1}(t)\de t + \int_0^{t^\ast} T_{2}(t)\de t + \int_{t^\ast}^\infty T_{3}(t)\de t.
\end{align}
Notice that 

\begin{align}
\int_{0}^{\infty}T_{1}(t)\de t= R(\delta)\frac{\sqrt{n_{1}+n_{2}}16d^{2}L^2\vert \mathcal{M}\vert^{2}\vert\mathcal{N}\vert}{\widetilde{\lambda}_{\min}(\delta)b_{K}(m)b(m)}\rho(m)
\end{align}
where we have used $\eta=\frac{\eta_{0}}{b_{K}(m)}$. Now observe that

\begin{align}
\int_{0}^{t^{\ast}}T_{2}(t)\de t=\frac{1}{\eta_{0}(\widetilde{\lambda}_{\min}(\delta))^{2}}\sqrt{n_{1}+n_{2}}(n_{1}+n_{2})\frac{8}{(b_{K}(m))^{2}}\frac{256d^{6}mL^{7}\vert\mathcal{M}\vert^{4}\vert\mathcal{N}\vert}{(b(m))^{3}}\rho(m)R(\delta)\log b(m).  
\end{align}
Lastly, we have that

\begin{align}
\int_{t^\ast}^\infty T_{3}(t)\de t= \frac{2}{b_{K}(m)\widetilde{\lambda}_{\min}(\delta)}\sqrt{n_{1}+n_{2}}\frac{16d^{2}L^2\vert\mathcal{M}\vert}{b(m)}R(\delta)e^{-\log b(m)}   
\end{align}
Then, we have obtained that

\begin{align}
 \big\|\Theta_{t}-\Theta_t^{\mathrm{lin}}\big\|_\infty&\leq  R(\delta)\frac{\sqrt{n_{1}+n_{2}}16d^{2}L^2\vert \mathcal{M}\vert^{2}\vert\mathcal{N}\vert}{\widetilde{\lambda}_{\min}(\delta)b_{K}(m)b(m)}\rho(m)+\\
 &+\frac{1}{\eta_{0}(\widetilde{\lambda}_{\min}(\delta))^{2}}\sqrt{n_{1}+n_{2}}(n_{1}+n_{2})\frac{8}{(b_{K}(m))^{2}}\frac{256d^{6}mL^{7}\vert\mathcal{M}\vert^{4}\vert\mathcal{N}\vert}{(b(m))^{3}}\rho(m)R(\delta)\log b(m)\\
 &+\frac{2}{b_{K}(m)\widetilde{\lambda}_{\min}(\delta)}\sqrt{n_{1}+n_{2}}\frac{16d^{2}L^2\vert\mathcal{M}\vert}{b(m)}R(\delta)\frac{1}{b(m)}\\
 &=\frac{16 d^{2} R(\delta)\sqrt{n_1+n_2}L^2 |\mathcal M|}{\widetilde\lambda_{\min}(\delta) b_{K}(m) b(m)}\Bigg[|\mathcal M| |\mathcal N|\rho(m)\\
&+\frac{128\, d^4 m L^5 |\mathcal M|^3 |\mathcal N|}
{\eta_0\widetilde\lambda_{\min}(\delta) b_{K}(m) (b(m))^{2}}
(n_1+n_2)\rho(m)\log b(m)+\frac{2}
{b_{K}(m)b(m)}
\Bigg].
\end{align}
Notice that by \autoref{lem:discrepancies}, $(2A_{1}+4A_{0})\leq 1$, \eqref{consequence}, and the latter bound for $\big\|\Theta_{t}-\Theta_t^{\mathrm{lin}}\big\|_\infty$ we have

\begin{align}
&\|\U(\Theta_{t},z)-\U^{\lin}(\Theta_{t}^{\lin},z)\|_{2}\\
\nonumber
&=\|\U(\Theta_{t},z)-\U(\Theta_0,z)-\nabla_\Theta \U(\Theta_0,z)^T(\Theta_t^{\mathrm{lin}}-\Theta_0)\|_{2}\\
\nonumber
&\leq \|\U(\Theta_{t},z)-\U^{\lin}(\Theta_{t},z)\|_{2}+\|\nabla_\Theta \U(\Theta_0,z)^T(\Theta_{t}^{\lin}-\Theta_{t})\|_{2}\\
\nonumber
&\leq \frac{8d^{2}L^3m}{b(m)}\vert\mathcal{M}\vert^{2}\vert \mathcal{N}\vert\norma{\Theta_{t}-\Theta_{0}}_{\infty}^{2}+\sqrt{2}\frac{16d^{2}mL^3\vert\mathcal{M}\vert}{b(m)}\|\Theta_t^{\mathrm{lin}}-\Theta_{t}\|_\infty\\
\nonumber
&\leq\frac{8d^{2}L^3m}{b(m)}\vert\mathcal{M}\vert^{2}\vert \mathcal{N}\vert\rho(m)^{2}+\\
&+\frac{(16)^{2}md^{4}R(\delta)\sqrt{n_1+n_2}L^{5} |\mathcal M|^{2}}{\widetilde\lambda_{\min}(\delta) b_{K}(m) (b(m))^{2}}\Bigg[|\mathcal M| |\mathcal N|\rho(m)
+\\
&+\frac{128\, d^4 m L^5 |\mathcal M|^3 |\mathcal N|}
{\eta_0\widetilde\lambda_{\min}(\delta) b_{K}(m) (b(m))^{2}}
(n_1+n_2)\rho(m)\log b(m)+\frac{2}
{b_{K}(m)b(m)}
\Bigg].
\end{align}
By recalling that
\begin{align}
\rho(m)=\frac{1}{\widetilde{\lambda}_{\min}(\delta)}(\sqrt{n_{1}}+\sqrt{n_{2}})\frac{4d^{2}L^2\vert\mathcal{M}\vert}{b_{K}(m)b(m)} R(\delta)    
\end{align}
we obtain
\begin{align}
\begin{aligned}
\|\U(\Theta_t,z)-\U^{\lin}(\Theta_{t}^{\lin},z)\|_{2}
\le\;&
\frac{128\, d^{6}L^{7}m|\mathcal M|^{4}|\mathcal N|}
{(\widetilde{\lambda}_{\min}(\delta))^{2}\,(b_{K}(m))^{2}\,(b(m))^{3}}
(\sqrt{n_1}+\sqrt{n_2})^{2}
R^{2}(\delta)\\
&+
\frac{2^{10} d^{6}L^{7}m|\mathcal M|^{3}|\mathcal N|}
{(\widetilde{\lambda}_{\min}(\delta))^{2}\,(b_{K}(m))^{2}\,(b(m))^{3}}
\sqrt{n_1+n_2}(\sqrt{n_1}+\sqrt{n_2})
R^{2}(\delta)\\
&+
\frac{2^{17}d^{10}L^{12}m^{2}|\mathcal M|^{5}|\mathcal N|}
{\eta_0\,(\widetilde{\lambda}_{\min}(\delta))^{3}\,
(b_{K}(m))^{3}\,(b(m))^{5}}(n_1+n_2)^{3/2}(\sqrt{n_1}+\sqrt{n_2})
R^{2}(\delta)\log b(m)\\
&+\frac{2^{8} d^{4} L^{5}m|\mathcal M|^{2}}
{\widetilde{\lambda}_{\min}(\delta)\,(b_{K}(m))^{2}\,(b(m))^{3}}
\sqrt{n_1+n_2}\,R(\delta)\\
&
\leq
\frac{2^{11} d^{6}L^{7}m|\mathcal M|^{4}|\mathcal N|}
{\widetilde{\lambda}_{\min}^{2}(\delta)\,(b_{K}(m))^{2}\,(b(m))^{3}}
(n_1+n_2)^{2}R^{2}(\delta)\\
&+\frac{2^{17} d^{10}L^{12}m^{2}|\mathcal M|^{5}|\mathcal N|}
{\eta_0\,(\widetilde{\lambda}_{\min}(\delta))^{3}\,
(b_{K}(m))^{3}\,(b(m))^{5}}
(n_1+n_2)^{2}R^{2}(\delta)\log b(m)\\
&+\frac{2^{9}  d^{4}L^{5}m|\mathcal M|^{2}}
{\widetilde{\lambda}_{\min}(\delta)\,(b_{K}(m))^{2}\,(b(m))^{3}}
\sqrt{n_1+n_2}\,R(\delta).
\end{aligned}
\end{align}
On the other hand, notice that by \eqref{pt:0}, and $|u_k(\Theta,\hat x)|\leq 1$ we have

\begin{align}
\|\U_k(\Theta,z)\|_2^2
&=|\A u_k(\Theta,x)|^2 + |u_k(\Theta,\hat x)|^2\\
&\le
\big(2d^2(2A_1+4A_0)L^2\big)^2 + 1\\
&\le 4d^4L^4 + 1\\
&\le 5d^4L^4.
\end{align}
Then

\begin{align}
\left\|\mathbb E[\U(\Theta,z)\U(\Theta,z)^T]\right\|_{F}
&=
\left\|
\frac{1}{(b(m))^{2}}
\sum_{k=1}^m\sum_{k'\in\mathcal P_k}
\mathbb E[\U_k(\Theta,z)\U_{k'}(\Theta,z)^T]
\right\|_{F}\\
&\le
\frac{1}{(b(m))^{2}}
\sum_{k=1}^m\sum_{k'\in\mathcal P_k}
5 d^4 L^4 \\
&\le
\frac{5 d^4 L^4}{(b(m))^{2}}\, m\,\max_k|\mathcal P_k| \\
&\le
\frac{5 d^4 L^4}{(b(m))^{2}}\, m\,|\mathcal M|\,|\mathcal N|.
\end{align}
Hence

\begin{align}\label{growsbm}
\begin{aligned}
  1= \max_{z\in \X_B \times \partial B}{\rm diag}\left(\E[\U(\Theta,z)(\U(\Theta,z))^{T}]\right)&\leq \left\|
\max_{z\in \X_B \times \partial B}\mathbb E[\U(\Theta,z)\U(\Theta,z)^T]\right\|_{F}\\
&\leq \frac{5 d^4 L^4}{(b(m))^{2}}\, m\,|\mathcal M|\,|\mathcal N|.
\end{aligned}
\end{align}
From here, we then conclude that

\begin{align}
\begin{aligned}
\|\U(\Theta_{t},z)-\U^{\lin}(\Theta_{t}^{\lin},z)\|_{2}
\le\;&\frac{(5)2^{11}d^{10}L^{11}m^{2}|\mathcal M|^{5}|\mathcal N|^{2}}
{(\widetilde{\lambda}_{\min}(\delta))^{2}\,(b_{K}(m))^{2}\,(b(m))^{5}}
(n_1+n_2)^{2}R^{2}(\delta)\\
&+
\frac{2^{17}d^{10}L^{12}m^{2}|\mathcal M|^{5}|\mathcal N|}
{\eta_0\,(\widetilde{\lambda}_{\min}(\delta))^{3}\,
(b_{K}(m))^{3}\,(b(m))^{5}}
(n_1+n_2)^{2}R^{2}(\delta)\log b(m)\\
&+\frac{(5)2^{9}d^{8}L^{9}m^{2}|\mathcal M|^{3}|\mathcal N|}
{\widetilde{\lambda}_{\min}(\delta)\,(b_{K}(m))^{2}\,(b(m))^{5}}
\sqrt{n_1+n_2}\,R(\delta)\eqqcolon \eta(\delta).
\end{aligned}
\end{align}
Let us now notice that 

\begin{align}
\eta(\delta)&\leq \frac{(5)2^{9}(n_{1}+n_{2})^{2}(R(\delta))^{2}d^{10}}{\widetilde{\lambda}_{\min}(\delta)}\Bigg[1+\frac{4}{\widetilde{\lambda}_{\min}(\delta)}+\\
&+\frac{256}{\eta_{0}(\widetilde{\lambda}_{\min}(\delta))^{2}b_{K}(m)}\Bigg]\frac{L^{12}m^{2}\vert\mathcal{M}\vert^{5}\vert\mathcal{N}\vert^{2}}{(b_{K}(m))^{2}\,(b(m))^{5}}(1+\log b(m)),   
\end{align}
and we are done.
\end{proof}
\section{The variational physics informed formulation}\label{sec:variational}
In this part, we analyze the lazy training of \eqref{linearPDE} by considering its variational formulation itself. In order to study this formulation, let us first recall how the PDE can be solved by using a functional setup.
\subsection{Existence and uniqueness of a solution}
In this subsection, we briefly recall the functional spaces needed to find a solution for \eqref{linearPDE}. Let $B$ an open bounded subset of $\R^{d}$ with Lipschitz boundary, and consider the \eqref{linearPDE}. Let us set

\begin{align}\label{bordospace}
H^{1}(B)\coloneqq \left\{ u\in L^{2}(B): \nabla u\in (L^{2}(B))^{d}\right\}
\end{align}
endowed with the norm

\begin{align}
 \norma{u}_{1,B}^{2}\coloneqq \int_{B}\vert u\vert^{2}\de x+ \int_{B}\norma{\nabla u}_{2}^{2}\de x.   
\end{align}
Let us denote by $C^{\infty}(B)$ the space of infinitely differentiable functions on $B$, and we set by $\D(B)$ the space defined by 

\begin{align}
\D(B)\coloneqq \{u\in C^{\infty}(B): \text{$u$ has compact support on $B$}\}.   
\end{align}

In what follows, we denote by $H_{0}^{1}(B)$ the closure of $\D(B)$ in $H^{1}(B)$. It is worth noting that this space can be characterized as the space of functions $u$ belonging to $H^{1}(B)$ that vanish at the boundary $\partial B$, namely,

\begin{align}
H_{0}^{1}(B)=\left\{ u\in H^{1}(B): u=0\hskip 0,2cm \text{on $\partial B$}\right\}.
\end{align}
In order to solve \eqref{linearPDE} different strategies can be adopted. In the next, we briefly recall the strategy based on the Lax-Milgram theorem \cite[Chapter 5]{brezis2011functional}. Hence, let us define

\begin{align}\label{tracciaspace}
H^{\frac{1}{2}}(\partial B)\coloneqq \left\{u_{|\partial B}: u\in H^{1}(B)\right\}
\end{align}
where $u_{|\partial B}$ denotes the restriction of $u$ to $\partial B$. This space is endowed with the norm

\begin{align}
 \norma{u}_{\frac{1}{2},\partial B}\coloneqq \inf\left\{ \norma{v}_{1,B}: v\in H^{1}(B),\, v_{|\partial B}=u\right\}.   
\end{align}
It is well-known that $H^{\frac{1}{2}}(\partial B)$ is a Hilbert space with continuous immersion in $L^{2}(\partial B)$. Furthermore, there exists a linear and continuous operator $\mathcal{R}: H^{\frac{1}{2}}(\partial B)\rightarrow H^{1}(B)$ such that $(\mathcal{R}(v))_{|\partial B}=v$ for all $v\in H^{\frac{1}{2}}(\partial B)$. Hence, by assuming that $f\in L^{2}(B)$ and $g\in H^{\frac{1}{2}}(\partial B)$, the Lax-Milgram theorem guarantees the existence of a unique weak solution $u\in H^{1}(B)$ to \eqref{linearPDE} satisfying $u_{|\partial B}=g$. The proof relies on verifying that the bilinear form $a(w,v)\coloneqq\sum_{i,j=1}^{d}\int_{B}a_{ij}(x)\frac{\partial w}{\partial x_{i}}\frac{\partial v}{\partial x_{j}}\de x$ is continuous and coercive on $H_{0}^{1}(B)$ (the coercivity following from the ellipticity condition \eqref{ellipticiy} and Poincaré's inequality), and that the right-hand side defines a continuous linear functional on $H_{0}^{1}(B)$. Let briefly describe the strategy that we can adopt to find its solution. Let us consider some $\widetilde{u}\in H^{1}(B)$ such that $\widetilde{u}_{|\partial B}=g$. Define

\begin{align}\label{spaceW}
\begin{aligned}
  W&\coloneqq \widetilde{u} + H_{0}^{1}(B)\\
  &=\left\{v\in H^{1}(B): v-\widetilde{u}\in H_{0}^{1}(B)\right\}
\end{aligned}
\end{align}
Notice that $W$ is a closed affine subspace of $H^{1}(B)$. We make formulate \eqref{linearPDE} as the following problem: we seek $u\in W$ such that

\begin{align}\label{weakfor}
\displaystyle\sum_{i,j=1}^{d}\int_{B}a_{ij}(x)\frac{\partial u}{\partial x_{i}}\frac{\partial v}{\partial x_{j}}\de x=\int_{B}fv\de x \hskip 0,2cm \text{for all $v\in H_{0}^{1}(B)$.}    
\end{align}
As for the existence of a solution of \eqref{weakfor}, we set $u=\widetilde{u}+w$ for $w\in H_{0}^{1}(B)$. Then we seek $w\in H_{0}^{1}(B)$ such that

\begin{align}\label{weakfor2}
\displaystyle\sum_{i,j=1}^{d}\int_{B}a_{ij}(x)\frac{\partial w}{\partial x_{i}}\frac{\partial v}{\partial x_{j}}\de x=\int_{B}fv\de x - \displaystyle\sum_{i,j=1}^{d}\int_{B}a_{ij}(x)\frac{\partial \widetilde{u}}{\partial x_{i}}\frac{\partial v}{\partial x_{j}}\de x\hskip 0,2cm \text{for all $v\in H_{0}^{1}(B)$.}    
\end{align}
By noting that the map 

\begin{align}
    v\longmapsto \int_{B}fv\de x - \displaystyle\sum_{i,j=1}^{d}\int_{B}a_{ij}(x)\frac{\partial \widetilde{u}}{\partial x_{i}}\frac{\partial v}{\partial x_{j}}\de x
\end{align}
defines a linear and continuous functional on $H_{0}^{1}(B)$. By assuming that $a_{ij}=a_{ji}$, then by Lax-Milgram there exists a unique solution to \eqref{weakfor2}. Let us notice that from the uniqueness of $w$ there is not reason to asserts that \eqref{weakfor} has a unique solution because of $\widetilde{u}$. However, by assuming that we have two solutions $u_{1}, u_{2}$, then $u_{1}-u_{2}\in H_{0}^{1}(B)$, and one has that 

\begin{align}
\displaystyle\sum_{i,j=1}^{d}\int_{B}a_{ij}(x)\frac{\partial(u_{1}(x)-u_{2}(x))}{\partial x_{i}}\frac{\partial v}{\partial x_{j}}\de x=0.
\end{align}
Choosing $v=u_{1}-u_{2}$, by using the ellipticity condition , we deduce that $u_{1}=u_{2}$. Furthermore, Poincaré inequality implies that there exists a positive constant $C>0$ such that 

\begin{align}
\norma{u}_{1,B}\leq C\left(\norma{f}_{L^{2}(B)}+ \norma{g}_{\frac{1}{2},\partial B}\right).   
\end{align}
\subsection{Variational loss function}
Let us consider the variational formulation \eqref{weakfor} of \eqref{linearPDE}. We want to fix a integer constant $M>0$, and we consider a finite dimensional space space $\H_{M}$ of dimension $M$ of test functions $v\in H_{0}^{1}(B)$, namely, 

\begin{align}
\H_{M}\coloneqq\mathrm{span}\left\{v_{i}:i=1,\ldots,M\right\}.   
\end{align}
We define the terms

\begin{align}
 \cE(\Theta,v)\coloneqq \displaystyle\sum_{i,j=1}^{d}\int_{B}a_{ij}(x)\frac{\partial u(\Theta,\cdot)}{\partial x_{i}}\frac{\partial v}{\partial x_{j}}\de x,\hskip 0,2cm \F(v)\coloneqq \int_{B}fv\de x.
\end{align}
where $u(\Theta,\cdot)$ is the quantum neural network. In what follows, we then set 

\begin{align}
\L(\Theta)\coloneqq \frac{1}{2}\sum_{i=1}^{M}\left\vert \cE(\Theta,v_{i})-\F(v_{i})\right\vert^{2}+\frac{1}{2}\sum_{i=1}^{n_{2}}\left(u(\Theta,\hat{x}^{(i)})-g(\hat{x}^{(i)})\right)^{2}.   
\end{align}
Let us then introduce the notation $w=(v,\hat{x})$, and consider the vector of observables
\begin{align}
\U(\Theta,w)\coloneqq
\begin{pmatrix}
\cE(\Theta,v)\\
u(\Theta,\hat{x})
\end{pmatrix},
\end{align}
and 
\begin{align}
\U(\Theta,W)\coloneqq
\begin{pmatrix}
\cE(\Theta,v_{1})\\
\vdots\\
\cE(\Theta,v_{M})\\
u(\Theta,\hat{x}^{(1)})\\
\vdots\\
u(\Theta,\hat{x}^{(n_{2})})
\end{pmatrix}.
\end{align}
Similarly, we define
\begin{align}
\hat{y}=\begin{pmatrix}
\F(v)\\
g(\hat{x})\\
\end{pmatrix},
\end{align}
and we set

\begin{align}
\hat{Y}=\begin{pmatrix}
\F(v_{1})\\
\vdots\\
\F(v_{M})\\
g(\hat{x}^{(1)})\\
\vdots\\
g(\hat{x}^{(n_{2})})
\end{pmatrix}.
\end{align}
The loss function can then be written as
\begin{align}
\L(\Theta)=\frac{1}{2}\|\U(\Theta,W)-Y\|_{2}^2 .
\end{align}
We consider the gradient flow dynamics

\begin{align}
\frac{\de\Theta_t}{\de t}=-\eta\nabla_\Theta \L(\Theta_t).
\end{align}
Using the chain rule we obtain
\begin{align}
\nabla_\Theta \L(\Theta)=\nabla_\Theta \U(\Theta,W)^T(\U(\Theta,W)-Y).
\end{align}
Hence the parameters evolve according to
\begin{align}
\frac{\de\Theta_t}{\de t}=-\eta\nabla_\Theta \U(\Theta_t,W)^T(\U(\Theta_t,W)-Y).
\end{align}
Differentiating $\U(\Theta_t)$ with respect to time yields
\begin{align}
\frac{\de}{\de t}\U(\Theta_t,W)=\nabla_\Theta \U(\Theta_t,W)\frac{\de\Theta_t}{\de t}.
\end{align}
Substituting the gradient flow equation gives
\begin{align}
\frac{\de}{\de t}\U(\Theta_t)=-\eta\nabla_\Theta \U(\Theta_t,W)\nabla_\Theta \U(\Theta_t,W)^T
(\U(\Theta_t)-Y).
\end{align}
\subsection{Variational empirical NTK}
Let $w=(v,\hat{x})$, and $w'=(v',\hat{x}')$. We define the variational ENTK as
\begin{align}\label{VNTK}
\hat{K}_{V}(\Theta,w,w')\coloneqq \frac{1}{b_{K}(m)}
\begin{pmatrix}
(\nabla_\Theta \cE(\Theta,v))^T
(\nabla_\Theta \cE(\Theta,v')) & (\nabla_\Theta \cE(\Theta,v))^T(\nabla_\Theta u(\Theta,\hat{x}'))\\
(\nabla_\Theta u(\Theta,\hat{x}))^{T}(\nabla_\Theta \cE(\Theta,v'))& (\nabla_\Theta u(\Theta,\hat{x})^T(\nabla_\Theta u(\Theta,\hat{x}'))
\end{pmatrix}
\end{align}
where $b_{K}(m)$ is the normalizing constant considered in \autoref{def:ENTK}. Notice that \eqref{VNTK} differs from \eqref{def:ENTK} because of $\cE(v)$. Let now consider the counterpart of \autoref{A3} when considering \eqref{VNTK}.

\begin{ass}\label{VA1}
Let us define $K_{V}(z,z')\coloneqq \mathbb E_{\Theta_0}[\widehat{K}_{V}(\Theta_{0},w,w')]$. We suppose that ${\rm diag}\left(\E[\U(\Theta,w)(\U(\Theta,w))^{T}]\right)$ is a positive matrix, and 
\begin{align*}
 &\max_{w\in \H_{M}\times \partial B}{\rm diag}\left(\E[\U(\Theta,w)(\U(\Theta,w))^{T}]\right)=1.
\end{align*}
\end{ass}
Let us further assume the following:
\begin{ass}\label{VA2}
We assume that the finite matrix $K_{V}\coloneqq K(W,W^T)\in\mathbb R^{(M+n_2)\times(M+n_2)}$ has strictly positive minimum eigenvalue $\lambda_{\min}^{K_{V}}$. We also denote by $\lambda_{\max}^{K_{V}}$ its maximum eigenvalue.
\end{ass}
\begin{ass}\label{VA3}
We assume that $\norma{v_{i}}_{1,B}\leq 1$ for all $i=1,\ldots,M$.    
\end{ass}
In what follows, for each $k=1,\ldots,m$ we set
\begin{align}
 \cE_{k}(\Theta,v)\coloneqq \displaystyle\sum_{i,j=1}^{d}\int_{B}a_{ij}(x)\frac{\partial u_{k}(\Theta,\cdot)}{\partial x_{i}}\frac{\partial v}{\partial x_{j}}\de x,\hskip 0,2cm \F(v)\coloneqq \int_{B}fv\de x.
\end{align}
\begin{lem}
Let us assume \autoref{A5}, and \autoref{VA3}. Then for each $v_{i}$, with $i=1,\ldots,M$, and $k=1,\ldots,m$, we have
\begin{align}\label{pt:V1}
&\vert\cE_{k}(\Theta,v_{i})\vert \leq 2d A_{0}({\rm vol}(B))^{\frac{1}{2}}L\\\label{pt:V2}
&\vert\partial_{\theta_{s}}\cE_{k}(\Theta,v_{i})\vert \leq 4d A_{0}({\rm vol}(B))^{\frac{1}{2}}L.\\\label{pt:V3}
&\vert\partial_{\theta_{s'}}\partial_{\theta_{s}}\cE_{k}(\Theta,v_{i})\vert\leq 8d A_{0}({\rm vol}(B))^{\frac{1}{2}}L.
\end{align}
\end{lem}
\begin{proof}
 Let us fix some $v\in \H_{M}$ such that $\norma{v}_{1,B}\leq 1$. Notice that by H\hol{o}lder inequality and \autoref{A5}, we have
 \begin{align}
 \vert\cE_{k}(\Theta,v)\vert\leq\sum_{i,j=1}^{d}A_{0}\left(\int_{B}\vert\partial_{x_{i}}u_{k}(\Theta,x)\vert^{2}\de x\right)^{\frac{1}{2}}\left(\int_{B}\vert\partial_{x_{j}}v\vert^{2}\de x\right)^{\frac{1}{2}} 
 \end{align}
 Recalling that 
 \begin{align*}
\frac{\partial u_{k}(\Theta,x)}{\partial x_{i}}= \bra{0^{m}}\frac{\partial U^{\dagger}(\Theta,x)}{\partial x_{i}}\O_{k}U(\Theta,x)\ket{0^{m}} + \bra{0^{m}}U^{\dagger}(\Theta,x)\O_{k}\frac{\partial U(\Theta,x)}{\partial x_{i}}\ket{0^{m}},
\end{align*}
and that 
\begin{align*}
\left\vert \frac{\partial U(\Theta,x)}{\partial x_{i}}\right\vert\leq L,
\end{align*}
we get that 

\begin{align}
\vert\partial_{x_{i}}u_{k}(\Theta,x)\vert\leq 2L,  
\end{align}
and thus

\begin{align}
\vert\cE_{k}(\Theta,v)\vert&\leq 2d A_{0}({\rm vol}(B))^{\frac{1}{2}} L\sum_{j=1}^{d}\left(\int_{B}\vert\partial_{x_{j}}v\vert^{2}\de x\right)^{\frac{1}{2}}\\ 
&\leq 2d A_{0}({\rm vol}(B))^{\frac{1}{2}}L
\end{align}
where ${\rm vol}(B)$ denotes the volume of $B$ with respect to the $d$-dimensional Lebesgue measure. Let us recall that by the parameter-shift rule \eqref{shiftrule}, one has
\begin{align}
\frac{\partial u_{k}(\Theta,x)}{\partial \theta_{s}}=u_{k}(\Theta+\Delta^{(s)},x)-u_{k}(\Theta-\Delta^{(s)},x).
\end{align}
Then
\begin{align}
&\vert\partial_{\theta_{s}}\cE_{k}(\Theta,v_{i})\vert \leq 4d A_{0}({\rm vol}(B))^{\frac{1}{2}}L\\
&\vert\partial_{\theta_{s'}}\partial_{\theta_{s}}\cE_{k}(\Theta,v_{i})\vert \leq 8d A_{0}({\rm vol}(B))^{\frac{1}{2}}L.
\end{align}
\end{proof}
Let us now prove that $\nabla_{\Theta}\cE(\Theta,v)$ is Lipschitz.
\begin{lem}\label{lipschitz_model_V}
Let us assume \autoref{A5}, and \autoref{VA3}. Then for each $v\in \H_{M}$, with $\norma{v}_{1,B}\leq 1$ one has
\begin{align}\label{lip:VA}
\norma{\nabla_{\Theta}\cE(\Theta,v)-\nabla_{\Theta}\cE(\Theta',v)}_{\infty}&\le \frac{8d A_{0}({\rm vol}(B))^{\frac{1}{2}}L\vert \mathcal{M}\vert^{2}\vert\mathcal{N}\vert}{b(m)}\norma{\Theta-\Theta'}_{\infty}.
\end{align}
\end{lem}
\begin{proof}
Let us set $h_{v}(\Theta)\coloneqq\nabla_{\Theta}\cE(\Theta,v)$, and consider its differential $\de h_{v}(\Theta): (\R^{\vert \Theta\vert},\ell^{\infty})\rightarrow (\R^{\vert\Theta\vert},\ell^{\infty})$ with operator norm given by

\begin{align*}
\norma{\de h_{v}(\Theta)}_{\ell^{\infty}\rightarrow\ell^{\infty}}= \sup_{\norma{z}_{\infty}\leq 1}\norma{\de h_{v}(\Theta)z}_{\infty}. 
\end{align*} 
Let us observe that 

\begin{align}\label{vseconder}
\begin{aligned}
\sup_{\norma{z}_{\infty}\leq 1}\norma{\de h_{v}(\Theta)z}_{\infty}&=\sup_{\norma{z}_{\infty}\leq 1}\norma{\sum_{s=1}^{\vert \Theta\vert}\partial_{\theta_{s}}\nabla_{\Theta} \cE(\Theta,v) z_{s}}_{\infty}\\
&\leq\sup_{\vert z_{s}\vert\leq 1}\max_{1\leq s'\leq \vert \Theta\vert}\sum_{s=1}^{\vert\Theta\vert}\left\vert \partial _{\theta_{s}}\partial_{\theta_{s'}}\cE(\Theta,v) z_{s}\right\vert\\
&\leq\max_{1\leq s'\leq \vert \Theta\vert}\sum_{s=1}^{\vert\Theta\vert}\frac{1}{b(m)}\sum_{k\in \mathcal{M}_{s}\cap \mathcal{M}_{s'}}\left\vert \partial _{\theta_{s}}\partial_{\theta_{s'}}\cE_{k}(\Theta,v)\right\vert\\
&\leq \frac{8d A_{0}({\rm vol}(B))^{\frac{1}{2}}L}{b(m)}\max_{1\leq s'\leq \vert \Theta\vert}\sum_{s=1}^{\vert \Theta\vert} \vert \mathcal{M}_{s}\cap \mathcal{M}_{s'}\vert\\
&\leq \frac{8d A_{0}({\rm vol}(B))^{\frac{1}{2}}L\vert \mathcal{M}\vert^{2}\vert\mathcal{N}\vert}{b(m)}
\end{aligned}
\end{align}
where in the last inequality we have used \eqref{stimafilippo}. Since the domain $\cP$ of $\Theta$ is convex, then we have that 
\begin{align}
\norma{\nabla_{\Theta}\cE(\Theta,v)-\nabla_{\Theta}\cE(\Theta',v)}_{\infty}&\le \frac{8d A_{0}({\rm vol}(B))^{\frac{1}{2}}L\vert \mathcal{M}\vert^{2}\vert\mathcal{N}\vert}{b(m)}\norma{\Theta-\Theta'}_{\infty}.
\end{align}
\end{proof}
\begin{lem}\label{lem:VPDEres}
Let us assume \autoref{A5}, and \autoref{VA3}. Then for each $v\in \H_{M}$, with $\norma{v}_{1,B}\leq 1$ one has
\begin{align}
\left\vert \cE (\Theta,v)-\cE(\Theta',v)\right\vert\le \frac{8d A_{0}({\rm vol}(B))^{\frac{1}{2}}L^{2}\vert \mathcal{M}\vert}{b(m)}\norma{\Theta-\Theta'}_{\infty}.    
\end{align}
\end{lem}
\begin{proof}
Notice that 
\begin{align}
\left\vert \cE(\Theta,v)-\cE(\Theta',v)\right\vert \le \max_{\Theta\in \cP}\norma{\de \cE(\Theta,v)}_{\mathcal{L}}\norma{\Theta-\Theta'}_{\infty},  
\end{align}
with 
\begin{align*}
\max_{\Theta\in \cP}\norma{\de \cE(\Theta,v)}_{\mathcal{L}}&=\sup_{\norma{z}_{\infty}\leq 1}\left\vert\sum_{s=1}^{\vert \Theta\vert}\partial_{\theta_{s}}\cE(\Theta,v)z_{s}\right\vert\\
&\leq \sum_{s=1}^{\vert \Theta\vert}\sum_{k\in \mathcal{M}_{s}}\frac{1}{b(m)}\left\vert \partial_{\theta_{s}}\cE_{k}(\Theta,v)\right\vert\\
&\leq \frac{8d A_{0}({\rm vol}(B))^{\frac{1}{2}}L}{b(m)}\sum_{s=1}^{\vert \Theta\vert}\vert \mathcal{M}_{s}\vert\\
&\leq \frac{8d A_{0}({\rm vol}(B))^{\frac{1}{2}}L^{2}\vert \mathcal{M}\vert}{b(m)}.
\end{align*}
\end{proof}
In what follows, we state a simlar Lipschitzness bound as in \autoref{lem:lipschitzntk} for the VNTK. 
\begin{lem}[Lipschitzness of the VNTK]\label{lem:Vlipschitzntk}
Let $w=(v,\hat{x})$, and $w'=(v',\hat{x}')$. The following holds true:
\begin{align}
\norma{ \hat{K}_{V}(\Theta',w,w')-\hat{K}_{V}(\Theta,w,w')}_{F}\leq \frac{1}{b_{K}(m)}\left(\frac{128d^{2}A_{0}^{2}{\rm vol}(B)mL^{3}\vert\mathcal{M}\vert^{3}\vert\mathcal{N}\vert}{(b(m))^{2}}+\frac{16Lm\vert\mathcal{M}\vert^{3}\vert\mathcal{N}\vert}{(b(m))^{2}}\right)\norma{\Theta-\Theta'}_{\infty}.
\end{align}
\end{lem}
\begin{proof}
  We have that 

  \begin{align}\label{vlimitare}
   \begin{aligned}
      \norma{\hat{K}_{V}(\Theta',w,w')-\hat{K}_{V}(\Theta,w,w')}_{F}^{2}&=\frac{1}{(b_{K}(m))^{2}}\left\vert \nabla_{\Theta}\cE(\Theta,v)\cdot \nabla_{\Theta}\cE(\Theta,v')-\nabla_{\Theta}\cE(\Theta',v)\cdot \nabla_{\Theta}\cE(\Theta',v')\right\vert^{2}\\
      &+\frac{1}{(b_{K}(m))^{2}}\left\vert \nabla_{\Theta}\cE(\Theta,v)\cdot \nabla_{\Theta}u(\Theta,\hat{x}')-\nabla_{\Theta}\cE(\Theta',v)\cdot \nabla_{\Theta}u(\Theta',\hat{x}')\right\vert^{2}\\
      &+\frac{1}{(b_{K}(m))^{2}}\left\vert \nabla_{\Theta}u(\Theta,\hat{x})\cdot \nabla_{\Theta}\cE(\Theta,v')-\nabla_{\Theta} u(\Theta',\hat{x})\cdot \nabla_{\Theta}\cE(\Theta',v')\right\vert^{2}\\
      &+\frac{1}{(b_{K}(m))^{2}}\left\vert \nabla_{\Theta}u(\Theta,\hat{x})\cdot \nabla_{\Theta} u(\Theta,\hat{x}')-\nabla_{\Theta} u(\Theta',\hat{x})\cdot \nabla_{\Theta}u(\Theta',\hat{x}')\right\vert^{2}.
   \end{aligned}   
  \end{align}
Let us now give a bound for the terms of \eqref{vlimitare} involving $\cE$. Notice that

\begin{align}
\begin{aligned}
 &\left\vert \nabla_{\Theta}\cE(\Theta,v)\cdot \nabla_{\Theta}\cE(\Theta,v')-\nabla_{\Theta}\cE(\Theta',x)\cdot \nabla_{\Theta}\cE(\Theta',v')\right\vert \\
 &=\left\vert \left(\nabla_{\Theta}\cE(\Theta',v)-\nabla_{\Theta}\cE(\Theta,v)\right)\cdot \nabla_{\Theta}\cE(\Theta',v')+ \right.\\
    &\left. +\nabla_{\Theta}\cE(\Theta,v)\cdot\left(\nabla_{\Theta}\cE(\Theta',v')-\nabla_{\Theta}\cE(\Theta,v')\right)\right\vert 
\end{aligned}    
\end{align}
and thus
\begin{align}
\begin{aligned}
&\left\vert \left(\nabla_{\Theta}\cE(\Theta',v)-\nabla_{\Theta}\cE(\Theta,v)\right)\cdot \nabla_{\Theta}\cE(\Theta',v')+\nabla_{\Theta}\cE(\Theta,v)\cdot\left(\nabla_{\Theta}\cE(\Theta',v')-\nabla_{\Theta}\cE(\Theta,v')\right)\right\vert\\ 
&\quad\leq \norma{\nabla_{\Theta}\cE(\Theta',v)-\nabla_{\Theta}\cE(\Theta,v)}_{\infty}\norma{ \nabla_{\Theta}\cE(\Theta',v')}_{1}\\
&\quad\quad+\norma{\nabla_{\Theta}\cE(\Theta',v')-\nabla_{\Theta}\cE(\Theta,v')}_{\infty}\norma{\nabla_{\Theta}\cE(\Theta,v)}_{1}\\
&\leq\frac{8d A_{0}({\rm vol}(B))^{\frac{1}{2}}L\vert \mathcal{M}\vert^{2}\vert\mathcal{N}\vert}{b(m)}\norma{\Theta-\Theta'}_{\infty}\left(\norma{\nabla_{\Theta}\cE(\Theta,v)}_{1}+\norma{ \nabla_{\Theta}\cE(\Theta',v')}_{1}\right).
   \end{aligned}   
  \end{align}
  Notice that 
  \begin{align}
  \begin{aligned}
    \norma{\nabla_{\Theta}\cE(\Theta,v)}_{1}&=\sum_{s=1}^{Lm}\vert\partial_{\theta_{s}}\cE(\Theta,v) \vert\\
    &\leq \frac{1}{b(m)}\sum_{s=1}^{Lm}\sum_{k\in \mathcal{M}_{s}}\vert\partial_{\theta_{s}}\cE_{k}(\Theta,v) \vert\\
    &\leq \frac{4d A_{0}({\rm vol}(B))^{\frac{1}{2}}mL^{2}\vert\mathcal{M}\vert}{b(m)}.
    \end{aligned}
  \end{align}
  Thus,

  \begin{align}
  \begin{aligned}
    &\left\vert \left(\nabla_{\Theta}\cE(\Theta',v)-\nabla_{\Theta}\cE(\Theta,v)\right)\cdot \nabla_{\Theta}\cE(\Theta',v')+\nabla_{\Theta}\cE(\Theta,v)\cdot\left(\nabla_{\Theta}\cE(\Theta',v')-\nabla_{\Theta}\cE(\Theta,v')\right)\right\vert\\ 
    &\quad\quad\leq\frac{64d^{2}A_{0}^{2}{\rm vol}(B)mL^{3}\vert\mathcal{M}\vert^{3}\vert\mathcal{N}\vert}{(b(m))^{2}}.
  \end{aligned}    
  \end{align}
  Recall that
  \begin{align}
  \begin{aligned}
\left\vert \nabla_{\Theta}u(\Theta,\hat{x})\cdot \nabla_{\Theta} u(\Theta,\hat{x}')-\nabla_{\Theta} u(\Theta',\hat{x})\cdot \nabla_{\Theta}u(\Theta',\hat{x}')\right\vert\leq \frac{16Lm\vert\mathcal{M}\vert^{3}\vert\mathcal{N}\vert}{(b(m))^{2}}.  
  \end{aligned}
  \end{align}
 Now, let us observe that

\begin{align}
\begin{aligned}
 &\left\vert \nabla_{\Theta}\cE(\Theta,v)\cdot \nabla_{\Theta}u(\Theta,\hat{x}')-\nabla_{\Theta}\cE(\Theta',v)\cdot \nabla_{\Theta} u(\Theta',\hat{x}')\right\vert \\
 &=\left\vert \left(\nabla_{\Theta}\cE(\Theta',v)-\nabla_{\Theta}\cE(\Theta,v)\right)\cdot \nabla_{\Theta} u(\Theta',\hat{x}')+ \right.\\
    &\left. +\nabla_{\Theta}\cE(\Theta,v)\cdot\left(\nabla_{\Theta}u(\Theta',\hat{x}')-\nabla_{\Theta}u(\theta,\hat{x}')\right)\right\vert \\
    &\leq \norma{\nabla_{\Theta}\cE(\Theta',v)-\nabla_{\Theta}\cE(\Theta,v)}_{\infty}\norma{\nabla_{\Theta} u(\Theta',\hat{x}')}_{1}+\\
    &+\norma{\nabla_{\Theta}u(\Theta',\hat{x}')-\nabla_{\Theta}u(\Theta,\hat{x}')}_{\infty}\norma{\nabla_{\Theta}\cE(\Theta,v)}_{1}\\
    &\leq \frac{8d A_{0}({\rm vol}(B))^{\frac{1}{2}}L\vert \mathcal{M}\vert^{2}\vert\mathcal{N}\vert}{b(m)}\norma{\Theta-\Theta'}_{\infty}\norma{\nabla_{\Theta} u(\Theta',\hat{x}')}_{1}+\\
    &+\norma{\nabla_{\Theta}u(\Theta',\hat{x}')-\nabla_{\Theta}u(\Theta,\hat{x}')}_{\infty}\frac{4d A_{0}({\rm vol}(B))^{\frac{1}{2}}mL^{2}\vert\mathcal{M}\vert}{b(m)}.
\end{aligned}    
\end{align}
 
Recall that 

\begin{align}
\begin{aligned}
&\norma{\nabla_{\Theta} u(\Theta',\hat{x}')}_{1}\leq \frac{2Lm\vert\mathcal{M}\vert}{b(m)}\\
&\norma{\nabla_{\Theta}u(\Theta',\hat{x}')-\nabla_{\Theta}u(\Theta,\hat{x}')}_{\infty}\leq \frac{4\vert\mathcal{M}\vert^2\vert \mathcal{N}\vert}{b(m)}\norma{\Theta-\Theta'}_{\infty}.
\end{aligned}
\end{align}
Then we have that

\begin{align}
\begin{aligned}
&\left\vert \nabla_{\Theta}\cE(\Theta,v)\cdot \nabla_{\Theta}u(\Theta,\hat{x}')-\nabla_{\Theta}\cE(\Theta',v)\cdot \nabla_{\Theta} u(\Theta',\hat{x}')\right\vert \\
&\leq \frac{16d({\rm vol}(B))^{\frac{1}{2}}A_{0}mL^{2}\vert \mathcal{M}\vert^{3}\vert\mathcal{N}\vert}{(b(m))^{2}}\norma{\Theta-\Theta'}_{\infty}+\\
    &+\frac{16d({\rm vol}(B))^{\frac{1}{2}}A_{0}mL^{2}\vert\mathcal{M}\vert^{3}\vert \mathcal{N}\vert}{(b(m))^{2}}\norma{\Theta-\Theta'}_{\infty}\\
    &=\frac{32d({\rm vol}(B))^{\frac{1}{2}}A_{0}mL^{2}\vert\mathcal{M}\vert^{3}\vert \mathcal{N}\vert}{(b(m))^{2}}\norma{\Theta-\Theta'}_{\infty}
\end{aligned}   
\end{align}
We conclude that

  \begin{align}
  \begin{aligned}
  \norma{\hat{K}_{V}(\Theta',w,w')-\hat{K}_{V}(\Theta,w,w')}_{F}^{2}
  &\leq \frac{1}{(b_{K}(m))^{2}}\left(\left(\frac{64d^{2}A_{0}^{2}{\rm vol}(B)mL^{3}\vert\mathcal{M}\vert^{3}\vert\mathcal{N}\vert}{(b(m))^{2}}\right)^{2}+\left(\frac{16Lm\vert\mathcal{M}\vert^{3}\vert\mathcal{N}\vert}{(b(m))^{2}}\right)^{2}+\right.\\
  &\left.+2\left(\frac{32d({\rm vol}(B))^{\frac{1}{2}}A_{0}mL^{2}\vert\mathcal{M}\vert^{3}\vert \mathcal{N}\vert}{(b(m))^{2}}\right)^{2}\right)\norma{\Theta-\Theta'}_{\infty}^{2}.
  \end{aligned}  
  \end{align}
  Therefore, since $\sqrt{a+b}\leq \sqrt{a}+\sqrt{b}$ for all $a,b\geq 0$, we get
\begin{align}
  \begin{aligned}
  \norma{\hat{K}_{V}(\Theta',w,w')-\hat{K}_{V}(\Theta,w,w')}_{F}
  &\leq \frac{1}{b_{K}(m)}\left(\frac{64d^{2}A_{0}^{2}{\rm vol}(B)mL^{3}\vert\mathcal{M}\vert^{3}\vert\mathcal{N}\vert}{(b(m))^{2}}+\frac{16Lm\vert\mathcal{M}\vert^{3}\vert\mathcal{N}\vert}{(b(m))^{2}}+\right.\\
  &\left.+\frac{(\sqrt{2})32d({\rm vol}(B))^{\frac{1}{2}}A_{0}mL^{2}\vert\mathcal{M}\vert^{3}\vert \mathcal{N}\vert}{(b(m))^{2}}\right)\norma{\Theta-\Theta'}_{\infty}.
  \end{aligned}  
  \end{align}
  and then

  \begin{align}
  \begin{aligned}
  \norma{\hat{K}_{\Theta'}(w,w')-\hat{K}_{\Theta}(w,w')}_{F}
  &\leq \frac{1}{b_{K}(m)}\left(\frac{128d^{2}A_{0}^{2}{\rm vol}(B)mL^{3}\vert\mathcal{M}\vert^{3}\vert\mathcal{N}\vert}{(b(m))^{2}}+\frac{16Lm\vert\mathcal{M}\vert^{3}\vert\mathcal{N}\vert}{(b(m))^{2}}\right)\norma{\Theta-\Theta'}_{\infty}.
  \end{aligned}  
  \end{align}
\end{proof}
We are position to state the following concentration phenomenon for the VNTK as in the case of the ENTK in \autoref{thm:NTK}.

\begin{thm}[VNTK concentration]\label{thm:VNTK}
Let $B\subset \R^{d}$ be an open bounded subset of $\R^{d}$ with Lipschitz boundary $\partial B$, and $d\geq 1$ the dimensionality $\R^{d}$. Assume that hypotheses \autoref{A1}, \autoref{A5}, \autoref{VA1}--\autoref{VA3} hold true. Then for any $w,w'\in \H_{M}\times \partial B$, it holds that
\begin{align}
\P\left[ \norma{\hat{K}_{V}(\Theta,w,w')-K_{V}(w,w')}_{F}\right]\geq \varepsilon]\leq \exp\left[
-\frac{\varepsilon^{2}(b_{K}(m))^{2}(b(m))^{4}}
{4(64)^{2}d^{4}A_{0}^{4}({\rm vol}(B))^{2}mL^{5}\vert\mathcal{M}\vert^{4}\vert\mathcal{N}\vert^{2}}
\right].
\end{align}
Furthermore, assume that
\begin{align}\label{Var_ipoconvergencentk}
\lim_{m\rightarrow+\infty}\frac{mL^{5}\vert\mathcal{M}\vert^{4}\vert\mathcal{N}\vert^{2}}{(b(m))^{4}}=0.    
\end{align}
Then the VNTK converges in probability to the analytic VNTK as $m\rightarrow +\infty$.
\end{thm}
\begin{proof}
Let us follows the argument provided when we proved \autoref{thm:NTK}. Recall that
\begin{align}
\P\left[\varepsilon\leq\norma{\hat{K}_{V}(\Theta,w,w')-K(w,w')}_{F}\right]=\P\left[\varepsilon^{2}\leq\norma{\hat{K}_{\Theta}(z,z')-K(z,z')}_{F}^{2}\right],
\end{align}
where
\begin{align}
  \norma{\hat{K}_{V}(\Theta,w,w')-K(w,w')}_{F}^{2}=\sum_{i,j=1}^{2}\left((\hat{K}_{V}(\Theta,w,w'))_{i,j}-(K(w,w'))_{i,j}\right)^{2}.  
\end{align}
Hence, we have that 

\begin{align}\label{Vbounddafa}
\P\left[\varepsilon^{2}\leq\norma{\hat{K}_{V}(\Theta,w,w')-K(w,w')}_{F}^{2}\right]\leq    \sum_{i,j=1}^{2}\P\left[\frac{\varepsilon^{2}}{4}\leq \left((\hat{K}_{V}(\Theta,w,w'))_{i,j}-(K(w,w'))_{i,j}\right)^{2}\right].
\end{align}
In what follows, we estimate the right-hand side of \eqref{Vbounddafa} by using McDiarmid's concentration inequality. Observe that 

\begin{align}
\begin{aligned}
(\hat{K}_{V}(\Theta,w,w'))_{11}&= \frac{1}{b_{K}(m)}\nabla_{\Theta}\cE(\Theta,v)\cdot\nabla_{\Theta}\cE(\Theta,v)\\
&=\frac{1}{b_{K}(m) (b(m))^{2}}\sum_{k,k'=1}^{m}\sum_{s=1}^{\vert\Theta\vert}\partial_{\theta_{s}}\cE_{k}(\Theta,v)\partial_{\theta_{s}}\cE_{k'}(\Theta,v)\\
&=\frac{1}{b_{K}(m) (b(m))^{2}}\sum_{k,k'=1}^{m}\sum_{s\in \mathcal{N}_{k}\cap\mathcal{N}_{k'}}\partial_{\theta_{s}}\cE_{k}(\Theta,v)\partial_{\theta_{s}}\cE_{k'}(\Theta,v),
\end{aligned}
\end{align}
and 
\begin{align}
\begin{aligned}
(\hat{K}_{V}(\Theta,w,w'))_{1,2}&=\frac{1}{b_{K}(m)}\nabla_{\Theta}\cE(\Theta,v)\cdot \nabla_{\Theta}u(\Theta,\hat{x})\\
&=\frac{1}{b_{K}(m) (b(m))^{2}}\sum_{k,k'=1}^{m}\sum_{s\in \mathcal{N}_{k}\cap\mathcal{N}_{k'}}\partial_{\theta_{s}}\cE_{k}(\Theta,v)\partial_{\theta_{s}}u_{k'}(\Theta,\hat{x}').
\end{aligned}  
\end{align}
By considering $\Psi_{i}$ as defined in \eqref{defpsi}, we have
\begin{align}
\begin{aligned}
&(\hat{K}_{V}(\Theta,w,w'))_{1,1}-(\hat{K}_{V}(\Theta',w,w'))_{1,1}=\\
&=\frac{1}{b_{K}(m)(b(m))^{2}}\sum_{(k,k',s)\in \Psi_{i}}\left[\partial_{\theta_{s}}\cE_{k}(\Theta,v)\partial_{\theta_{s}}\cE_{k'}(\Theta,v)-\partial_{\theta_{s}}\cE_{k}(\Theta',v)\partial_{\theta_{s}}\cE_{k'}(\Theta',v)\right]. 
\end{aligned}    
\end{align}
Then, by considering \eqref{pt:V2}, one gets
\begin{align}
    \begin{aligned}
    \left\vert (\hat{K}_{V}(\Theta,w,w'))_{1,1}-(\hat{K}_{V}(\Theta,w,w'))_{1,1}\right\vert\leq \frac{2(4d A_{0}({\rm vol}(B))^{\frac{1}{2}}L)^{2}\vert \Psi_{i}\vert}{b_{K}(m)(b(m))^{2}}.  
    \end{aligned}
\end{align}
Since $\vert \Psi_{i}\vert\leq 2\vert\mathcal{M}_{i}\vert\vert\mathcal{M}\vert\vert\mathcal{N}\vert$, one has

\begin{align}
\begin{aligned}
\left\vert (\hat{K}_{V}(\Theta,w,w'))_{1,1}-(\hat{K}_{V}(\Theta',w,w'))_{1,1}\right\vert&\leq \frac{\vert\mathcal{M}_{i}\vert\vert\mathcal{M}\vert\vert\mathcal{N}\vert}{b_{K}(m)(b(m))^{2}}64d^{2}A_{0}^{2}{\rm vol}(B)L^{2}\\
&\leq \frac{64d^{2}A_{0}^{2}{\rm vol}(B)L^{2}\vert\mathcal{M}\vert^{2}\vert\mathcal{N}\vert}{b_{K}(m)(b(m))^{2}}.
\end{aligned}   
\end{align}
Hence, by using the McDiarmid's concentration inequality as stated in \autoref{thm:McDiarmid}, we have
\begin{align}
c_{i}\coloneqq\frac{64d^{2}A_{0}^{2}{\rm vol}(B)L^{2}\vert\mathcal{M}\vert^{2}\vert\mathcal{N}\vert}{b_{K}(m)(b(m))^{2}}, 
\end{align}
and for any $\delta>0$

\begin{align}
\P\left[\vert(\hat{K}_{V}(\Theta,w,w'))_{11}-\E[(\hat{K}_{V}(\Theta,w,w'))_{1,1}]\vert\geq \delta\right]\leq \exp\left[\frac{-2\delta^{2}}{\sum_{i=1}^{\vert \Theta\vert}c_{i}^{2}}\right].  
\end{align}
Notice that
\begin{align}
\sum_{i=1}^{\vert \Theta\vert}c_{i}^{2}= \frac{(64)^{2}d^{4}A_{0}^{4}({\rm vol}(B))^{2}mL^{5}\vert\mathcal{M}\vert^{4}\vert\mathcal{N}\vert^{2}}{(b_{K}(m))^{2}(b(m))^{4}}.
\end{align}
Hence, we conclude that for $\delta=\frac{\varepsilon}{2}$
\begin{align}
\P\left[ \vert (\hat{K}_{V}(\Theta,w,w'))_{1,1}-\E[(\hat{K}_{V}(\Theta,w,w'))_{1,1}]\vert\geq \frac{\varepsilon}{2}\right]\leq \exp\left[\frac{-2\varepsilon^{2}(b_{K}(m))^{2}(b(m))^{4}}{4(64)^{2}d^{4}A_{0}^{4}({\rm vol}(B))^{2}mL^{5}\vert\mathcal{M}\vert^{4}\vert\mathcal{N}\vert^{2}}\right].
\end{align}
On the other hand, we recall that
\begin{align}
\P\left[ \vert (\hat{K}_{V}(\Theta,w,w'))_{2,2}-\E[(\hat{K}_{V}(\Theta,w,w'))_{2,2}]\vert\geq \frac{\varepsilon}{2}\right]\leq \exp\left[\frac{-2\varepsilon^{2}(b_{K}(m))^{2}(b(m))^{4}}{4(16)^{2}mL\vert\mathcal{M}\vert^{4}\vert\mathcal{N}\vert^{2}}\right].
\end{align}
Similarly, 

\begin{align}
\begin{aligned}
&(\hat{K}_{V}(\Theta,w,w'))_{1,2}-(\hat{K}_{V}(\Theta',w,w'))_{1,2}=\\
&=\frac{1}{b_{K}(m)(b(m))^{2}}\sum_{(k,k',s)\in \Psi_{i}}\left[\partial_{\theta_{s}}\cE_{k}(\Theta,v)\partial_{\theta_{s}}u_{k'}(\Theta,\hat{x}')-\partial_{\theta_{s}}\cE_{k}(\Theta',v)\partial_{\theta_{s}}u_{k'}(\Theta',\hat{x}')\right]. 
\end{aligned}    
\end{align}
Then
\begin{align}\label{vmidterm}
    \begin{aligned}
    \left\vert (\hat{K}_{W}(\Theta,w,w'))_{1,2}-(\hat{K}_{W}(\Theta',w,w'))_{1,2}\right\vert&\leq \frac{8(4d A_{0}({\rm vol}(B))^{\frac{1}{2}}L)\vert \Psi_{i} \vert}{b_{K}(m)(b(m))^{2}}\\
    &\leq \frac{64d A_{0}({\rm vol}(B))^{\frac{1}{2}}L\vert\mathcal{M}\vert^{2}\vert\mathcal{N}\vert}{b_{K}(m)(b(m))^{2}}
    \end{aligned}
\end{align}
where we have used $\vert\partial_{\theta_{s}}u_{k}(\Theta,\cdot) \vert\leq 2$,  and in the last inequality \eqref{pt:V2}.  We conclude that 

\begin{align}
\P\left[ \vert (\hat{K}_{V}(\Theta,w,w'))_{1,2}-\E[(\hat{K}_{V}(\Theta,w,w'))_{1,2}]\vert\geq \frac{\varepsilon}{2}\right]\leq \exp\left[\frac{-2\varepsilon^{2}(b_{K}(m))^{2}(b(m))^{4}}{4(64)^{2}A_{0}^{2}d^{2}{\rm vol}(B)mL^{3}\vert\mathcal{M}\vert^{4}\vert\mathcal{N}\vert^{2}}\right].
\end{align}
Using the same argument of \eqref{vmidterm}, we have

\begin{align}
    \left\vert (\hat{K}_{W}(\Theta,w,w'))_{2,1}-(\hat{K}_{W}(\Theta',w,w'))_{2,1}\right\vert\leq \frac{64d A_{0}({\rm vol}(B))^{\frac{1}{2}}L\vert\mathcal{M}\vert^{2}\vert\mathcal{N}\vert}{b_{K}(m)(b(m))^{2}},
\end{align}
and thus
\begin{align}
\P\left[ \vert (\hat{K}_{V}(\Theta,w,w'))_{2,1}-\E[(\hat{K}_{V}(\Theta,w,w'))_{2,1}]\vert\geq \frac{\varepsilon}{2}\right]\leq \exp\left[\frac{-2\varepsilon^{2}(b_{K}(m))^{2}(b(m))^{4}}{4(64)^{2}A_{0}^{2}d^{2}{\rm vol}(B)mL^{3}\vert\mathcal{M}\vert^{4}\vert\mathcal{N}\vert^{2}}\right].
\end{align}
Therefore,
\begin{align}\label{V_bounddafa_last}
\begin{aligned}
\P\left[\varepsilon^{2}\leq\norma{\hat{K}_{V}(\Theta,w,w')-K_{V}(w,w')}_{F}^{2}\right]&\leq  \exp\left[\frac{-2\varepsilon^{2}(b_{K}(m))^{2}(b(m))^{4}}{4(64)^{2}d^{4}A_{0}^{4}({\rm vol}(B))^{2}mL^{5}\vert\mathcal{M}\vert^{4}\vert\mathcal{N}\vert^{2}}\right]+\exp\left[\frac{-2\varepsilon^{2}(b_{K}(m))^{2}(b(m))^{4}}{4(16)^{2}mL\vert\mathcal{M}\vert^{4}\vert\mathcal{N}\vert^{2}}\right]\\
&+2\exp\left[\frac{-2\varepsilon^{2}(b_{K}(m))^{2}(b(m))^{4}}{4(64)^{2}A_{0}^{2}d^{2}{\rm vol}(B)mL^{3}\vert\mathcal{M}\vert^{4}\vert\mathcal{N}\vert^{2}}\right]\\
&\leq 4\exp\left[
-\frac{2\varepsilon^{2}(b_{K}(m))^{2}(b(m))^{4}}
{4(64)^{2}d^{4}A_{0}^{4}({\rm vol}(B))^{2}mL^{5}\vert\mathcal{M}\vert^{4}\vert\mathcal{N}\vert^{2}}
\right].
\end{aligned}
\end{align}
Furthermore, for m large enough we can have 

\begin{align}
 4\exp\left[
-\frac{2\varepsilon^{2}(b_{K}(m))^{2}(b(m))^{4}}
{4(64)^{2}d^{4}A_{0}^{4}({\rm vol}(B))^{2}mL^{5}\vert\mathcal{M}\vert^{4}\vert\mathcal{N}\vert^{2}}
\right]\leq \exp\left[
-\frac{\varepsilon^{2}(b_{K}(m))^{2}(b(m))^{4}}
{4(64)^{2}d^{4}A_{0}^{4}({\rm vol}(B))^{2}mL^{5}\vert\mathcal{M}\vert^{4}\vert\mathcal{N}\vert^{2}}
\right],
\end{align}
and we are done.
\end{proof}
\begin{rem}\label{rem:comparison_variational}
We observe several differences between the case using the operator $\A$ approach (Theorems \ref{thm:NTK} and \ref{thm:lazytraining}) and the variational approach (Theorems \ref{thm:VNTK} and \ref{thm:Vlazytraining} below):
\begin{enumerate}
\item The scaling condition for NTK convergence in thecase using the operator $\A$case \eqref{ipoconvergencentk} requires
\begin{align}
\lim_{m\rightarrow+\infty}\frac{mL^{9}\vert\mathcal{M}\vert^{4}\vert\mathcal{N}\vert^{2}}{(b(m))^{4}}=0,
\end{align}
while in the variational case \eqref{Var_ipoconvergencentk} we only need
\begin{align}
\lim_{m\rightarrow+\infty}\frac{mL^{5}\vert\mathcal{M}\vert^{4}\vert\mathcal{N}\vert^{2}}{(b(m))^{4}}=0.
\end{align}
The variational formulation thus requires a weaker growth condition on the depth $L$, improving from $L^9$ to $L^5$. This improvement stems from the fact that in the variational approach, we only need to control first-order spatial derivatives of $u_k(\Theta,x)$ (which grow as $L$), whereas the case using the operator $\A$ approach requires controlling second-order spatial derivatives appearing in $\A u_k(\Theta,x)$ (which grow as $L^2$).

\item In the case using the operator $\A$ approach, we imposed the constraint \eqref{dimensionalconstraint}:
\begin{align}\label{condt}
2A_{1}+4A_{0}\leq 1,
\end{align}
where $A_0=\sup_{x\in\overline B}|a_{ij}(x)|$ and $A_1=\sup_{x\in\overline B}|\partial_{x_k} a_{ij}(x)|$. This constraint was necessary to simplify the Lipschitz constants appearing in \autoref{lem:lipschitzntk} and to ensure that the exponential concentration bounds in \autoref{thm:NTK} have favorable dependence on the problem parameters.

In contrast, the variational formulation does not require assumption \eqref{condt}. To see why, observe that in the variational case, the relevant quantities are:
\begin{align}
\cE_k(\Theta,v)= \sum_{i,j=1}^{d}\int_{B}a_{ij}(x)\frac{\partial u_k(\Theta,x)}{\partial x_{i}}\frac{\partial v}{\partial x_{j}}\de x.
\end{align}
The bounds in \autoref{lipschitz_model_V} and \autoref{lem:Vlipschitzntk} involve only $A_0$ (and not $A_1$), and the key estimate \eqref{lip:VA} reads:
\begin{align}
\norma{\nabla_{\Theta}\cE(\Theta,v)-\nabla_{\Theta}\cE(\Theta',v)}_{\infty}\le \frac{8dA_{0}\sqrt{{\rm vol}(B)}L\vert \mathcal{M}\vert^{2}\vert\mathcal{N}\vert}{b(m)}\norma{\Theta-\Theta'}_{\infty}.
\end{align}
This bound depends linearly on $A_0$ without requiring any constraint relating $A_0$ and $A_1$. Consequently, the coefficients $a_{ij}$ need only satisfy
\begin{align}
\sup_{x\in\overline{B}}|a_{ij}(x)|\leq A_0<\infty,\quad a_{ij}\in C^1(\overline{B}).
\end{align}
In particular, $A_0$ can be arbitrarily large, meaning the coefficients $a_{ij}$ are merely required to be bounded and $C^1$—they are essentially $L^\infty$ functions with bounded first derivatives.

\item The variational bounds involve the volume ${\rm vol}(B)$ of the spatial domain, which appears naturally through the integral formulation. In the case using the operator $\A$ approach, the domain geometry enters implicitly through the choice of training points, but does not appear explicitly in the concentration bounds.

\item The differential case uses $n_1$ interior points and $n_2$ boundary points, leading to $n=n_1+n_2$ training samples. The variational method uses $M$ test functions in $H_0^1(B)$ and $n_2$ boundary points, giving $n_V=M+n_2$ degrees of freedom. The choice of $M$ and the test function space $\H_M$ provides additional flexibility in the variational approach.
\end{enumerate}
\end{rem}

\begin{ass}\label{VA4}
Assume that there exists a deterministic limit matrix kernel of size $2\times 2$ denoted $\overline{K}_{V}(w,w')$ such that
\begin{align}
\lim_{m\to\infty}\sup_{w,w'} \norma {K_{V}(w,w') - \overline{K}_{V}(w,w')}_{F} = 0,    
\end{align}
with $\overline{K}_{V}$ not identically zero, and $\norma{\cdot}_{F}$ denotes the Frobenius norm.
\end{ass}
\begin{ass}\label{VA5}
Assume that
\begin{align*}
    \lim_{m\rightarrow+\infty}\sup_{w,w'\in \H_{M}\times\partial B}\norma{ \E[\U(\Theta,w)(\U(\Theta,w'))^{T}]-\K_{V}(w,w')}_{F}=0,
\end{align*}
where $\K_{V}: (\H_{M}\times \partial B) \times (\H_{M} \times \partial B) \to \mathbb{M}_{2\times 2}$ is a positive semi-definite matrix-valued function with strictly positive diagonal entries for all $w$.  
\end{ass}

\begin{thm}\label{thm:Vconvergence_at_initialization}
Assume hypotheses \autoref{A0}-\autoref{A2}, \autoref{A6} and \autoref{VA1}-\autoref{VA5}. Suppose that
\begin{align}\label{Vconvergencehypothesis}
\lim_{m\to\infty} \frac{m L^{3}\vert \mathcal{M}\vert^{2}\vert \mathcal{N}\vert^{2}}{(b(m))^{3}} = 0 .
\end{align}
Then, for any $w=(v,\hat{x})$,
\begin{align}
\U^{\lin}(\Theta_{t}^{\lin},w)\;\xrightarrow[m\to\infty]{\mathcal D}\;
\U^{(\infty)}(w)-\overline K_V(w,W^T)\,\overline K_V^{-1}\,\big(1-e^{-\eta_0 \overline K_V t}\big)
\big(\U^{(\infty)}(W)-Y\big),
\end{align}
where $\U^{(\infty)}(w)$ is a Gaussian vector with covariance $\K_V$.
\end{thm}
\begin{proof}
In what follows, we only modify some useful notations, and we give the key steps. In particular, notice that $\U(\Theta,w) \to \U^{(\infty)}(w)$ in distribution as $m \to \infty$, with covariance matrix $\K_V$.  Consider two collections of points
\begin{align}
\sigma_{V_1} = (v_\alpha: \alpha \in A_1)\subset \H_M    
\end{align}
and
\begin{align}
\sigma_{A_2} = (\hat{x}_\beta: \beta \in A_2)\subset \partial B.    
\end{align}
 For $\zeta=(\zeta_1,\zeta_2)\in \R^{|A_1|}\times \R^{|A_2|}$, define $\|\zeta\|_1 = \sum_\alpha |\zeta_\alpha| + \sum_\beta |\zeta_\beta|$. The characteristic function of the finite collection of observables $\U(\Theta, \sigma_{V_1}, \sigma_{\hat X_2})$ reads
\begin{align}
\phi_m(\zeta) =\E\Bigg[\exp\Bigg(i\frac{1}{b(m)} \sum_{k=1}^m \Big[ \cE_k(\Theta,\sigma_{V_1})\cdot \zeta_1 + u_k(\Theta,\sigma_{A_2})\cdot \zeta_2 \Big]\Bigg)\Bigg],
\end{align}
where
\begin{align}
\cE_k(\Theta,\sigma_{V_1})\cdot \zeta_1 \coloneqq \sum_{\alpha \in A_1}\cE_k(\Theta,v_\alpha) \zeta_{1,\alpha}, \quad u_k(\Theta,\sigma_{\hat X_2})\cdot \zeta_2 \coloneqq \sum_{\beta\in A_2} u_k(\Theta,\hat{x}_\beta) \zeta_{2,\beta}.
\end{align}

For \(t\ge 0\), define
\begin{align}
\phi_m(\zeta,t) = 
\E\Bigg[\exp\Bigg(
it\frac{1}{b(m)} \sum_{k=1}^m \Big[ \cE_k(\Theta,\sigma_{V_1})\cdot \zeta_1 + u_k(\Theta,\sigma_{A_2})\cdot \zeta_2 \Big]
\Bigg)\Bigg].
\end{align}
Using the cumulant expansion, we have
\begin{align}
\log \phi_m(\zeta,t) = \sum_{r=1}^{\infty} \frac{\kappa_r^{(m)}}{r!} (it)^r, 
\quad \kappa_r^{(m)} = (-i)^r \frac{d^r}{dt^r} \log \phi_m(\zeta,t).
\end{align}
From \eqref{pt:V2}, we obtain
\begin{align}
\Big|\frac{1}{b(m)} \cE_k(\Theta,\sigma_{V_{1}})\cdot \zeta_1\Big| \le \frac{4d A_{0}({\rm vol}(B))^{\frac{1}{2}}L}{b(m)} \|\zeta_1\|_1, \quad
\Big|\frac{1}{b(m)} u_k(\Theta,\sigma_{A_{2}})\cdot \zeta_2 \Big| \le \frac{1}{b(m)} \|\zeta_2\|_1.
\end{align}
Hence,
\begin{align}
\Big| \frac{1}{b(m)} \U_k(\Theta, \sigma_{V_1}, \sigma_{A_2}) \cdot \zeta \Big| \le \frac{4d A_{0}({\rm vol}(B))^{\frac{1}{2}}L\|\zeta\|_1}{b(m)}.
\end{align}
By \autoref{thm_dependencygraph}, for $r\ge 3$,
\begin{align}
|\kappa_r^{(m)}| \le \frac{m}{D} \Big(\frac{4dA_{0}({\rm vol}(B))^{\frac{1}{2}}DL\|\zeta\|_1}{b(m)}\Big)^r r!.
\end{align}
Let
\begin{align}
R(\zeta,t) = \sum_{r=3}^\infty \frac{\kappa_r^{(m)}}{r!} (it)^r.
\end{align}
Notice that

\begin{align}
\begin{aligned}
\vert R(\zeta,t)\vert &\leq \sum_{r=3}^{\infty}\frac{m}{D}\Big(\frac{16 edA_{0}({\rm vol}(B))^{\frac{1}{2}}DL\|\zeta\|_1}{b(m)}\Big)^r\\
&=\frac{m}{D}\Big(\frac{16 edA_{0}({\rm vol}(B))^{\frac{1}{2}}DL\|\zeta\|_1}{b(m)}\Big)^{3}\sum_{r=0}^{\infty}\frac{m}{D}\Big(\frac{16 edA_{0}({\rm vol}(B))^{\frac{1}{2}}DL\|\zeta\|_1}{b(m)}\Big)^r\\
&\leq \frac{m L^{3}\vert \mathcal{M}\vert^{2}\vert \mathcal{N}\vert^{2}}{(b(m))^{3}}\left(16 edA_{0}({\rm vol}(B))^{\frac{1}{2}}t\norma{\zeta}_{1}\right)^{3}\sum_{r=0}^{\infty}\left(\frac{16edA_{0}({\rm vol}(B))^{\frac{1}{2}}L\vert \mathcal{M}\vert\vert \mathcal{N}\vert t}{b(m)}\norma{\zeta}_{1}\right)^{r}
\end{aligned}
\end{align}
Thus,
\begin{align}
\log \phi_m(\zeta) = -\frac12 \frac{1}{b(m)^2} \sum_{k,k'} 
\E\Big[\big(\cE_k(\Theta,\sigma_{V_1})\cdot \zeta_1 + u_k(\Theta,\sigma_{A_2})\cdot \zeta_2\big)
\big(\cE_{k'}(\Theta,\sigma_{V_1})\cdot \zeta_1 + u_{k'}(\Theta,\sigma_{A_2})\cdot \zeta_2\big)\Big] + R(\zeta,1).
\end{align}
On the other hand, notice that \eqref{Vconvergencehypothesis} implies that 

\begin{align}
 \lim_{m\rightarrow +\infty}\frac{L\vert\mathcal{M}\vert\vert\mathcal{N} \vert}{b(m)}=0.   
\end{align}
By \autoref{VA2} and Lévy’s continuity theorem,
\begin{align}
\lim_{m\to\infty} \phi_m(\zeta) =\exp\Big(-\frac12 \zeta^T \K_V(\sigma_{V_1},\sigma_{A_2}, \sigma_{V_1},\sigma_{A_2}) \zeta \Big),
\end{align}
so that $U(\Theta,w) \xrightarrow{\mathcal D} \U^{(\infty)}(w)$, Gaussian with covariance $\K_V$. The rest of the proof of \autoref{thm:Vconvergence_at_initialization} follows as in the proof of \autoref{lem:convergnce at initialization}.
\end{proof}
An immediate consequence of the previous \autoref{thm:Vconvergence_at_initialization} is the following.

\begin{lem}\label{Lem_Var_at_init}
Let us assume the same conditions of \autoref{thm:Vconvergence_at_initialization}. We have that
$\left\{ \U^{\lin}(\Theta_{t}^{\lin},w)\right\}_{w\in \H_{M}\times \partial B}$ converges in distribution to a Gaussian process $\left\{\U_{t}^{(\infty)}(w)\right\}_{z\in \H_{M}\times \partial B}$ as $m\rightarrow+\infty$ with mean and covariance given by

\begin{align}\label{V_mediat}
&\mu_{t}(w)=\overline{K}_{V}(w,W^{T})\overline{K}_{V}^{-1}\left(\mathbbm{1}-e^{\eta_{0}\overline{K}_{V}t}\right)\hat{Y},\\\label{V_variancet}
&\begin{aligned}
\K_{V,t}(w,w')&=\K_{V,0}(w,w')-\overline{K}_{V}(w,W^{T})\overline{K}_{V}^{-1}\left(\mathbbm{1}-e^{-\eta_{0}\overline{K}_{V}t}\right)\K_{V,0}(W,w')\\
&-\overline{K}_{V}(w',W^{T})\overline{K}_{V}^{-1}\left(\mathbbm{1}-e^{-\eta_{0}\overline{K}_{V}t}\right)\K_{0}(W,w)+\\
&+\overline{K}_{V}(w,W^{T})\overline{K}_{V}^{-1}\left(\mathbbm{1}-e^{-\eta_{0}\overline{K}_{V}t}\right)\K_{V,0}(W,W^{T})\left(\mathbbm{1}-e^{-\eta_{0}\overline{K}_{V}t}\right)\overline{K}_{V}^{-1}\K_{V,0}(W,w'),
\end{aligned}
\end{align}
where
\begin{align}
\hat{Y}=\begin{pmatrix}
\F(v_{1})\\
\vdots\\
\F(v_{M})\\
g(\hat{x}^{(1)})\\
\vdots\\
g(\hat{x}^{(n_{2})})
\end{pmatrix}.
\end{align}
\end{lem}
Lastly, let us state our lazy training result for the variational case.
\begin{thm}[Lazy training for variational QPINN]\label{thm:Vlazytraining}
Let us assume that \autoref{A0}, \autoref{A1}, \autoref{A6}, \autoref{VA1}--\autoref{VA3} hold true. Let us set $n_V\coloneqq M+n_{2}$, and define
\begin{align}\label{eq:R2_V}
    R_V(\delta)&\coloneqq \|Y\|_{2}+\sqrt{\frac{2n_V}{\delta}}
\end{align}
for a fixed constant $0<\delta<1$ such that

\begin{align}\label{hp:lambda_pos_V}
\lambda_{\min}^{K_V}\geq\frac{4}{3}\left(g_V(\delta)+\sqrt{4B_VC_V}\right)   
\end{align}
where

\begin{align*}
&g_V(\delta)\coloneqq \frac{128\sqrt{m}L^{\frac{5}{2}}d^{2}A_0\sqrt{{\rm vol}(B)}|\mathcal{M}|^{2}|\mathcal{N}|n_V}{b_K(m)(b(m))^2}\sqrt{\log\left(\frac{2n_V^2}{\delta}\right)},\\
&B_V\coloneqq n_V\frac{256d^2A_0^2{\rm vol}(B)mL^3|\mathcal M|^3|\mathcal N|}{b_K(m)(b(m))^2},\\
&C_V\coloneqq \sqrt{n_V}\frac{4dA_0\sqrt{{\rm vol}(B)}L|\mathcal M|}{b_K(m)b(m)}R_V(\delta).
\end{align*}
Then, there exists a positive number $\widetilde{\lambda}_{\min}^{V}(\delta)$ satisfying
\begin{align}\label{relation1_V}
    \widetilde{\lambda}_{\min}^{V}(\delta)\geq \frac{1}{4}\lambda_{\min}^{K_V},
\end{align} 
whose explicit expression is provided below, such that, when applying gradient flow with learning rate $\eta_{0}$, the following inequalities hold with probability at least $1-\delta$ over random initialization:
\begin{align}
\label{grad1_V}\mathcal{L}(\Theta_{t})&\leq \frac{R_V^2(\delta)}{2} e^{-2\eta_{0}\widetilde{\lambda}_{\min}^{V}(\delta)t} &\forall\,  t\geq 0,\\
\label{grad2_V}\|\Theta_{t}-\Theta_{0}\|_\infty&\leq \frac{1}{\widetilde{\lambda}_{\min}^{V}(\delta)}\sqrt{n_V}\frac{4dA_0\sqrt{{\rm vol}(B)}L\vert\mathcal{M}\vert}{b_{K}(m)b(m)} R_V(\delta) &\forall\,  t\geq 0,\\
\label{grad3_V}
\sup_{\substack{w\in\H_M\times \partial B \\t\geq 0}}\norma{\U(\Theta_{t},w)-\U^{\lin}(\Theta_{t}^{\lin},w)}_{2}&\leq \frac{(5)2^{9}n_V^{2}(R_V(\delta))^{2}d^{10}A_0^4{\rm vol}(B)^2}{\widetilde{\lambda}_{\min}^{V}(\delta)}\Bigg[1+\frac{4}{\widetilde{\lambda}_{\min}^{V}(\delta)}+\\
&+\frac{256}{\eta_{0}(\widetilde{\lambda}_{\min}^{V}(\delta))^{2}b_{K}(m)}\Bigg]\frac{L^{8}m^{2}\vert\mathcal{M}\vert^{5}\vert\mathcal{N}\vert^{2}}{(b_{K}(m))^{2}\,(b(m))^{5}}(1+\log b(m)).
\end{align}
\end{thm}
\begin{proof}
We follow the strategy of the proof of \autoref{thm:lazytraining}, adapting it to the variational setting. By Chebyshev's inequality for random vectors and \autoref{VA1}, with probability at least $1-\frac{\delta}{2}$,
\begin{align}
\|\U(\Theta_{0},W)\|_2 \le \sqrt{\frac{2}{\delta}}
\left\|\big(\mathrm{diag}(\mathcal{K}_{V,0}(W,W^T))\big)^{1/2}
\right\|_2\leq \sqrt{\frac{2n_V}{\delta}}.
\end{align}
Thus
\begin{align}
R_V(\delta)\coloneqq \|Y\|_2+\sqrt{\frac{2n_V}{\delta}}
\end{align}
satisfies
\begin{align}\label{eq:corollaryR_V}
\|\U(\Theta_{0},W)-Y\|_2 \le R_V(\delta),
\end{align}
with probability at least $1-\frac{\delta}{2}$.

Applying \autoref{thm:VNTK} to $M_{ij}\coloneqq K_V(w_{i},w_{j})-\hat K_{V}(\Theta_{0},w_i,w_j)$, we have
\begin{align}
\mathbb{P}[\|M\|_{{\rm F}}\geq \varepsilon]
&\leq n_V^2 \exp\left(-\frac{(b_{K}(m))^{2}(b(m))^{4}}
{4(64)^{2}d^{4}A_{0}^{4}({\rm vol}(B))^{2}mL^{5}\vert\mathcal{M}\vert^{4}\vert\mathcal{N}\vert^{2}}\frac{\varepsilon^{2}}{n_V^{2}}
\right).
\end{align}
By letting 
\begin{align}
 g_V(\delta)\coloneqq \frac{2(64)\sqrt{m}L^{\frac{5}{2}}d^{2}A_0\sqrt{{\rm vol}(B)}\vert\mathcal{M}\vert^{2}\vert\mathcal{N}\vert n_V}{b_{K}(m)(b(m))^{2}}\sqrt{\log\frac{2n_V^{2}}{\delta}}
\end{align}
one has
\begin{align}
    \mathbb{P}\left[\|M\|_{{\rm F}}\geq g_V(\delta)\right]\leq \frac{\delta}{2}.
\end{align}
When $\|M\|_{{\rm F}}< g_V(\delta)$, we have
\begin{align*}
\hat K_{V}(\Theta_{0},W,W^T)\succ\left(\lambda_{\min}^{K_V}-g_V(\delta)\right)\id 
\end{align*}
with probability at least $1-\frac{\delta}{2}$.

By \autoref{lem:Vlipschitzntk}, for $t>0$,
\begin{align}
\|\hat K_{V}(\Theta_{t},w,w')-\hat K_{V}(\Theta_{0},w,w')\|_{F}&\leq 
\frac{1}{b_{K}(m)}\frac{256d^{2}A_{0}^{2}{\rm vol}(B)mL^{3}\vert\mathcal{M}\vert^{3}\vert\mathcal{N}\vert}{(b(m))^{2}}\norma{\Theta_{t}-\Theta_{0}}_{\infty}.
\end{align}
Define
\begin{align}
t_1=\inf\left\{t:\|\Theta_{t}-\Theta_{0}\|_\infty\geq \rho_V(m)\right\}
\end{align}
for some $\rho_V(m)>0$ to be determined. Then
\begin{align}\label{ntkbound_V}
\|\hat K_{V}(\Theta_{t},W,W^{T})-\hat K_{V}(\Theta_{0},W,W^{T})\|_{{\rm F}}&\leq n_V\frac{256d^{2}A_{0}^{2}{\rm vol}(B)mL^{3}\vert\mathcal{M}\vert^{3}\vert\mathcal{N}\vert}{b_{K}(m)(b(m))^{2}}\rho_V(m)\eqqcolon h_V(\delta),
\end{align}
whence
\begin{align}\label{usataforq_V}
\hat K_{V}(\Theta_{t},W,W^T)\succ\left(\lambda_{\min}^{K_V}-g_V(\delta)-h_V(\delta)\right)\id\eqqcolon\widetilde{\lambda}_{\min}^{V}(\delta)\id \qquad \forall \,t\leq t_1,
\end{align}
with probability at least $1-\delta$. We choose
\begin{align}
\rho_V(m)=\frac{1}{\widetilde{\lambda}_{\min}^{V}(\delta)}\sqrt{n_V}\frac{4dA_0\sqrt{{\rm vol}(B)}L\vert\mathcal{M}\vert}{b_{K}(m)b(m)} R_V(\delta).
\end{align}
This leads to the quadratic equation
\begin{align}
B_V(\rho_V(m))^2 - A_V\rho_V(m) + C_V = 0,
\label{eq:quadratic_rho_V}
\end{align}
with $A_V=\lambda_{\min}^{K_V}-g_V(\delta)$. By hypothesis \eqref{hp:lambda_pos_V}, this equation admits a positive solution.

Following the same argument as in the proof of \autoref{thm:lazytraining}, we obtain
\begin{align}
\widetilde{\lambda}_{\min}^{V}(\delta) = \frac{A_V + \sqrt{A_V^2 - 4B_VC_V}}{2}\geq \frac{1}{4}\lambda_{\min}^{K_V}.
\end{align}

The gradient flow analysis proceeds analogously. For $t\leq t_1$, with probability at least $1-\delta$,
\begin{align}
\mathcal{L}(\Theta_{t})&\leq \frac{R_V^2(\delta)}{2}e^{-2\eta_{0}\widetilde{\lambda}_{\min}^{V}(\delta)t},\\
|\theta_{i}(t)-\theta_{i}(0)|&\leq \frac{1}{\widetilde{\lambda}_{\min}^{V}(\delta)}\sqrt{n_V}\frac{4dA_0\sqrt{{\rm vol}(B)}L\vert\mathcal{M}\vert}{b_{K}(m)b(m)} R_V(\delta)\left(1-e^{-\eta_{0}\widetilde{\lambda}_{\min}^{V}(\delta)t}\right).
\end{align}
Since the bound is strictly less than $\rho_V(m)$ for finite $t$, we must have $t_1=\infty$.

For the discrepancy bound, by \autoref{lem:discrepancies} adapted to the variational case and using \eqref{pt:V1},
\begin{align}
\|\U(\Theta,w)-\U^{\lin}(\Theta,w)\|_{2}&\leq \frac{4dA_0\sqrt{{\rm vol}(B)}L^2m}{b(m)}\vert\mathcal{M}\vert^{2}\vert \mathcal{N}\vert\norma{\Theta_{0}-\Theta}_{\infty}^{2}+\\
&\quad+\frac{Lm}{b(m)}\vert\mathcal{M}\vert^{2}\vert \mathcal{N}\vert\norma{\Theta_{0}-\Theta}_{\infty}^{2}.
\end{align}
Following the same integration argument as in \autoref{thm:lazytraining} and using \eqref{growsbm} adapted to the variational setting, we obtain \eqref{grad3_V}.
\end{proof}

\section{A comparison with previous results}\label{sec:comparison}
In \cite{wang2022and}, the authors investigate the training dynamics of physics-informed neural networks through the lens of the neural tangent kernel. They were restricted to the analysis of classical fully connected architectures and consider a specific elliptic operator for which all coefficients are equal to one. In particular, the differential operator does not involve variable coefficients, and the resulting PDE structure is considerably simplified. Within this setting, the authors analyze the evolution of the neural tangent kernel along training and prove that, under suitable assumptions, the kernel remains close to its initialization. More precisely, in \cite[Theorem 4.4]{wang2022and}, it is shown that the time-dependent kernel $K(t)$ stays close to $K(0)$ provided that two key conditions are satisfied. 
The first condition assumes a uniform bound on the network parameters in the $\ell^\infty$-norm, while the second condition requires a form of closeness of the network output to the forcing term and boundary data throughout training. These assumptions allow the authors to control the variation of the kernel and to justify a lazy training regime for the specific operator under consideration. However, the required hypotheses are imposed a priori and are not derived from the dynamics itself. 
Moreover, the analysis does not account for variable coefficients in the differential operator, nor does it provide quantitative probabilistic estimates on the deviation of the empirical NTK from its expected value as a function of the network width.

Let us mention the approach developed in \cite{cheng2026consistency}. The authors propose a reproducing kernel Hilbert space (RKHS) framework for analyzing physics-informed neural networks. A key assumption in their work is that the true weak solution to \eqref{linearPDE} belongs to the RKHS associated with the NTK as the number of neurons in a fully-connected neural network diverges. Moreover, their analysis requires that the NTK possesses sufficient regularity for the differential operator $\A$ to be applied.
The present work differs from \cite{wang2022and,cheng2026consistency} in several key aspects. We consider a general class of elliptic operators with variable coefficients and analyze the training dynamics without assuming uniform bounds on the parameters or proximity to the target data as a priori conditions. 
Instead, we derive explicit estimates showing that these properties hold with high probability in the overparameterized regime, as a consequence of kernel concentration and controlled parameter evolution. 
This yields a fully dynamical and probabilistic justification of the lazy training regime, together with explicit finite-width error bounds.

\section{Conclusions}\label{sec:conclusions}
The present work provides a probabilistic and dynamical analysis of physics-informed neural networks in the overparameterized regime, with a particular focus on elliptic partial differential equations. 
By exploiting the neural tangent kernel approach, we have shown that the training dynamics of quantum physics-informed neural networks is well approximated by a linearized model, and we have established finite-width upper bounds to the deviation. Several natural directions for future research emerge from this analysis. 
A first important extension concerns the study of training dynamics for evolutionary problems, such as parabolic or hyperbolic partial differential equations. 
In this setting, the interaction between temporal evolution, kernel dynamics, and the accumulation of approximation errors poses significant analytical challenges and calls for a refined treatment of time-dependent operators and data.

Another promising direction is the investigation of more general boundary conditions and heterogeneous data regimes. 
While the present work focuses on a fixed class of boundary conditions, many applications involve mixed, time-dependent, or noisy boundary data, whose impact on the kernel dynamics and training stability remains largely unexplored.
Another direction concerns the strengthening of the convergence results from finite-dimensional distributions to uniform convergence in function spaces such as $C(B\times\partial B)$, which would require tightness estimates via equicontinuity arguments and additional regularity conditions.
Finally, a fundamental open problem concerns the quantitative relationship between the size of the training dataset and the optimization landscape of physics-informed neural networks. 
In particular, it would be of great interest to derive estimates on the number of points required to ensure that the empirical loss is sufficiently close to its continuous counterpart and that the training dynamics effectively reaches a regime where the loss is fully minimized. 
Understanding this trade-off between data availability, network width, and training time is essential for establishing rigorous complexity guarantees for physics-informed neural networks.

\section*{Acknowledgements}
GDP has been supported by the UNA EUROPA SeedFunding project QUANTUMUnaE (CUP J37G25000380006).
GDP is a member of the ``Gruppo Nazionale per la Fisica Matematica (GNFM)'' of the ``Istituto Nazionale di Alta Matematica ``Francesco Severi'' (INdAM)''.
The author AMH is a member of the ``Gruppo Nazionale per l'Analisi Matematica, la Probabilità e le loro Applicazioni (GNAMPA)'' of the ``Istituto Nazionale di Alta Matematica ``Francesco Severi'' (INdAM)''. We thank Dario Trevisan for suggesting to us the study of PINN's.

\section*{Declarations}
\noindent
{\bf Data Availability} Authors can confirm that all relevant data are included in the article. \\
\vskip 0,1cm
\noindent
{\bf Conflict of interest} The authors confirm that there is no Conflict of interest.

\appendix
\section{Limit theorems for stochastic processes}\label{sec:limit_thms}
In this part, we collect some known results about limit theorems for stochastic processes. In particular, we recall some well-known facts regarding Gaussian processes \cite{Billingsley-Convergence,janson1988}.
\begin{defn}
Let $A$ be a set of indexes. A stochastic process $\{X_{\alpha}\}_{\alpha\in A}$ is said to be a Gaussian process if for every subset $B\subset A$ of finite cardinality, there exists a vector $\mu_{B}\in \R^{\vert B\vert}$, and a positive semi-definite covariance matrix $\K_{B}\in \R^{\vert B\vert\times \vert B\vert}$ such that $\{X_{\alpha}\}_{\alpha\in B}$ is distributed according to a multivariate Gaussian random variable with mean $\mu_{B}$, and covariance $\K_{B}$. In this case, the probability density function is given by
\begin{align}
    p(x)=\frac{1}{(2\pi)^{\frac{\vert B\vert}{2}}}\frac{1}{\sqrt{{\rm det}(\K_{B})}}e^{-\frac{1}{2}(x-\mu_{B})^{T}\K_{B}^{-1}(x-\mu_{B})},\hskip 0,1cm x\in \R^{\vert B \vert}.
\end{align}
\end{defn}
Let us now recall what the convergence in distribution is

\begin{defn}
Let $\{X_{n}\}_{n\in\N}$ be a sequence of real-valued random variables. The sequence converges in distribution to a real-valued random variable $X$ if

\begin{align}
\lim_{n\rightarrow +\infty}\E[f(X_{n})]=\E[f(X)]    
\end{align}
for all bounded and continuous function $f:\R\rightarrow\R$. It is a customary fact to use the notation
\begin{align}
X_{n}\stackrel{\D}{\rightarrow}X.    
\end{align}
\end{defn}
In what follows, we recall the so-called Kolmogorov consistency theorem which helps us to define Gaussian process indexed by infinite-dimensional sets. For $F=(t_1,\dots,t_k)\in\mathcal{F}$ we denote by $\mu_{F}$ a probability measure on $\R^k$.

\begin{thm}[Kolmogorov consistency theorem]\label{thm:KolmogorovConsistency}
Let $\{\mu_F\}_{F\in\mathcal F}$ be a family of probability measures satisfying the following conditions.

\begin{enumerate}
\item given $F=(t_1,\dots,t_k)$, and $G=(t_1,\dots,t_k,t_{k+1},\dots,t_{k+\ell})$, one has for every Borel set
$A\subset\R^k$ that

\begin{align}
\mu_G(A\times\mathbb R^\ell)=\mu_F(A).    
\end{align}
\item Let $\pi$ be a permutation of $\{1,\dots,k\}$ and define
\begin{align}
F_\pi\coloneqq (t_{\pi(1)},\dots,t_{\pi(k)}).    
\end{align}
Then for every Borel set $A\subset\R^k$

\begin{align}
\mu_{F_\pi}(A)=\mu_F\big(\pi^{-1}A\big),    
\end{align}
where $\pi^{-1}A$ denotes the inverse permutation of the coordinates.
\end{enumerate}
Then there exists a stochastic process $X=(X_t)_{t\in T}$ such that for every ordered finite set

\begin{align}
F=(t_1,\dots,t_k)    
\end{align}
the random vector $(X_{t_1},\dots,X_{t_k})$ has distribution $\mu_F$. Moreover the process is unique in distribution.
\end{thm}
The Kolmogorov theorem allows us to construct Gaussian processes from a positive semidefinite covariance kernel.

\begin{thm}[Gaussian process generated by a kernel]\label{thm:GaussianKernel}
Let $T$ be an arbitrary index set and let $K:T\times T\to\R$ be symmetric and positive semidefinite, namely
for every finite collection $t_1,\dots,t_k\in T$, $(K(t_i,t_j))_{i,j=1}^k$ is positive semidefinite. Then there exists a centred Gaussian process $(X_t)_{t\in T}$ such that

\begin{align}
\E[X_t]=0,\qquad \E[X_sX_t]=K(s,t).    
\end{align}
\end{thm}

\begin{proof}
For each ordered finite set $F=(t_1,\dots,t_k)$ define $\mu_F$ to be the Gaussian probability measure
on $\mathbb R^k$ with mean zero and covariance matrix
\begin{align}
\Sigma_F=(K(t_i,t_j))_{i,j=1}^k .    
\end{align}
We verify the assumptions of Theorem \ref{thm:KolmogorovConsistency}. Let us start with the consistency. If
$G=(t_1,\dots,t_k,t_{k+1},\dots,t_{k+\ell})$ then the covariance matrix of $F$ is the principal
submatrix of the covariance matrix of $G$. Therefore the Gaussian measure $\mu_F$ is exactly the marginal of $\mu_G$. Let us verify the invariance. Let $\pi$ be a permutation of $\{1,\dots,k\}$. Since Gaussian measures transform under permutations of coordinates by permuting the covariance matrix, we obtain $\mu_{F_\pi}(A)=\mu_F(\pi^{-1}A)$. Thus both assumptions of Theorem \ref{thm:KolmogorovConsistency} are satisfied. Applying that theorem yields a stochastic process whose finite-dimensional distributions are exactly $\{\mu_F\}$. Since these distributions are Gaussian, the resulting process is a Gaussian process with covariance kernel $K$.
\end{proof}
A fundamental tool to prove the convergence in distribution of a sequence of real-valued random variables $\{X_{n}\}_{n\in\N}$ is the so-called Lévy's continuity theorem \cite{williams1991probability}.

\begin{thm}[Lévy's continuity theorem]\label{thm:Levi}
    Let $\{X_{n}\}_{n\in\N}$ be a sequence of real-valued random variables, and consider the characteristic functions 

    \begin{align}
     \phi_{n}(t)\coloneqq \E[e^{itX_{n}}], \hskip0,1cm n\in\N.  
    \end{align}
    If there exists a real-valued random variable $X$ whose characteristic function $\phi(t)$ is the point-wise limit

    \begin{align}
    \phi(t)=\lim_{n\rightarrow +\infty}\phi_{n}(t)    
    \end{align}
    for all $t\in\R$, then $X_{n}$ converges in distribution to $X$.
\end{thm}
The characteristic function of a random variable is linked to its cumulants. In what follows, we outline this relationship and incorporate a review of the dependency graph concept.
\begin{defn}
    Let $X$ be a real-valued random variable. We denote by $\kappa_{j}\equiv \kappa_{j}(X)$ its cumulant of order $j\in\N$ which is defined as

    \begin{align}
       \kappa_{j}(X\coloneqq (-i)^{j}\frac{\de^{j}}{\de t^{j}}\Bigg|_{t=0}\log[e^{itX}],
    \end{align}
    provided that the derivative exists.
\end{defn}
We now recall the theory of dependency graphs as presented in \cite{feray2016}.

\begin{defn}
Given a family of random variables $\{X_{\alpha}\}_{\alpha\in B}$, a graph $\mathscr{G}$ with vertex set $V$ is called dependency graph for the family if the following holds true: whenever $V_{1}$, and $V_{2}$ are disjoint sets of $V$ such that there are no edges in $\mathscr{G}$ with one end in $V_{1}$ and one in $V_{2}$, the subfamilies of random variables $\{X_{\alpha}\}_{\alpha\in V_{1}}$, and $\{X_{\alpha}\}_{\alpha\in V_{2}}$ are independent. Note that each subfamily can internally have dependencies between its variables.
\end{defn}
Let $\mathscr{G}=(V,E)$ be a graph with vertex set $V$ and edge set $E$. Let us define the maximal degree $D$ of $\mathscr{G}$ as

\begin{align}
   D\coloneqq \max_{k\in V}\left\vert\left\{k'\in V: (k,k')\in E\right\} \right\vert. 
\end{align}
Let us now state the following result that estimates the cumulants of a random variable \cite{janson1988}.
\begin{thm}[{\cite{feray2016}}]\label{thm_dependencygraph}
Let $\{X_{\alpha}\}_{\alpha\in V}$ be a family of random variables with dependency graph $\mathscr{G}=(V,E)$, and denote by $N\coloneqq \vert V\vert$ the cardinality of $V$, and by $D$ the maximal degree of $\mathscr{G}$.  Assume that the random variables are uniformly bounded by a positive constant $A$. Then if 

\begin{align}
 X=\sum_{\alpha\in V}X_{\alpha},  
\end{align}
one has

\begin{align}
 \vert \kappa_{j}(X)\vert\leq C_{j}N(D+1)^{j-1}A^{j}, \hskip 0,1cm j\geq 1, 
\end{align}
where
\begin{align}
 C_{j}\coloneqq 2^{j-1}j^{j-2}.   
\end{align}
\end{thm}
A further result that we neeed in the present work is the following \cite{slutsky1925stochastische}.

\begin{thm}[Slutsky’s theorem]\label{thm:Slutsky}
 Let us consider a sequence of random vectors or matrices $\{X_{n}\}_{n\in\N}$ converging in distribution to a vector or matrix $X$, and let $\{Y_{n}\}_{n\in\N}$be a sequence of random vectors or matrices converging in probability to a vector or matrix $Y$. Then
 \begin{align}
  &X_{n}+Y_{n}\stackrel{d}{\rightarrow} X+Y\\
  &X_{n}Y_{n}\stackrel{d}{\rightarrow} XY.
 \end{align}
 Furthermore, if $Y$ is invertible, then
 \begin{align}
     \frac{X_{n}}{Y_{n}}\stackrel{d}{\rightarrow} \frac{X}{Y}.
 \end{align}
\end{thm}
Lastly, let us recall the following concentration inequality that we used to provide an explicit rate of convergence to the analytic NTK \cite{mcdiarmid1989method}.

\begin{thm}[McDiarmid’s concentration inequality]\label{thm:McDiarmid}
 Let $X_{1},\ldots, X_{n}$ be independent random variables with values in a Polish space $\mathbb{X}$. Let $f:\mathbb{X}^{n}\rightarrow \R$ be a function such that for every $i\in\{1,\ldots,n\}$, and every $(x_{1},\ldots,x_{n})$, $(x_{1}',\ldots,x_{n'})$ in $\mathbb{X}^{n}$ that differ only in the $i$-th coordinate (that is, $x_{j}=x_{j}'$ for all $j\neq i$)
 \begin{align}
     \left\vert f(x_{1},\ldots,x_{n})-f(x_{1}',\ldots,x_{n}')\right\vert \leq c_{i}.
 \end{align}
 Then for any $\varepsilon>0$, 

 \begin{align}
     \P\left[f(x_{1},\ldots,x_{n})-\E[f(x_{1},\ldots,x_{n})]\geq \varepsilon\right]\leq \exp\left(\frac{-2\varepsilon^{2}}{\sum_{i=1}^{n}c_{i}^{2}}\right).
 \end{align}
\end{thm}

\section{Poincaré and reverse Poincaré inequalities}
In what follows, we present the proof of a Poincaré inequality and a reverse Poincaré inequality, which may be of interest due to their intrinsic importance in the literature. Furthermore, it provides an estimate on the growth of the constant $b_{K}(m)$.

\begin{lem}[Poincaré and reverse Poincaré inequality]
 Suppose that \autoref{A4} holds true.
 \begin{align}\label{desiredpoincare}
4\E[\norma{\U(\Theta,z)}_{2}^{2}]\leq \E[\norma{\nabla_{\Theta}\U(\Theta,z)}_{F}^{2}]\leq 4\vert \mathcal{N}\vert \E[\norma{\U(\Theta,z)}_{2}^{2}].
 \end{align}
\end{lem}
\begin{proof}
In what follows, we consider closely the argument provided in the proof of \cite[Lemma 4.30]{girardi2025}
to obtain \eqref{boundu} because we need to translate it to a bound for $\A u(\Theta,x)$. Notice that

\begin{align}
\E[\norma{\U(\Theta,z)}_{2}^{2}]=\E[\vert u(\Theta,\hat{x})\vert^{2}]+\E[\vert \A u(\Theta,x)\vert^{2}].  
\end{align}
By \cite[Lemma 4.30]{girardi2025} we have that
 \begin{align}\label{boundu}
4\E[\vert u(\Theta,\hat{x})\vert^{2}]\leq \E[\norma{\nabla_{\Theta}u(\Theta,\hat{x})}_{2}^{2}]\leq \vert \mathcal{N}\vert\E[\vert u(\Theta,\hat{x})\vert^{2}].
 \end{align}
Since $u(\Theta,x)$ is periodic on $\Theta$ with period $\pi$ in each component, we can use its Fourier transform and thus
 \begin{align}
 u(\Theta,x)=\sum_{v\in \Z^{Lm}}\tilde{u}_{v}(x)e^{2i\Theta\cdot v}, \hskip 0,2cm \tilde{u}_{v}(x)\coloneqq \prod_{i=1}^{Lm}\displaystyle\int_{0}^{\pi} e^{-2i\Theta\cdot v} u(\Theta,x)\left(\frac{\de \theta_{i}}{\pi}\right). 
 \end{align}
 Let us recall that,
 \begin{align}
 \begin{aligned}
u(\Theta,x)&=\bra{0^m}U^{\dagger}(\Theta,x)\O U(\Theta,x)\ket{0^{m}}\\
&=\bra{0^{m}}W_{1}^{\dagger}(\Theta)V_{1}^{\dagger}(x)\cdots W_{L}^{\dagger}(\Theta)V_{L}^{\dagger}(x)\O V_{L}(x)W_{L}(\Theta)\cdots V_{1}(x)W_{1}(\Theta)\ket{0^{m}}
 \end{aligned}
 \end{align}
 where
 \begin{align}
   W_{i}(\Theta)=\bigotimes_{\ell=1}^{m}e^{-i\theta_{\ell}\mathcal{G}_{i}}.
 \end{align}
 Since the spectrum of $\mathcal{G}_{i}$ is $\{-1,1\}$, we can write $\mathcal{G}_{i}$ in terms of its projectors $P_{i}^{+},P_{i}^{-}$ as
 \begin{align*}
 \mathcal{G}_{i}=P_{i}^{+}-P_{i}^{-},\hskip 0,1cm P_{i}^{+}+P_{i}^{-}=\mathbbm{1},
 \end{align*}
 and thus
 \begin{align}
   W_{i}(\Theta)=\bigotimes_{\ell=1}^{m}\left[e^{-i\theta_{\ell}}P_{i}^{+}+e^{i\theta_{\ell}}P_{i}^{-}\right].
 \end{align}
 Hence, we can conclude that 

\begin{align}\label{expresssionofu}
u(\Theta,x)= \sum_{v\in \{-1,0,1\}^{Lm}}\tilde{u}_{v}(x)e^{2i\Theta\cdot v}.   
\end{align}
In what follows, we can restrict ourselves to the case $\tilde{u}_{v}(x)\neq 0$. Let us further set

\begin{align*}
{\rm supp}(v)\coloneqq \{i\in \{1,\ldots, Lm\}:v_{i}\neq 0\}    
\end{align*}
Notice that

\begin{align}
\begin{aligned}
\tilde{u}_{v}(x)=\prod_{i=1}^{Lm}\displaystyle\int_{0}^{\pi}e^{-2i\Theta\cdot v}\sum_{k=1}^{m}u_{k}(\Theta_{\mathcal{N}_{k}},x)\left(\frac{\de \theta_{i}}{\pi}\right)  
\end{aligned}    
\end{align}
where we have written $u(\Theta,x)$ depending on limited light cone. From here one has

\begin{align}\label{coeffdiu}
\begin{aligned}
\tilde{u}_{v}(x)&=\sum_{k=1}^{m}\prod_{i\notin \mathcal{N}_{k}}\displaystyle\int_{0}^{\pi}e^{-2i\theta_{i}v_{i}}\left(\frac{\de \theta_{i}}{\pi}\right)\prod_{i\in \mathcal{N}_{k}}\displaystyle\int_{0}^{\pi}e^{-2i\theta_{i}v_{i}}u_{k}(\Theta_{\mathcal{N}_{k}},x)\left(\frac{\de \theta_{i}}{\pi}\right)\\
&=\sum_{k=1}^{m}\chi_{\mathcal{N}_{k}}({\rm supp}(v))\prod_{i\in \mathcal{N}_{k}}\displaystyle\int_{0}^{\pi}e^{-2i\theta_{i}v_{i}}u_{k}(\Theta_{\mathcal{N}_{k}},x)\left(\frac{\de \theta_{i}}{\pi}\right),
\end{aligned}
\end{align}
where

\begin{align}
 \begin{aligned}
 \chi_{\mathcal{N}_{k}}({\rm supp}(v))\coloneqq
 \begin{cases}
     1&\text{if ${\rm supp}(v)\subseteq \mathcal{N}_{k}$,}\\
     0&\text{otherwise.}
 \end{cases}
 \end{aligned}   
\end{align}
Notice that, we can express $\A u(\Theta,x)$ as

\begin{align}
 \A u(\Theta,x)=\sum_{v\in \{-1,0,1\}^{Lm}}\A\tilde{u}_{v}(x)e^{2i\Theta\cdot v}. 
 \end{align}
 where we have used \eqref{expresssionofu}. On the other hand, by \eqref{coeffdiu}, we have that

 \begin{align}
 \begin{aligned}
\A\tilde{u}_{v}(x)&=\sum_{k=1}^{m}\chi_{\mathcal{N}_{k}}({\rm supp}(v))\A\prod_{i\in \mathcal{N}_{k}}\displaystyle\int_{0}^{\pi}e^{-2i\theta_{i}v_{i}}u_{k}(\Theta_{\mathcal{N}_{k}},x)\left(\frac{\de \theta_{i}}{\pi}\right)\\
&=\sum_{k=1}^{m}\chi_{\mathcal{N}_{k}}({\rm supp}(v))\prod_{i\in \mathcal{N}_{k}}\displaystyle\int_{0}^{\pi}e^{-2i\theta_{i}v_{i}}\A u_{k}(\Theta_{\mathcal{N}_{k}},x)\left(\frac{\de \theta_{i}}{\pi}\right).
 \end{aligned}
 \end{align}
 Hence, 
 \begin{align*}
\A\tilde{u}_{v}(x)\neq 0\quad\rightarrow   \quad\text{there exists $k\in \{1,\ldots,m\}$ such that ${\rm supp}(v)\subseteq \mathcal{N}_{k}$.}
 \end{align*}
 Furthermore, 

 \begin{align}
 \partial_{\theta_{j}}\A u(\Theta,x)=2i\sum_{v\in\{-1,0,1\}^{Lm}}v_{j}\A \tilde{u}_{v}(x)e^{2i\Theta\cdot v}
 \end{align}
 By Parseval's identity one has

 \begin{align}
 \begin{aligned}
 &\E[\vert\A u(\Theta,x)\vert^{2}]=\sum_{v\in\{-1,0,1\}^{Lm}}\vert \A \tilde{u}_{v}(x)\vert^{2}\\
 &\E[(\partial_{\theta_{j}}\A u(\Theta,x))^{2}]=4\sum_{v\in\{-1,0,1\}^{Lm}}v_{j}^{2}\vert \A \tilde{u}_{v}(x)\vert^{2}.
 \end{aligned}
 \end{align}
Now, let us observe that 

\begin{align}
\begin{aligned}
\E[\norma{\nabla_{\Theta}\A u(\Theta,x)}_{2}^{2}]&=4\sum_{j=1}^{Lm}\sum_{v\in\{-1,0,1\}^{Lm}}v_{j}^{2}\vert \A \tilde{u}_{v}(x)\vert^{2}\\
&=4\sum_{j=1}^{Lm}\sum_{\stackrel{v\in\{-1,0,1\}^{Lm}}{v_{j}\neq 0}}v_{j}^{2}\vert \A \tilde{u}_{v}(x)\vert^{2}\\
&=4\sum_{v\in\{-1,0,1\}^{Lm}}\sum_{v_{j}\neq 0}\vert \A \tilde{u}_{v}(x)\vert^{2}\\
&=4\sum_{v\in\{-1,0,1\}^{Lm}}{\rm supp}(v)\vert \A \tilde{u}_{v}(x)\vert^{2}\\
&\leq 4\vert\mathcal{N}\vert\sum_{v\in\{-1,0,1\}^{Lm}}\vert \A \tilde{u}_{v}(x)\vert^{2}\\
&=4\vert\mathcal{N}\vert\E[\vert \A u(\Theta,x)\vert^{2}].
\end{aligned}
\end{align}
Therefore,
\begin{align}
\begin{aligned}\label{disuforA}
\E[\norma{\nabla_{\Theta}\U(\Theta,z)}_{F}^{2}]&=\E[\norma{\nabla_{\Theta}u(\Theta,\hat{x})}_{2}^{2}]+\E[\norma{\nabla_{\Theta}\A u(\Theta,x)}_{2}^{2}]\\
&\leq 4\vert \mathcal{N}\vert\E[\norma{\U(\Theta,z)}_{2}^{2}].
\end{aligned}
\end{align}
Lastly, notice that

\begin{align}
\begin{aligned}
\E[\norma{\nabla_{\Theta}\A u(\Theta,x)}_{2}^{2}]&=\sum_{j=1}^{Lm}\E[(\partial_{\theta_{j}}\A u(\theta,x))^{2}]\\
&=4\sum_{v\in \Z^{Lm}}\sum_{j=1}^{Lm}v_{j}^{2}\vert \A \tilde{u}_{v}(x)\vert^{2}\\
&\geq 4\sum_{v\in \Z^{Lm}\backslash \{0\}^{Lm}}\sum_{j=1}^{Lm}v_{j}^{2}\vert \A \tilde{u}_{v}(x)\vert^{2}\\
&\geq 4\sum_{v\in \Z^{Lm}\backslash \{0\}^{Lm}}\vert \A \tilde{u}_{v}(x)\vert^{2}\\
&=4\sum_{v\in \Z^{Lm}}\vert \A \tilde{u}_{v}(x)\vert^{2}
\end{aligned}    
\end{align}
where in the last inequality, we have used that $\E[u(\Theta,x)]=0$ implies that $\tilde{u}_{v}=0$ with $v=\{0\}^{Lm}$. Thus,
\begin{align}\label{disu2}
4\E[\vert \A u(\Theta,x) \vert^2]\leq \E[\norma{\nabla_{\Theta}\A u(\Theta,x)}_{2}^{2}].
\end{align}
By combining \eqref{disuforA}, \eqref{disu2}, and \eqref{boundu}, we reach \eqref{desiredpoincare}.
\end{proof}
\begin{cor}\label{cor:ordinegrandezza}
It holds true that
\begin{align}
  \Omega(1)\leq b_{K}(m)\leq O(\vert \mathcal{N}\vert).  
\end{align}
\end{cor}
\section{Convergence of the linearized model}
In this part, we show that the QPINN $\U^{\lin}(\Theta_{t}^{\lin},z)$ converges in distribution to a Gaussian process. Let us assume the following conditions.
\begin{ass}\label{A9}
Assume that there exists a deterministic limit matrix kernel of size $2\times 2$ denoted $\overline{K}(z,z')$ such that
\[
\lim_{m\to\infty}\sup_{z,z'} \norma {K(z,z') - \overline{K}(z,z')}_{F} = 0,
\]
with $\overline{K}$ not identically zero, and $\norma{\cdot}_{F}$ denotes the Frobenius norm.
\end{ass}
\begin{ass}\label{A10}
We assume that 
\begin{align*}
    \lim_{m\rightarrow+\infty}\sup_{z,z'\in B\times\partial B}\norma{ \E[\U(\Theta,z)(\U(\Theta,z'))^{T}]-\K(z,z')}_{F}=0,
\end{align*}
where $\K: (B \times \partial B) \times (B \times \partial B) \to \mathbb{M}_{2\times 2}$ is a positive semi-definite matrix-valued function with strictly positive diagonal entries for all $z$. 
\end{ass}
Here, we notice that $b(m)$ is chosen such that the limit yields a nontrivial matrix-valued operator kernel $\K$ with positive entries $\K_{11}(z,z),\K_{22}(z,z)$. 
\begin{thm}\label{lem:convergnce at initialization}
 Let us assume that \autoref{A0}--\autoref{A6}, \autoref{A9}--\autoref{A10} hold true. Furthermore, suppose that 
 \begin{align}\label{convergencehypothesis}
  \lim_{m\rightarrow +\infty}\frac{m L^{6}\vert \mathcal{M}\vert^{2}\vert \mathcal{N}\vert^{2}}{(b(m))^{3}}=0.   
 \end{align}
 Then
 \begin{align}\label{limitp}
 \U^{\lin}(\Theta_{t}^{\lin},z)\xrightarrow[m\to\infty]{\mathcal D} \U^{(\infty)}(z)-\overline{K}(z,Z^{T})\overline{K}^{-1}(1-e^{-\eta_{0}\overline{K}t})\left(\U^{(\infty)}(Z)-Y\right),  
 \end{align}
 where $\U^{(\infty)}(z)$ is a Gaussian vector with covariance matrix $\K$.
\end{thm}
\begin{proof}
In what follows, we analyze the characteristic function associated to $\U(\Theta,z)$. To this aim, we need to take of the contribution of $\A u$ since the contribution of $u$ is already analyzed in \cite{girardi2025}. Notice that we need to take care of the contribution of $\A u$ when studying the convergence of $\U^{\lin}(\Theta_{t}^{\lin},z)$.

{\textbf Step 1}\, Let us first prove that $\U(\Theta,z)\rightarrow \U^{(\infty)}(z)$ in the distribution as $m\rightarrow +\infty$ with the covariance matrix $\K$. Let us consider two collection of points $\sigma_{A_{1}}\coloneqq (x_{\alpha}:\alpha\in A_{1})$, and $\sigma_{A_{2}}\coloneqq (\hat{x}_{\beta}:\beta\in A_{2})$. Further, let us set for $\zeta=(\zeta_{1},\zeta_{2})\in \R^{\vert A_{1}\vert}\times \R^{\vert A_{2}\vert}\equiv \R^{\vert A_{1}\vert+\vert A_{2}\vert}$, $\norma{\zeta}_{1}\coloneqq \sum_{\alpha}\zeta_{\alpha}$. In the next, we analyze the characteristic function associated to $\U(\Theta,\sigma_{A_{1}},\sigma_{A_{2}})$:

\begin{align}
\phi_{m}(\zeta)=\E\left[\exp\left(i\frac{1}{b(m)}\sum_{k=1}^{m}\left[\A u_{k}(\Theta,\sigma_{A_{1}})\cdot \zeta_{1}+u_{k}(\Theta,\sigma_{A_{2}})\cdot \zeta_{2}\right]\right)\right]
\end{align}
where
\begin{align}
\A u_{k}(\Theta,\sigma_{A_{1}})\cdot \zeta_{1}\coloneq \sum_{\alpha\in A_{1}}u_{k}(\Theta,x_{\alpha})\zeta_{1,\alpha},\hskip 0,1cm u_{k}(\Theta,\sigma_{A_{2}})\cdot \zeta_{2}\coloneq \sum_{\beta\in A_{2}}u_{k}(\Theta,\hat{x}_{\beta})\zeta_{2,\beta}  
\end{align}
Let $t\in [0,+\infty)$ and consider
\begin{align}
\phi_{m}(\zeta,t)=\E\left[\exp\left(it\frac{1}{b(m)}\sum_{k=1}^{m}\left[\A u_{k}(\Theta,\sigma_{A_{1}})\cdot \zeta_{1}+u_{k}(\Theta,\sigma_{A_{2}})\cdot \zeta_{2}\right]\right)\right]
\end{align}
It is well-known that we can expand the characteristic function in terms of its cumulants as follows \cite{janson2024}. We have that

\begin{align}
\log\phi_{m}(\zeta,t)=\sum_{r=1}^{+\infty}\frac{\kappa_{r}^{(m)}}{r!}(it)^{r},\hskip 0,2cm \kappa_{s}^{(m)}\equiv (-i)^{r}\frac{\de^{r}}{\de t^{r}}\log\phi_{m}(\zeta,t),  
\end{align}
where $\phi_{m}(\zeta)=\phi_{m}(\zeta,1)$. Notice that

\begin{align}
\left\vert\frac{1}{b(m)}\A u_{k}(\Theta,\sigma_{A_{1}})\cdot \zeta_{1}\right\vert\leq \frac{2d^{2}(2A_{1}+4A_{0})L^2}{b(m)}\norma{\zeta_{1}}_{1},\hskip 0,3cm \left\vert\frac{1}{b(m)}u_{k}(\Theta,\sigma_{A_{2}})\cdot \zeta_{2}\right\vert\leq \frac{1}{b(m)}\norma{\zeta_{2}}_{1},
\end{align}
so that 

\begin{align}
\left\vert \frac{1}{b(m)}\U_{k}(\Theta,\sigma_{A_{1}},\sigma_{A_{2}})\cdot \zeta\right\vert\leq  \frac{4d^{2}(2A_{1}+4A_{0})L^2}{b(m)}\norma{\zeta}_{1} 
\end{align}
where we have used \eqref{pt:0}. Therefore, by \autoref{thm_dependencygraph}, one has

\begin{align}
\begin{aligned}
\vert \kappa_{r}^{(m)}\vert&\leq  2^{r-1}r^{r-2}m(D+1)^{r-1}\left(\frac{2d^{2}(2A_{1}+4A_{0})L^2 D}{b(m)}\norma{\zeta}_{1}\right)^{r}\\
&\leq \frac{m}{D}\left(\frac{8d^{2}(2A_{1}+4A_{0})L^2 D}{b(m)}\norma{\zeta}_{1}\right)^{r}r^{r}.
\end{aligned}
\end{align}
Since $r^{r}\leq e^{r}r!$, then 

\begin{align}
\vert \kappa_{r}^{(m)}\vert\leq \frac{m}{D}\left(\frac{8d^{2}e(2A_{1}+4A_{0})L^2 D}{b(m)}\norma{\zeta}_{1}\right)^{r}r!.
\end{align}
As pointed out in \cite{girardi2025}, this bound allow us to control the Taylor expansion of $\log\phi_{m}(\zeta,t)$. Let us take

\begin{align}
R(\zeta,t)\coloneqq \sum_{r=3}^{\infty}\frac{\kappa_{r}^{(m)}}{r!}(it)^{r}.  
\end{align}
Notice that 

\begin{align}
\begin{aligned}
\vert R(\zeta,t)\vert &\leq \sum_{r=3}^{\infty}\frac{m}{D}\left(\frac{8d^{2}e(2A_{1}+4A_{0})L^{2} Dt}{b(m)}\norma{\zeta}_{1}\right)^{r}\\
&=\frac{m}{D}\left(\frac{8d^{2}e(2A_{1}+4A_{0})L^{2} Dt}{b(m)}\norma{\zeta}_{1}\right)^{3}\sum_{r=0}^{\infty}\left(\frac{8d^{2}e(2A_{1}+4A_{0})L^{2}Dt}{b(m)}\norma{\zeta}_{1}\right)^{r}\\
&\leq \frac{m L^{6}\vert \mathcal{M}\vert^{2}\vert \mathcal{N}\vert^{2}}{(b(m))^{3}}\left(8d^{2}e(2A_{1}+4A_{0}) t\norma{\zeta}_{1}\right)^{3}\sum_{r=0}^{\infty}\left(\frac{8d^{2}e(2A_{1}+4A_{0})L^{2}\vert \mathcal{M}\vert\vert \mathcal{N}\vert t}{b(m)}\norma{\zeta}_{1}\right)^{r}
\end{aligned}
\end{align}
 Notice that from \eqref{convergencehypothesis}, we have that 
 \begin{align}
     \lim_{m\rightarrow +\infty}\frac{L^{2}\vert M\vert\vert \mathcal{N}\vert}{b(m)}=0.
 \end{align}
For fixed $t,\zeta$, we have that for some $m_{0}\geq 1$, 

\begin{align}
\frac{8d^{2}e(2A_{1}+4A_{0})L^{2}\vert \mathcal{M}\vert\vert \mathcal{N}\vert t}{b(m)}\norma{\zeta}_{1}\leq \frac{1}{2}, \hskip 0,12cm \text{for all $m\geq m_{0}$,} 
\end{align}
and since $\sum_{r=0}^{\infty}2^{-r}=2$, we get for all $m\geq m_{0}$ that

\begin{align}
   \vert R(\zeta,t) \vert\leq  \frac{m L^{6}\vert \mathcal{M}\vert^{2}\vert \mathcal{N}\vert^{2}}{(b(m))^{3}}\left(8d^{2}e(2A_{1}+4A_{0}) t\norma{\zeta}_{1}\right)^{3}.
\end{align}
Then $\lim_{m\rightarrow+\infty}\vert R(\zeta,t)\vert=0.$ On the other hand, we have that,

\begin{align}
\begin{aligned}
\log\phi_{m}(\zeta,1)&= \kappa_{1}^{(m)}-\frac{1}{2}\kappa_{2}^{(m)}+ R(\zeta,1)\\ 
&=-\frac{1}{2}\frac{1}{(b(m))^{2}}\sum_{k,k'=1}\E\left[\left(\A u_{k}(\Theta,x_{A_{1}})\cdot \zeta_{1}+ u_{k}(\Theta,x_{A_{2}})\cdot \zeta_{2}\right)\right.\\
&\phantom{formulaformula}\left.\left(\A u_{k'}(\Theta,x_{A_{1}})\cdot \zeta_{1}+ u_{k'}(\Theta,x_{A_{2}})\cdot \zeta_{2}\right)\right]+R(\zeta,1)\\
&=-\frac{1}{2}\frac{1}{(b(m))^{2}}\sum_{k,k'=1}^{m}\E\left[\zeta^{T}\U(\Theta,\sigma_{A_{1}},\sigma_{A_{2}})(\U(\Theta,\sigma_{A_{1}},\sigma_{A_{2}}))^{T}\zeta\right]+R(\zeta,1).
\end{aligned}  
\end{align}

By the continuity of the logarithm, and by Lévy's continuity theorem stated in \autoref{thm:Levi}, we have by \autoref{A4} that
\begin{align}
\lim_{m\rightarrow +\infty}\phi_{m}(\zeta)=\exp\left(-\frac{1}{2}\zeta^{T}\K(\sigma_{A_{1}},\sigma_{A_{2}}, \sigma_{A_{1}},\sigma_{A_{2}})\zeta\right),
\end{align}
where the limit is the characteristic function of a Gaussian variable. Therefore $\U(\Theta,z)$ converges in distribution to a Gaussian vector with covariance matrix $\K$.\\
{\textbf Step 2}\, We now combine the previous step with Slutsky’s theorem. Indeed, let us notice that the expression of $\U^{\lin}$ in \eqref{Ulin_solution} is a linear combination of the outputs $\{\U(\Theta_{0},z_{\alpha})\}_{\alpha\in \mathcal{B}}$ where $\mathcal{B}$ is an index set such that the corresponding inputs $\{z_{\alpha}\}_{\alpha\in \mathcal{B}}$ belong to $B\times \partial B$. In particular, we notice that 

\begin{align}
\U^{\lin}(\Theta^{\lin},z_{\beta})=\sum_{\alpha\in \mathcal{B}}M_{\beta,\alpha,t}[\hat{K}_{\Theta_{0}}]\U(\Theta_{0},z_{\alpha})+\left(R(\hat{K}_{\Theta_{0}})\right)^{T}Y  
\end{align}
where the entries of $M_{\beta,\alpha,t}[\hat{K}_{\Theta_{0}}]$, and $R_{t}(\hat{K}_{\Theta_{0}})$ are continuous functions of the empirical NTK $\left\{ \hat{K}_{\Theta_{0}}(z_{\alpha},\alpha')\right\}_{\alpha,\alpha'\in\mathcal{B}}$. Therefore, by continuity the finite matrix $M_{\beta,\alpha,t}[\hat{K}_{\Theta_{0}}]$ converges in probability to $M_{\beta,\alpha,t}[\overline{K}]$, and $R_{t}(\hat{K}_{\Theta_{0}})$ to $R_{t}(\overline{K})$, and thus by \autoref{thm:Slutsky}, we conclude that 

\begin{align}
\left\{\U^{\lin}(\Theta_{t}^{\lin}, z_{\alpha})\right\}_{\alpha\in \mathcal{B}}\rightarrow \left\{\sum_{\alpha\in \mathcal{B}}M_{\beta,\alpha,t}[\hat{K}_{\Theta_{0}}]\U(\Theta_{0},z_{\alpha})+\left(R(\hat{K}_{\Theta_{0}})\right)^{T}Y \right\}_{z_{\beta}\in \mathcal{B}}
\end{align}
in distribution as $m\rightarrow +\infty$, and we are done.
\end{proof}
An immediate consequence of the previous result is the following:
\begin{lem}\label{lem_gauss_lim}
Let us assume \autoref{A0}--\autoref{A6}, and \autoref{A9}--\autoref{A10}. We have that
$\left\{ \U^{\lin}(\Theta_{t}^{\lin},z)\right\}_{z\in B\times \partial B}$ converges in distribution to a Gaussian process $\left\{\U_{t}^{(\infty)}(z)\right\}_{z\in B\times \partial B}$ as $m\rightarrow+\infty$ with mean and covariance given by

\begin{align}\label{mediat}
&\mu_{t}(z)=\overline{K}(z,Z^{T})\overline{K}^{-1}\left(\mathbbm{1}-e^{\eta_{0}\overline{K}t}\right)Y,\\\label{variancet}
&\begin{aligned}
\K_{t}(z,z')&=\K_{0}(z,z')-\overline{K}(z,Z^{T})\overline{K}^{-1}\left(\mathbbm{1}-e^{-\eta_{0}\overline{K}t}\right)\K_{0}(Z,z')\\
&-\overline{K}(z',Z^{T})\overline{K}^{-1}\left(\mathbbm{1}-e^{-\eta_{0}\overline{K}t}\right)\K_{0}(Z,z)+\\
&+\overline{K}(z,Z^{T})\overline{K}^{-1}\left(\mathbbm{1}-e^{-\eta_{0}\overline{K}t}\right)\K_{0}(Z,Z^{T})\left(\mathbbm{1}-e^{-\eta_{0}\overline{K}t}\right)\overline{K}^{-1}\K_{0}(Z,z').
\end{aligned}
\end{align}
\end{lem}
\begin{proof}[{Proof of \autoref{lem_gauss_lim}}]
Notice that the right-hand side of \eqref{limitp} is the linear combination of Gaussian process, then we obtain a Gaussian process as well, and we denote it $\U_{t}^{(\infty)}(z)$. Notice that,

\begin{align}
\begin{aligned}
\mu_{t}(z)&=\E[\U_{t}^{(\infty)}(z)]\\
&=\E\left[\U^{(\infty)}(z)-\overline{K}(z,Z^{T})\overline{K}^{-1}(1-e^{-\eta_{0}\overline{K}t})\left(\U^{(\infty)}(Z)-Y\right)\right]\\
&=\overline{K}(z,Z^{T})\overline{K}^{-1}(1-e^{-\eta_{0}\overline{K}t})Y
\end{aligned}
\end{align}
where we have used that $\E\left[\U^{(\infty)}(z)\right]$ is the null vector. Let us now compute $\K_{t}(z,z')$. Notice that

\begin{align}
\begin{aligned}
&\K_{t}(z,z')={\rm cov}\left[\U_{t}^{(\infty)}(z)-\mu_{t}(z);\U_{t}^{(\infty)}(z')-\mu_{t}(z')\right]\\
&={\rm cov}\left[\U^{(\infty)}(z)-\overline{K}(z,Z^{T})\overline{K}^{-1}(1-e^{-\eta_{0}\overline{K}t})\U^{(\infty)}(Z); \right.\\
&\phantom{fomula}\left.\U^{(\infty)}(z')-\overline{K}(z',Z^{T})\overline{K}^{-1}(1-e^{-\eta_{0}\overline{K}t})\U^{(\infty)}(Z)\right]\\
&={\rm cov}\left[\U^{(\infty)}(z);\U^{(\infty)}(z')\right]\\
&-\overline{K}(z,Z^{T})\overline{K}^{-1}(1-e^{-\eta_{0}\overline{K}t}){\rm cov}\left[\U^{(\infty)}(Z);\U^{(\infty)}(z')\right]\\
&-\overline{K}(z',Z^{T})\overline{K}^{-1}(1-e^{-\eta_{0}\overline{K}t}){\rm cov}\left[\U^{(\infty)}(Z);\U^{(\infty)}(z)\right]\\
&+\overline{K}(z,Z^{T})\overline{K}^{-1}(1-e^{-\eta_{0}\overline{K}t}){\rm cov}\left[\U^{(\infty)}(Z);\U^{(\infty)}(Z^{T})\right]\times\\
&\phantom{formula}\times (1-e^{-\eta_{0}\overline{K}t})\overline{K}^{-1}\overline{K}(Z^{T},z')
\end{aligned}
\end{align}
Since $\K_{0}(z,z')={\rm cov}\left[\U^{(\infty)}(z);\U^{(\infty)}(z')\right]$, we are done.
\end{proof}
\begin{thm}
Let us assume \autoref{A0}--\autoref{A6}, and \autoref{A9}--\autoref{A10}. Furthermore, let us suppose that 
\begin{align}\label{obt_from_linearized}
\lim_{m\rightarrow +\infty}\frac{L^{12}m^{2}\vert\mathcal{M}\vert^{5}\vert\mathcal{N}\vert^{3}}{(b_{K}(m))^{2}(b(m))^{5}}(1+\log b(m))=0.    
\end{align}
Then $\left\{ \U(\Theta_{t},z)\right\}_{z\in B\times \partial B}$ converges in distribution to the Gaussian process $\left\{\U_{t}^{(\infty)}(z)\right\}_{z\in B\times \partial B}$ as $m\rightarrow+\infty$ with mean and covariance given by \autoref{mediat}, and \autoref{variancet}.
\end{thm}
\begin{proof}
Let us first prove that \autoref{obt_from_linearized} implies that assumptions \autoref{convergencehypothesis}, and \autoref{ipoconvergencentk} hold true. Let us start with \eqref{ipoconvergencentk}. We need to prove that 

\begin{align}
   \lim_{m\rightarrow +\infty}\frac{m L^{9}\vert\mathcal{M} \vert^{4}\vert\mathcal{N}\vert^{2}}{(b(m))^{4}}=0. 
\end{align}
Since $b(m)\leq \sqrt{5}d^{2}L^{2}\sqrt{m\vert \mathcal{M}\vert\vert\mathcal{N}\vert}$ one has
\begin{align}
\begin{aligned}
0\leq \frac{m L^{9}\vert\mathcal{M} \vert^{4}\vert\mathcal{N}\vert^{2}}{(b(m))^{4}}&\leq \frac{\sqrt{5} d^{2} m^{\frac{3}{2}} L^{11}\vert\mathcal{M} \vert^{\frac{9}{2}}\vert\mathcal{N}\vert^{\frac{5}{2}}}{(b(m))^{5}}\\
&\leq \frac{\sqrt{5} d^{2} m^{2} L^{12}\vert\mathcal{M} \vert^{5}\vert\mathcal{N}\vert^{3}}{(b(m))^{5}}
\end{aligned}
\end{align}
On the other hand, since $1\leq b_{K}(m)$ and $d$ is fixed, by applying \autoref{obt_from_linearized} our conclusion follows. Let us now prove \autoref{convergencehypothesis}. We need to show that

\begin{align}
 \lim_{m\rightarrow +\infty}\frac{m L^{6}\vert\mathcal{M} \vert^{2}\vert\mathcal{N}\vert^{2}}{(b(m))^{3}}=0.    
\end{align}
By following the previous reasoning and using the bound for $b(m)$, we conclude that

\begin{align}
\begin{aligned}
0\leq \frac{m L^{6}\vert\mathcal{M} \vert^{2}\vert\mathcal{N}\vert^{2}}{(b(m))^{3}}&\leq \frac{5d^{4}m^{2} L^{10}\vert\mathcal{M} \vert^{3}\vert\mathcal{N}\vert^{3}}{(b(m))^{5}}\\
&\leq \frac{5d^{4}m^{2} L^{12}\vert\mathcal{M} \vert^{5}\vert\mathcal{N}\vert^{3}}{(b(m))^{5}}.
\end{aligned}
\end{align}
Thus by our hypothesis \eqref{obt_from_linearized}, we are done. On the other hand, notice that by \autoref{thm:lazytraining} one has

\begin{align}
\begin{aligned}
\sup_{\substack{z\in B\times \partial B \\t\geq 0}}\norma{\U(\Theta_{t},z)-\U^{\lin}(\Theta_{t}^{\lin},z)}_{2}&\leq \frac{(5)2^{9}(n_{1}+n_{2})^{2}(R(\delta))^{2}d^{10}}{\widetilde{\lambda}_{\min}(\delta)}\Bigg[1+\frac{4}{\widetilde{\lambda}_{\min}(\delta)}+\\
&+\frac{256}{\eta_{0}(\widetilde{\lambda}_{\min}(\delta))^{2}b_{K}(m)}\Bigg]\frac{L^{12}m^{2}\vert\mathcal{M}\vert^{5}\vert\mathcal{N}\vert^{2}}{(b_{K}(m))^{2}\,(b(m))^{5}}(1+\log b(m))\\
&\leq \frac{(5)2^{9}(n_{1}+n_{2})^{2}(R(\delta))^{2}d^{10}}{\widetilde{\lambda}_{\min}(\delta)}\Bigg[1+\frac{4}{\widetilde{\lambda}_{\min}(\delta)}+\\
&+\frac{256}{\eta_{0}(\widetilde{\lambda}_{\min}(\delta))^{2}b_{K}(m)}\Bigg]\frac{L^{12}m^{2}\vert\mathcal{M}\vert^{5}\vert\mathcal{N}\vert^{3}}{(b_{K}(m))^{2}\,(b(m))^{5}}(1+\log b(m)).
\end{aligned}
\end{align}
From here and \autoref{obt_from_linearized} we have that $\U(\Theta_{t},z)-\U^{\lin}(\Theta,z)\stackrel{p}{\rightarrow}0$ as $m\longrightarrow +\infty$. Then together with \autoref{lem_gauss_lim}, we conclude that 
\begin{align}
\U(\Theta_{t},z)\stackrel{d}{\longrightarrow}\U_{t}^{(\infty)}(z)   
\end{align}
as $m\rightarrow +\infty$.
\end{proof}

\section{Asymptotic behavior of QNNs}\label{sec:towgauss}
In this Appendix, we briefly summarize some of the key results established in \cite{girardi2025,melchor2025quantitative}. There, the authors consider a sequence of quantum neural networks with diverging width satisfying \autoref{A1}-\autoref{A4} and the following further assumption: 
\begin{ass}\label{AP1}
Suppose that $b(m)$ grows sufficiently fast such that
\begin{align}
\lim_{m\rightarrow +\infty}\frac{m \vert \mathcal{M}\vert^{2}\vert\mathcal{N}\vert^{2}}{(b(m))^{3}}=0.
\end{align}  
\end{ass}
The following is one of the main achievements of \cite{girardi2025,melchor2025quantitative}.
\begin{thm}[See {\cite[Theorem 3.7]{girardi2025}, \cite[Theorem B.1]{melchor2025quantitative}}]\label{convergencethm}
Suppose that \autoref{A1}-\autoref{A4}, and \autoref{AP1} hold true. Then as $m\rightarrow +\infty$, the probability distribution of the generated function at initialization converges in distribution to the Gaussian process with mean zero and covariance $\K$.
\end{thm}
We notice that \autoref{convergencethm} does not provide any convergence rate.
\begin{lem}[Lipschitzness of the gradient {\cite[Lemma 4.20]{girardi2025}}]
The following inequalities hold:
\label{lemma4.20} 
\begin{align}
\label{lemma1}|\partial_{\theta_i}f(\Theta,x)|&\leq 2\frac{|\mathcal{M}_i|}{b(m)},\\
\label{lemma2}\|\partial_{\theta_i}f(\Theta,X)\|_2&\leq 2\sqrt n \,\frac{|\mathcal{M}|}{b(m)},\\
\label{lemma3}\|\nabla_{\Theta}f(\Theta,x)\|_1&\leq 2L\frac{m}{b(m)}|\mathcal{M}|, \\
\label{lemma4}\|\nabla_\Theta f(\Theta,x)-\nabla_{\Theta} f(\Theta',x)\|_\infty&\leq 4\,\frac{|\mathcal{M}|^2|\mathcal{N}|}{b(m)}\|\Theta-\Theta'\|_\infty.
\end{align}
\end{lem}

\bibliographystyle{unsrt}
\bibliography{bibliography_1}

\end{document}